\documentclass{article}

\usepackage[utf8]{inputenc}

\usepackage{amsmath, amsthm, amssymb, graphicx}
\usepackage{hyperref}
\usepackage{multicol}
\usepackage[top=2cm, bottom=2cm,left=1.5cm, right=1.5cm]{geometry}
\usepackage{cancel}
\usepackage{MnSymbol}
\usepackage{turnstile}
\usepackage{tikz}
\usepackage{mathtools} 
\usepackage{latexsym}
\usepackage{enumitem} 
\usepackage{caption}
\usepackage{varwidth}
\usepackage{colortbl} 
\usepackage{color}
\usepackage{setspace} 
\usepackage{makeidx} 
\usepackage{array,longtable} 
\usepackage{multicol}
\usepackage{mathpartir}
\usepackage[scr = pxtx]{mathalpha}
\usepackage{float}
\usepackage{bold-extra}
\usepackage{upgreek}
\usepackage{tasks}

\usepackage{verbatim}

\DeclareFontShape{OT1}{cmr}{m}{scit}{<->ssub * cmr/m/sc}{}

\usetikzlibrary{arrows,decorations.pathmorphing,backgrounds,positioning,fit,shapes,shapes.misc}
\usetikzlibrary{decorations.pathmorphing,positioning,decorations.pathreplacing,patterns}
\usetikzlibrary{patterns.meta}

\newcommand{\eqdef}{\overset{\text{def}}{=\joinrel=}}

\newcommand{\pmu}{p~\kern -0.35em_\mu}
\newcommand{\pmutilde}{\widetilde{p}~\kern -0.35em_\mu}
\newcommand{\pmuhat}{\widehat{p}~\kern -0.35em_\mu}

\newcommand{\traunotres}{  {T\kern -0.15emR}_1^3}
\newcommand{\tratresuno}{  {T\kern -0.15emR}_3^1}
\newcommand{\traunoseis}{  {T\kern -0.15emR}_1^6}
\newcommand{\traseisuno}{  {T\kern -0.15emR}_6^1}
\newcommand{\tratresseis}{ {T\kern -0.15emR}_3^6}
\newcommand{\traseistres}{ {T\kern -0.15emR}_6^3}

\newcommand{\forsigma}{ {F\kern -0.2em or}\pa{\Sigma}}
\newcommand{\DNF}{ {D\kern -0.15em N\kern -0.15em F}}
\newcommand{\fde}{{F\kern -0.1em D\kern -0.1em E}}
\newcommand{\CPL}{{C\kern -0.1em P\kern -0.1em L}}
\newcommand{\letkp}{\ensuremath{{L\kern -0.1em ET}_K^+}}
\newcommand{\LETK}{{L\kern -0.1em ET}_K}
\newcommand{\LETF}{{L\kern -0.1em ET}_F}

\newcommand{\Dashv}{%
  \mathrel{\text{\reflectbox{$\vDash$}}}%
}
\newtheorem{definition}{Definition}
\newtheorem{obs}{Observation}
\newtheorem{proposition}{Proposition}

\newtheorem{lemma}{Lemma}

\newtheorem{corolario}{Corollary} 
\newtheorem{example}{Example}

\newcommand{\D}[1]{D_{#1}}
\newcommand{\Mod}[1]{M\kern -0.15em od\pa{#1}}
\newcommand{\val}[1]{\overline{v}\pa{#1}}
\newcommand{\tland}{\mathbin{\tilde{\land}}}
\newcommand{\tlor}{\mathbin{\tilde{\lor}}}
\newcommand{\tto}{\mathbin{\tilde{\rightarrow}}}
\newcommand{\tneg}{\tilde{\neg}}
\newcommand{\tcirc}{\tilde{\circ}}

\newcommand{\pa}[1]{\mathopen{}\left( {#1}_{{}_{}}\,\negthickspace\right)\mathclose{}}
\newcommand{\lla}[1]{\mathopen{}\left\{ {#1}_{{}_{}}\,\negthickspace\right\}\mathclose{}}
\newcommand{\cor}[1]{\mathopen{}\left[ {#1}_{{}_{}}\,\negthickspace\right]\mathclose{}}

\title{Paraconsistent and paracomplete probabilities:\\ Measuring reliable and unreliable beliefs
\\
{\Large (Extended version)}}
\author{Verónica Borja Macías$^1$, Marcelo E. Coniglio$^2$  and Alejandro Hernández-Tello$^3$}
\date{
    $^\textup{\scriptsize 1}$\textit{\small Centro de Modelación Matemática Vinculación y Consultoría (CEMMVyC),}
    \\
     \textit{\small Universidad Tecnológica de la Mixteca (UTM), Huajuapan de León, Oaxaca, Mexico}
   \\
    $^\textup{\scriptsize 2}$\textit{\small Centre for Logic, Epistemology and the History of Science (CLE), and}
    \\
    \textit{\small Institute of Philosophy and the Humanities (IFCH), University of Campinas (UNICAMP), Campinas, Brazil}
        \\
    $^\textup{\scriptsize 3}$\textit{\small Instituto de Física y Matemáticas (IFM),}
    \\
     \textit{\small Universidad Tecnológica de la Mixteca (UTM), Huajuapan de León, Oaxaca, Mexico}
   \\
   {\small e-mails: \texttt{vero0304@gmail.com}; \ \  \texttt{coniglio@unicamp.br}; \ \  \texttt{alheran@gmail.com}}
   }

\begin{document}

\maketitle 

\begin{abstract}
In this paper, we introduce the study of probability functions based on the 6-valued paradefinite (i.e., paraconsistent and paracomplete) logic \letkp. This logic is a powerful and versatile Logic of Evidence and Truth (LET) which is a conservative expansion of both classical logic and \fde. The framework introduced here allowed us to consider gaps, gluts, and reliability (or classicality) of the events, extending the detailed proposal for \fde-based probabilities presented by Klein, Majer, and Rafiee~Rad. A distinctive feature of our proposal is the use of twist structure semantics, which gives rise to a natural interpretation of logical probabilities over \letkp\ in terms of the three or six regions associated with each formula by a valuation in such models. The \letkp-probability functions are defined axiomatically and semantically, obtaining soundness and completeness results, as one would expect. Finally, conditional probabilities based on $\letkp$ are also studied. Specifically, both a semantic and a syntactic characterization of Jeffrey's update over \letkp-based probabilities is proposed, showing their equivalence.
\end{abstract}



\section{Introduction}

The connection between logic and probability can be traced back to Keynes~\cite{keynes} and Jeffreys~\cite{jeffreys}, and, in a more rigorous form, to Carnap~\cite{carnap:50, carnap:52} --- specifically through his so-called {\em inductive logic}.  However, according to Hacking~\cite[Chapter~15]{hacking}, it can be argued that this philosophical programme was already anticipated by Leibniz.\footnote{The reader interested in the relationship between logic and probability from the historical perspective can consult, for instance, \cite[Chapter~2]{good} and~\cite[Section~8.1]{carnielli2017paraconsistent}.}

In Carnap's approach, probabilities are assigned to (propositional) sentences, instead of sets (of events) in a sample space. There is an obvious relationship between Carnap's logical probability and Kolmogorov's mathematical probability. On the one hand, Carnap assumes classical logic underlying his inductive logic. On the other hand, Kolmogorov's $\sigma$-algebras of events (in which probability functions are defined) are special cases of Boolean algebras of sets (i.e., fields of sets), the ones closed under denumerable unions. Clearly, $\sigma$-algebras constitute a sound and complete semantics for propositional classical logic (\CPL),\footnote{This follows from the following facts: (i)~every $\sigma$-algebra is a Boolean algebra; and  (ii)~the two-element Boolean algebra {\bf 2} is a $\sigma$-algebra.} in which conjunctions, disjunctions and negations are modeled by intersections, unions and complements, respectively. In this way, events in a $\sigma$-algebra correspond (or are described by) sentences in the language of propositional classical logic. 
It seems natural to expect that, by generalizing Carnap's idea, probability functions can be defined over sentences in the formal language of logics other than the classical one.

Indeed, Gaifman~\cite{gaifman} introduced probability functions over first-order classical logic. In turn, in his seminal paper~\cite{weatherson2003probability}, Weatherson introduced a notion of probability for logics having at least standard conjunction and disjunction connectives, by suitably adapting Carnap's postulates. This allows to keep the requirement encoded in the additivity axiom, namely $p(\varphi \vee \psi)=p(\varphi)+ p(\psi) -p(\varphi \wedge \psi)$, which plays the role of the classical Inclusion-Exclusion Principle. In particular, he considered probabilities over intuitionistic logic satisfying Carnap's postulates in a suitable form. The property $p(\varphi \wedge \neg\varphi)=0$ of classical Carnapian probabilities  is still valid in the intuitionistic framework, since the intuitionistic negation $\neg$, as the classical one, is explosive. However, and as one could expect, the property $p(\neg\varphi) = 1 - p(\varphi)$ does not hold in general, since the formula $(\varphi \vee \neg\varphi)$ is not valid in general in intuitionistic logic; thus, $p(\varphi \vee \neg\varphi)$ is not necessarily $1$ in this context. 

Weatherson justifies the use of an intuitionistic theory of probability based, among other arguments, on the following: if, as several authors argue, (logical) probabilities are degrees of belief (seen as dispositions to bet), then probability functions over intuitionistic logic provide a more flexible framework than classical ones. Indeed, in some situations one has little or no evidence for or against an event represented by a sentence (say, $\varphi$). In such cases, it should be reasonable to have low degrees of belief in both $\varphi$ and $\neg\varphi$. Recall that a logic with disjunction and negation is said to be {\em paracomplete} if the schema $\varphi \vee \neg\varphi$ is not valid. Clearly, intuitionistic logic is paracomplete.


Of course, a dual (i.e., {\em paraconsistent}) perspective could be considered, in which a contradiction is not necessarily an antithesis. That is, assuming a standard conjunction, one can assign a positive non-zero degree of belief to a contradiction, say $(\psi \land \neg \psi)$. Bueno-Soler and Carnielli investigated this possibility (namely, paraconsistent Carnapian-style probabilities) in~\cite{bueno2016paraconsistent, carnielli2017paraconsistent} specifically for paraconsistent logics known as {\em Logics of Formal Inconsistency} (LFIs).

In order to handle both kinds of events (contradictory and incomplete) simultaneously, it is necessary to adopt a {\em paradefinite} logic (i.e., paraconsistent and paracomplete) as the underlying logic system. The most natural candidate for a paradefinite logic is the well-known {\em logic of first-degree entailment} (\fde), also known as Belnap-Dunn four-valued logic.
\fde\ was designed to model the logic underlying a computer which manages information coming from (possibly) different sources that are not entirely reliable. The pieces of information are represented by sentences in a propositional language only containing conjunction $\land$, disjunction $\vee$ and negation $\neg$. Because of this, a given sentence, say $\varphi$, can be considered or marked as `true' from one source and `false' from another one (which, in this context, is equivalent to the fact that $\neg\varphi$ is marked or considered as `true'). Such a contradictory sentence $\varphi$ will receive the truth-value $\mathsf{b}$ (`Both'), which indicates precisely a contradiction. In case $\varphi$ is marked as `true' in all the sources, it receives the truth-value $\mathsf{t}$ (`Only True'). Dually, if $\varphi$ is marked as `false' in all the sources, it receives the truth-value $\mathsf{f}$ (`Only False'). Finally, if no information is provided for $\varphi$ (and so for $\neg\varphi$) by the sources, then  it receives the truth-value $\mathsf{n}$ (`None'). A contradiction (i.e., a sentence receiving the value $\mathsf{b}$) indicates conflicting (or excessive) information, which is known as a `glut'. Dually, a sentence receiving the value $\mathsf{n}$ indicates lack of information, which is known as a `gap'.

Mares~\cite{mares1997paraconsistent} was the first to consider  probabilities over a variant of \fde. Dunn introduced in~\cite{dunn2010contradictory} a formal framework to deal with  four-dimensional probabilities based on \fde. That   is, the probability functions assign to every sentence of \fde\ a value in $[0,1]^4$.

In~\cite{klein2021probabilities} Klein et al. also proposed four-dimensional probabilities based on \fde, but in this case they correspond  faithfully to the measure of the four regions (pure belief, pure disbelief, conflict and uncertainty)  determined by each sentence in that logic. They also introduced $[0,1]$-valued probability functions for \fde\ formulas that are induced by a standard measure, showing that they are in a one-to one correspondence with the four-dimensional ones via a translation function.
In addition, they introduced sound and complete axiomatizations for both kinds of probabilities -- the one-dimensional ones being axiomatized by Carnapian-style probability functions based on \fde. This constitutes an improvement over the proposal in~\cite{dunn2010contradictory}, which relied on the strong and problematic assumption that events are mutually probabilistically independent. By avoiding such an assumption, Klein et al. provide a framework in which conditional probabilities can be meaningfully defined. To this aim, the authors develop a theory of Jeffrey's conditioning update based on the two (equivalent) notions of \fde-probabilities they proposed. The Jeffrey conditioning is defined both syntactically and semantically, showing their equivalence. As a particular case, they define Bayesian updating for both kinds of probabilities over \fde. It is worth mentioning that  Bilkova et al.~\cite{bilkova2024, bilkova2025two}  present interesting further developments in probability theory based on \fde\ beyond those of~\cite{klein2021probabilities}.

In 2015, Carnielli and Rodrigues introduced the \emph{Logics of Evidence and Truth} (LETs) by presenting the system $\LETK$ in~\cite[Section~3.3.2]{carnielli2015}. LETs constitute an expansion of \fde\ by adding a unary connective $\circ$ (a recovery, classicality, or reliability operator) to the language. This operator allows the recovery of the classical behavior of the De Morgan negation $\neg$ in a local way. In the case of $\LETK$ and other LETs, an implication can also be added to \fde, in addition to $\circ$. The recovery of the explosion law by means of a unary operator is exactly the idea of LFIs which, in turn, are a generalization of da Costa's well-known $C$-systems. The difference between LFIs and LETs, at least in the most studied systems in each hierarchy, is that the former  satisfy the excluded middle for the paraconsistent negation. Then, the LFIs studied in the literature are only paraconsistent, and not paracomplete. In turn, LETs are paradefinite and feature a De Morgan negation, i.e., an involutive negation that satisfies the usual De Morgan rules. Thus, in general

\begin{itemize}
    \item $ \varphi, \neg \varphi \nvdash \psi$  \  \  \ [L is paraconsistent]
    \item $\nvdash \varphi \lor \neg \varphi$  \  \  \ [L is paracomplete]
  \end{itemize}

\noindent in a given LET L. On the other hand, the unary operator $\circ$ recovers in any LET the properties of explosion and excluded middle for sentences under its scope, by validating the following inferences:

\begin{itemize}
    \item $\circ \varphi, \varphi, \neg \varphi \vdash \psi$
    \item $\circ \varphi \vdash \varphi \lor \neg \varphi$.
  \end{itemize}

The meaning of a sentence $\circ\varphi$ is that the information conveyed by $\varphi$ --- whether positive or negative --- is reliable. In terms of evidence, asserting $\varphi$ is not saying that `$\varphi$ is true', but rather `there is positive evidence for $\varphi$'. Analogously, asserting $\neg\varphi$ is not saying that `$\varphi$ is false', but `there is negative evidence for $\varphi$'. In this sense, $\circ\varphi$ means that there is conclusive evidence for $\varphi$ (or for $\neg\varphi$). Thus, LETs enrich the perspective of \fde\ by allowing us to qualify, via the $\circ$ operator, the reliability of the information that the computer receives from its sources. 

The term `reliability' might be misleading, so an explanation is in order at this point. Instead of considering a computer managing a knowledge base of information items, let us consider the more general situation of a rational agent dealing with beliefs, as understood in the theory of belief change --- more specifically, the AGM model (see, for instance,~\cite{Gardenfors1988, hans:99}). The knowledge of the agent can be modeled by the total amount of accepted facts (formally represented by propositions in a formal language). New situations may lead the agent to revise some of its beliefs, giving rise to a process of revision, contraction, or expansion of the agent's epistemic state. Although most of the developments in AGM theory have been carried out within classical logic, Testa et al.~\cite{testa_etal, coniglio_etal} have proposed a robust AGM belief revision theory for paraconsistent logics in general, and for LFIs in particular. 

Within the paraconsistent belief framework AGM$\circ$ introduced in~\cite{testa_etal}, it is possible to move from accepting only $\varphi$ to also accepting $\neg\varphi$ (an expansion); or to move from a knowledge state $K$ that accepts $\varphi$ to another state in which $\varphi$ is no longer accepted (a contraction). However, the contraction (or revision) of $\varphi$ in $K$ is allowed only if $\circ\varphi$ does not belong to $K$. Otherwise, $\varphi$ (or $\neg\varphi$) cannot be removed, and if added, should this cause a contradiction, the resulting belief state  is the trivial one.

Taking this into consideration, and still using AGM as an analogy, in a paradefinite environment such as \fde, where the operator $\circ$ is not available, the standard epistemic attitudes of accepting or rejecting a sentence are revisable --- that is, they are flexible. In terms of truth-values, agents could move, for instance, from $\mathsf{t}$ to $\mathsf{f}$, or to $\mathsf{b}$ or $\mathsf{n}$, via revision, contraction, or expansion of their knowledge states. However, by incorporating the reliability operator $\circ$ into the picture (using LETs), it is possible to describe the agent's disposition to change its mind with respect to each piece of knowledge. Thus, accepting $\circ\varphi$ means that the agent is unwilling to revise its opinion regarding that specific sentence, even though it may have no settled view as to its truth or falsity. A paradigmatic example of this situation is a mathematical conjecture. One may ignore whether the result holds or not, but one certainly knows that the result is either true or false, and that it cannot be both simultaneously. Furthermore, once proved (or refuted), this situation will persist forever.

To summarize, $\circ\varphi$ may be understood as encoding (some of) the agent's reasons for not revising its opinion regarding that sentence. This reflects an {\em inertia of information} captured by $\circ$ in  LETs. Accordingly, `reliability of information' can be interpreted here as `unwillingness to change one's opinion'.

In this paper we will focus on the logic $\letkp$, a very strong LET introduced  by Coniglio and Rodrigues in~\cite{coniglio2024belnap}. This logic is an extension of $\LETK$ by adding rules for propagating classicality/reliability. Because of this, this $\letkp$ expands simultaneously \fde\ and \CPL, having a very strong expressive power which will be useful to our purposes. $LET_K^+$ has a six-valued semantics captured by a logic matrix based on the lattice  $\mathbf{L6}$, known as the crystal lattice (see Section~\ref{sect:LETK+}). 

The idea of this paper is to adapt the detailed study of probabilities and conditional probabilities over \fde, presented in~\cite{klein2021probabilities}, to $\letkp$, obtaining similar results in a richer logical context. Indeed, besides gap and gluts in the space of events (which is the distinctive feature of \fde), the logic $\letkp$ allows us to assess the reliability or classical behavior of the information conveyed by events in the sense discussed above. In analogy with~\cite{klein2021probabilities}, where equivalent one-dimensional and four-dimensional probability functions were defined for the four-valued logic \fde, we introduce analogous one-dimensional and six-dimensional versions for the six-valued logic $\letkp$. As expected, these two formulations turn out to be equivalent. A distinctive feature of our proposal is the introduction of three-dimensional probability functions for $\letkp$, which are also equivalent to the former ones. These functions are interpreted over the natural three-dimensional twist structures semantics for $\letkp$ given in~\cite{coniglio2024belnap}; thus, a given standard probability function measures the three areas of each event described above. These algebras of events can be seen as a kind of twist $\sigma$-algebras over $\letkp$, as we will discuss below.\footnote{Of course, the same kind of twist probabilistic models could be considered for \fde\ in the framework of~\cite{klein2021probabilities}. In that case, the standard two-dimensional twist structures for \fde\ defined over the $\sigma$-algebra $\wp(X)$ (for a finite set of states $X$) would correspond to $\sigma$-algebras based on that logic, and would give rise to a two-dimensional class of probability functions for it, equivalent to the other two. The details of such constructions are straightforward, as they follow by adapting the more general construction for $\letkp$ presented here. It suffices to observe that twist models for \fde\ are special cases of twist models for $\letkp$ in the reduct $\{\land,\vee,\neg\}$, obtained by taking valuations whose third coordinate is always the empty set $\emptyset$.}

The present paper combines, to a certain degree, the approaches of~\cite{rodrigues2021measuring} in the formal framework of~\cite{klein2021probabilities}, but based on a logic that, in our view, is particularly suitable for this purpose: $\letkp$. Its advantage is that, as a conservative expansion of both \CPL\ and \fde, it offers the best of both worlds: expressive power and an extremely versatile semantics, which allows us to approximate Carnap's logical probabilities to Kolmogorov's mathematical probabilities in a paradefinite setting.

The organization of this paper is as follows. Section~\ref{sect:LETK+} briefly discusses the main technical features of the logic $\letkp$. Section~\ref{sect:twist-models} introduces the twist models for $\letkp$. Section~\ref{sect:prob-functions} defines the three kinds of probabilities over $\letkp$ --- namely, the one-, three-, and six-dimensional ones --- both syntactically and semantically, and proves that they are all equivalent in Section~\ref{sec:completeness}. Section~\ref{sect:condit-prob} addresses conditional probabilities based on $\letkp$. More precisely, we provide both a semantic and a syntactic characterization of Jeffrey's update and prove their equivalence. We also consider the special case of Bayesian updating. Finally, Section~\ref{sect:final} offers a summary of the paper's main contributions and outlines some possibilities for future work.

\section{The logic  $\letkp$} \label{sect:LETK+}

Let $V=\{\mathsf{p}_1, \ldots, \mathsf{p}_m\}$ be a finite set of propositional variables, and $\Sigma = \{\land, \lor, \rightarrow, \neg, \circ\}$ a propositional signature.\footnote{As usual, $\land$, $\lor$, and $\rightarrow$ are binary connectives and denote conjunction, disjunction, and implication, respectively. The connectives $\neg$ and $\circ$ are unary and denote negation and consistency (or reliability, or classicality), respectively.} By $\forsigma$ we denote the algebra of formulas freely generated by $V$ over $\Sigma$.

In this section, we recall the natural deduction system for $\letkp$ presented in \cite{coniglio2024belnap}, as well as several sound and complete semantics for this logic.

\begin{definition}\label{def:ded_nat-LETk+}
\textbf{(Natural deduction system for $\letkp$)}
Let $\varphi,\psi,\gamma \in \forsigma$. Consider the following abbreviations: $\varphi^{T}\eqdef \varphi\land \circ\varphi$ and $\varphi^{F}\eqdef \neg\varphi\land \circ\varphi$. A natural deduction system over $\Sigma$ for the logic $\letkp$ is given by the following set of rules:

\begin{mathpar}
\inferrule*[Right=$\land I$]
{ \varphi \\  \psi}
{ \varphi \land \psi}

\inferrule*[Right=$\land E_l$]
{ \varphi \land \psi}
{ \varphi}

\inferrule*[Right=$\land E_r$]
{ \varphi \land \psi}
{ \psi}
\end{mathpar}
\begin{mathpar}
\inferrule*[Right=$\lor I_l$]
{ \varphi}
{ \varphi \lor \psi}

\inferrule*[Right=$\lor I_r$]
{ \psi}
{ \varphi \lor \psi}

\inferrule*[Right=$\lor E$]
{ \varphi \lor \psi \\ \inferrule*{[\varphi]\\[\psi] \\\\ \vdots \\ \hspace*{0.5cm}\vdots \\\\\gamma\\\hspace*{0.3cm}\gamma}{}}
{ \gamma}
\end{mathpar}
\begin{mathpar}
\inferrule*[Right=${\neg}{\land} I_l$]
{ \neg \varphi}
{ \neg(\varphi \land \psi)}

\inferrule*[Right=${\neg}{\land} I_r$]
{ \neg \psi}
{ \neg(\varphi \land \psi)}

\inferrule*[Right=${\neg}{\land} E$]
{\neg(\varphi \land \psi) \\ \inferrule*{[\neg \varphi]\\[\neg \psi] \\\\ \vdots \\ \hspace*{0.5cm}\vdots \\\\\gamma\\\hspace*{0.3cm}\gamma}{}}
{ \gamma}
\end{mathpar}
\begin{mathpar}
\inferrule*[Right=${\neg}{\lor} I$]
{ \neg \varphi \\  \neg \psi}
{\neg( \varphi \lor \psi)}

\inferrule*[Right=${\neg}{\lor} E_l$]
{\neg( \varphi \lor \psi)}
{\neg \varphi}

\inferrule*[Right=${\neg}{\lor} E_r$]
{\neg( \varphi \lor \psi)}
{\neg \psi}
\end{mathpar}
\begin{mathpar}
\inferrule*[Right=${\rightarrow}I$]
{\inferrule*{[\varphi] \\\\ \vdots\\\\\psi}{}}
{ \varphi \rightarrow \psi}

\inferrule*[Right=${\rightarrow} E$]
{ \varphi \rightarrow \psi \\  \varphi}
{ \psi}

\inferrule*[Right=${\rightarrow} CL$]
{~}
{ \varphi\lor (\varphi \rightarrow \psi)}
\end{mathpar}
\begin{mathpar}
\inferrule*[Right=${\neg}{\rightarrow} I$]
{ \varphi \\  \neg \psi}
{\neg( \varphi \rightarrow \psi)}

\inferrule*[Right=${\neg}{\rightarrow} E_l$]
{\neg( \varphi \rightarrow \psi)}
{ \varphi}

\inferrule*[Right=${\neg}{\rightarrow} E_r$]
{\neg( \varphi \rightarrow \psi)}
{\neg \psi}
\end{mathpar}
\begin{mathpar}
\inferrule*[Right=$DNI$]
{\varphi}
{ \neg\neg \varphi}

\inferrule*[Right=$DNE$]
{ \neg\neg \varphi}
{ \varphi}

\inferrule*[Right=$EXP^\circ$]
{\circ \varphi\\\varphi\\\neg \varphi}
{\psi}

\inferrule*[Right=$PEM^\circ$]
{\circ \varphi}
{\varphi\lor\neg \varphi}
\end{mathpar}
\begin{mathpar}
\inferrule*[Right=$\circ I$]
{~}
{{\circ\circ} \varphi}

\inferrule*[Right=$\circ\neg I$]
{\circ \varphi}
{{\circ\neg} \varphi}

\inferrule*[Right=$\circ\neg E$]
{{\circ\neg} \varphi}
{\circ \varphi}
\end{mathpar}
\begin{mathpar}
\inferrule*[Right=${T\land} I$]
{ \varphi^T \\   \psi^T}
{( \varphi \land \psi)^T}

\inferrule*[Right=${F\land} I_l$]
{\varphi^F}
{( \varphi \land \psi)^F}

\inferrule*[Right=${F\land} I_r$]
{\psi^F}
{( \varphi \land \psi)^F}
\end{mathpar}
\begin{mathpar}
\inferrule*[Right=${T\lor} I_l$]
{\varphi^T}
{( \varphi \lor \psi)^T}

\inferrule*[Right=${T\lor} I_r$]
{\psi^T}
{( \varphi \lor \psi)^T}

\inferrule*[Right=${F\lor} I$]
{ \varphi^F \\   \psi^F}
{( \varphi \lor \psi)^F}

\end{mathpar}
\begin{mathpar}
\inferrule*[Right=${T}{\rightarrow} I_l$]
{\varphi^F}
{( \varphi \rightarrow \psi)^T}

\inferrule*[Right=${T}{\rightarrow} I_r$]
{\psi^T}
{( \varphi \rightarrow \psi)^T}

\inferrule*[Right=${F}{\rightarrow} I$]
{ \varphi \\   \psi^F}
{( \varphi \rightarrow \psi)^F}

\end{mathpar}
\begin{mathpar}
\inferrule*[Right=${T}{\land} E_l$]
{ ( \varphi \land \psi)^T}
{ \varphi^T}

\inferrule*[Right=${T}{\land} E_r$]
{ ( \varphi \land \psi)^T}
{ \psi^T}

\inferrule*[Right=${F}{\land} E$]
{ ( \varphi \land \psi)^F \\ \inferrule*{[\varphi^F]\\[\psi^F] \\\\ \vdots \\ \hspace*{0.5cm}\vdots \\\\\gamma\\\hspace*{0.3cm}\gamma}{}}
{ \gamma}
\end{mathpar}
\begin{mathpar}
\inferrule*[Right=${T}{\lor} E$]
{ ( \varphi \lor \psi)^T \\ \inferrule*{[\varphi^T]\\[\psi^T] \\\\ \vdots \\ \hspace*{0.5cm}\vdots \\\\\gamma\\\hspace*{0.3cm}\gamma}{}}
{ \gamma}

\inferrule*[Right=${F}{\lor} E_l$]
{ ( \varphi \lor \psi)^F}
{ \varphi^F}

\inferrule*[Right=${F}{\lor} E_r$]
{ ( \varphi \lor \psi)^F}
{ \psi^F}
\end{mathpar}
\begin{mathpar}
\inferrule*[Right=${T}{\rightarrow} E$]
{ ( \varphi \rightarrow \psi)^T \\ \inferrule*{[\varphi^F]\\[\psi^T] \\\\ \vdots \\ \hspace*{0.5cm}\vdots \\\\\gamma\\\hspace*{0.3cm}\gamma}{}}
{ \gamma}

\inferrule*[Right=${F}{\rightarrow} E_l$]
{ ( \varphi \rightarrow \psi)}
{ \varphi}

\inferrule*[Right=${F}{\rightarrow} E_r$]
{ ( \varphi \rightarrow \psi)^F}
{ \psi^F}
\end{mathpar}

A deduction of $\varphi$ from a set of premises $\Gamma$ is defined as usual in natural deduction systems, and we write $\Gamma\vdash_{\letkp} \varphi$ to denote such deductions.
\end{definition}

%
%
%
%
%
%

\begin{obs} \label{obs:Hilbert}
In~\cite{CarnielliLETS}, a Hilbert-style axiomatization of $\letkp$ was also presented, which is equivalent to the original natural deduction presentation. This axiomatization is obtained by adding the following axioms to the well-known Hilbert calculus for positive classical logic (over the signature $\{\land, \vee, \to\}$), where $\alpha \leftrightarrow \beta \eqdef (\alpha \to \beta) \land (\beta \to \alpha)$:

$$\begin{array}{ll}
\neg\neg \varphi \leftrightarrow \varphi &  \hspace{2cm} \circ{\circ} \varphi\\
\neg(\varphi \land \psi) \leftrightarrow (\neg \varphi \lor \neg \psi) &  \hspace{2cm} \circ \neg \varphi \leftrightarrow \circ \varphi\\
\neg(\varphi \lor \psi) \leftrightarrow (\neg \varphi \land \neg \psi) &  \hspace{2cm} \circ(\varphi \land \psi) \leftrightarrow (\varphi^T \land \psi^T) \lor \varphi^F \lor \psi^F\\
\neg (\varphi \to \psi) \leftrightarrow (\varphi \land \neg \psi) & \hspace{2cm} \circ(\varphi \lor \psi) \leftrightarrow (\varphi^F \land \psi^F) \lor \varphi^T \lor \psi^T\\
{\circ} \varphi \to (\varphi \to (\neg \varphi \to \psi)) & \hspace{2cm} \circ(\varphi \to \psi) \leftrightarrow (\varphi \land \psi^F) \lor \varphi^F \lor \psi^T.\\
\circ \varphi \to (\varphi \lor \neg \varphi) & 
\end{array}$$
\end{obs}

By analyzing the constraints that the propagation rules of the $\circ$ operator impose on $\LETK$ to obtain the logic $\letkp$, Coniglio and Rodrigues, in \cite{coniglio2024belnap} (Theorems 17 and 19), show that the logic $\letkp$ is sound and complete with respect to the semantics induced by the following.

\begin{definition}
\textbf{(Logical matrix $\mathcal{M}_6$)}
Let $\mathcal{M}_6$ be the six-valued logical matrix with domain $\mathbf{6}=\{\mathsf{T, t, b,} \allowbreak \mathsf{n, f, F}\}$, designated values $D_6=\{\mathsf{T, t, b}\}$, and with operations defined in the following tables.

$$
\begin{array}{ccccc}
\begin{array}{c|c}
  z & \tilde{\neg}z \\
  \hline
  {\mathsf{T}} & {\mathsf{F}} \\
  {\mathsf{t}} & {\mathsf{f}} \\
  {\mathsf{b}} & {\mathsf{b}} \\
  {\mathsf{f}} & {\mathsf{t}} \\
  {\mathsf{F}} & {\mathsf{T}} \\
  {\mathsf{n}} & {\mathsf{n}} \\
\end{array}
&&
\begin{array}{c|c}
  z & \tilde{\circ}z \\
  \hline
  {\mathsf{T}} & {\mathsf{T}}\\
  {\mathsf{t}} & {\mathsf{F}} \\
  {\mathsf{b}} & {\mathsf{F}} \\
  {\mathsf{f}} & {\mathsf{F}} \\
  {\mathsf{F}} & {\mathsf{T}} \\
  {\mathsf{n}} & {\mathsf{F}} \\
\end{array}
&&
\begin{array}{c|cccccc}
\tilde{\wedge} & {\mathsf{T}} & {\mathsf{t}} & {\mathsf{b}} & {\mathsf{f}} & {\mathsf{F}} & {\mathsf{n}} \\
\hline
  {\mathsf{T}} & {\mathsf{T}} & {\mathsf{t}} & {\mathsf{b}} & {\mathsf{f}} & {\mathsf{F}} & {\mathsf{n}} \\
  {\mathsf{t}} & {\mathsf{t}} & {\mathsf{t}} & {\mathsf{b}} & {\mathsf{f}} & {\mathsf{F}} & {\mathsf{n}} \\
  {\mathsf{b}} & {\mathsf{b}} & {\mathsf{b}} & {\mathsf{b}} & {\mathsf{f}} & {\mathsf{F}} & {\mathsf{f}} \\
  {\mathsf{f}} & {\mathsf{f}} & {\mathsf{f}} & {\mathsf{f}} & {\mathsf{f}} & {\mathsf{F}} & {\mathsf{f}} \\
  {\mathsf{F}} & {\mathsf{F}} & {\mathsf{F}} & {\mathsf{F}} & {\mathsf{F}} & {\mathsf{F}} & {\mathsf{F}} \\
  {\mathsf{n}} & {\mathsf{n}} & {\mathsf{n}} & {\mathsf{f}} & {\mathsf{f}} & {\mathsf{F}} & {\mathsf{n}} \\
\end{array}
\end{array}
$$

$$
\begin{array}{ccc}
\begin{array}{c|cccccc}
\tilde{\vee} & {\mathsf{T}} & {\mathsf{t}} & {\mathsf{b}} & {\mathsf{f}} & {\mathsf{F}} & {\mathsf{n}} \\
\hline
{\mathsf{T}} & {\mathsf{T}} & {\mathsf{T}} & {\mathsf{T}} & {\mathsf{T}} & {\mathsf{T}} & {\mathsf{T}} \\
{\mathsf{t}} & {\mathsf{T}} & {\mathsf{t}} & {\mathsf{t}} & {\mathsf{t}} & {\mathsf{t}} & {\mathsf{t}} \\
{\mathsf{b}} & {\mathsf{T}} & {\mathsf{t}} & {\mathsf{b}} & {\mathsf{b}} & {\mathsf{b}} & {\mathsf{t}} \\
{\mathsf{f}} & {\mathsf{T}} & {\mathsf{t}} & {\mathsf{b}} & {\mathsf{f}} & {\mathsf{f}} & {\mathsf{n}} \\
{\mathsf{F}} & {\mathsf{T}} & {\mathsf{t}} & {\mathsf{b}} & {\mathsf{f}} & {\mathsf{F}} & {\mathsf{n}} \\
{\mathsf{n}} & {\mathsf{T}} & {\mathsf{t}} & {\mathsf{t}} & {\mathsf{n}} & {\mathsf{n}} & {\mathsf{n}} \\
\end{array}
&&
\begin{array}{c|cccccc}
\tilde{\rightarrow} & {\mathsf{T}} & {\mathsf{t}} & {\mathsf{b}} & {\mathsf{f}} & {\mathsf{F}} & {\mathsf{n}} \\
\hline
{\mathsf{T}} & {\mathsf{T}} & {\mathsf{t}} & {\mathsf{b}} & {\mathsf{f}} & {\mathsf{F}} & {\mathsf{n}} \\
{\mathsf{t}} & {\mathsf{T}} & {\mathsf{t}} & {\mathsf{b}} & {\mathsf{f}} & {\mathsf{F}} & {\mathsf{n}} \\
{\mathsf{b}} & {\mathsf{T}} & {\mathsf{t}} & {\mathsf{b}} & {\mathsf{f}} & {\mathsf{F}} & {\mathsf{n}} \\
{\mathsf{f}} & {\mathsf{T}} & {\mathsf{t}} & {\mathsf{t}} & {\mathsf{t}} & {\mathsf{t}} & {\mathsf{t}} \\
{\mathsf{F}} & {\mathsf{T}} & {\mathsf{T}} & {\mathsf{T}} & {\mathsf{T}} & {\mathsf{T}} & {\mathsf{T}} \\
{\mathsf{n}} & {\mathsf{T}} & {\mathsf{t}} & {\mathsf{t}} & {\mathsf{t}} & {\mathsf{t}} & {\mathsf{t}} \\
\end{array}
\end{array}
$$
\end{definition}

As we can observe in this matrix there are six truth-values, the values $\mathsf{t, b, n, f}$ that correspond respectively to the true, both, neither, and false values of $\fde$ as well as the values $\mathsf{T}$ and $\mathsf{F}$ that behave as the classical true and false. As pointed out  in \cite{coniglio2024belnap} these values can be ordered according to the logical lattice $\mathbf{L6}$, the crystal lattice, below.
\begin{center}
\begin{tikzpicture}[scale=.80]
\tikzset{myroundednode/.style={fill=white, rounded corners=3mm, text centered}}
\draw (0,0) -- (-1,1) --(0,2) -- (1,1)-- cycle;
\draw (0,0) -- (0,-1.5);
\draw (0,2) -- (0,3.5);
\node[myroundednode] at (0,-1.5)  {$\mathsf{F}$};
\node[myroundednode] at (0,0)     {$\mathsf{f}$};
\node[myroundednode] at (-1,1)    {$\mathsf{n}$};
\node[myroundednode] at (1,1)     {$\mathsf{b}$};
\node[myroundednode] at (0,2)     {$\mathsf{t}$};
\node[myroundednode] at (0,3.5)   {$\mathsf{T}$};
\end{tikzpicture}
\end{center}

Indeed, authors in \cite{coniglio2024belnap} present several types of semantics for the logic $\letkp$, among them the semantics induced by twist structures. They are algebras over $\Sigma$ induced by Boolean algebras. Their domains are sets of triples $z=(z_1,z_2,z_3) \in B^3$, where $\mathcal{B}$ is the underlying Boolean algebra. For an arbitrary formula $\varphi$, the components $z_1$, $z_2$, and $z_3$ represent, respectively, the values of $\varphi$, $\neg\varphi$, and $\circ\varphi$ in $\mathcal{B}$.

\begin{definition}\label{def:estructura-twist}\textbf{(Twist structures for $\letkp$)} Let $\mathcal{B} = \langle B, \sqcup, \sqcap, \sim, 0, 1 \rangle$ be a Boolean algebra, where $\sqcup$ denotes the supremum, $\sqcap$ denotes the infimum, and $\sim$ denotes the Boolean complement. Let us also consider one operator $\Rightarrow$, called implication, defined as usual by $x \Rightarrow y \eqdef {\sim x} \sqcup y$. The twist structure for $\letkp$ induced by $\mathcal{B}$ is the algebraic structure $T_{\mathcal{B}} = \langle B^3_{\letkp}, \tilde{\wedge}, \tilde{\vee}, \tilde{\rightarrow}, \tilde{\neg}, \tilde{\circ} \rangle$, where its domain is the set $B^3_{\letkp} = \big\{(z_1,z_2,z_3) \in B^{3}: z_{3} \leq z_{1} \sqcup z_{2} \mbox{ and }  z_{1} \sqcap z_{2} \sqcap z_{3} = 0 \big\}$, where $x \leq y$ means that $x \sqcap y = x$, and the operations are defined as follows:\footnote{Observe the analogy between the third coordinate of the operations, which describes the value of $\circ\varphi$ for each kind of formula, and the axioms of $\letkp$ for $\circ$ given in Observation~\ref{obs:Hilbert}. Moreover, the restrictions on the values of the triples in $B^3_{\letkp}$ correspond to the two basic properties of $\circ$ in LETs.}

\begin{enumerate}
    \item $(z_{1}, z_{2}, z_{3}) \tilde{\wedge} (w_{1}, w_{2}, w_{3}) = \bigl(z_{1} \sqcap w_{1}, z_{2} \sqcup w_{2}, (z_{1} \sqcap z_{3} \sqcap w_{1} \sqcap w_{3}) \sqcup (z_{2} \sqcap z_{3}) \sqcup (w_{2} \sqcap w_{3})\bigr)$.
    \item $(z_1, z_2, z_3) \tilde{\vee} (w_1, w_2, w_3) = \bigl(z_1 \sqcup w_1, z_2 \sqcap w_2, (z_2 \sqcap z_3 \sqcap w_2 \sqcap w_3) \sqcup (z_1 \sqcap z_3) \sqcup (w_1 \sqcap w_3)\bigr)$.
    \item $(z_1, z_2, z_3) \tilde{\rightarrow} (w_1, w_2, w_3) = \bigl(z_1 \Rightarrow w_1, z_1 \sqcap w_2, (z_1 \sqcap w_2 \sqcap w_3) \sqcup (z_2 \sqcap z_3) \sqcup (w_1 \sqcap w_3)\bigr)$.
    \item $\tilde{\neg} (z_1, z_2, z_3) = (z_2, z_1, z_3)$.
    \item $\tilde{\circ} (z_1, z_2, z_3) = (z_3, {\sim}z_3, 1)$.
\end{enumerate}
\end{definition}

\begin{obs}\label{obs:twist-constantes-y-negación}
In the logic $\letkp$, the truth constants $\top$ and $\bot$ can be defined from any formula $\varphi$ by $\top \eqdef \circ{\circ}\varphi)$, $\bot \eqdef \varphi \land \neg\varphi \land \circ\varphi.$ The language also includes an implication $\rightarrow$, which yields a classical negation via ${\sim}\varphi \eqdef \varphi \rightarrow \bot$, restoring the laws of excluded middle and explosion. The truth table for classical negation is:
$$
\begin{array}{c|c}
z            & \tilde{\sim}z \\
\hline
{\mathsf{T}} & {\mathsf{F}} \\
{\mathsf{t}} & {\mathsf{F}} \\
{\mathsf{b}} & {\mathsf{F}} \\
{\mathsf{f}} & {\mathsf{t}} \\
{\mathsf{F}} & {\mathsf{T}} \\
{\mathsf{n}} & {\mathsf{t}} \\
\end{array}
$$
Turning to the twist‑structure semantics induced by a Boolean algebra $\mathcal{B}$, the interpretations for the constants $\top$ and $\bot$ are: $\tilde{\top} = (1,0,1)$, $\tilde{\bot} = (0,1,1),$ acting as the top and bottom elements, respectively. The classical negation is interpreted as $\tilde{\sim}\,(z_1,z_2,z_3) = \bigl({\sim}z_1,\; z_1,\; z_1 \sqcup (z_2 \sqcap z_3)\bigr)$.
\end{obs}

Since twist structures are algebraic structures composed of a domain (a set of elements) and a collection of operations, the construction of a logical matrix additionally requires the specification of a set of designated values. In the search for a suitable semantics for $\letkp$ , the following matrix arises.

\begin{definition}\label{def:matriz-M(B)}
\textbf{(Logical matrix induced by a twist structure)}
Given a Boolean algebra $\mathcal{B}$, the logical matrix induced by the twist structure $T_{\mathcal{B}}$ for $\letkp$ is   
\[
\mathcal{M}(\mathcal{B}) = \langle T_{\mathcal{B}}, D_{\mathcal{B}} \rangle,
\qquad
D_{\mathcal{B}} = \big\{(z_{1}, z_{2}, z_{3}) \in B^{3}_{\letkp} : z_{1} = 1\big\}.
\]
\end{definition}

Two particular instances of the Definition \ref{def:matriz-M(B)} will be of interest to us: first, the one arising from the use of the two-element Boolean algebra, and second, the one employing the algebra of subsets of a given set. 

The first case arises by letting $\mathcal{B}$ be the two-element Boolean algebra $\mathcal{B}_{2} = \langle \mathbf{2}, \sqcup, \sqcap, \sim, 0, 1 \rangle$, where $\mathbf{2}=\{0,1\}$. This construction produces the twist structure $T_{\mathcal{B}_{2}}$, and consequently the logical matrix

\[
\mathcal{M}(\mathcal{B}_{2}) = \langle T_{\mathcal{B}_{2}}, D_{\mathcal{B}_{2}} \rangle,
\qquad
D_{\mathcal{B}_{2}} = \big\{(z_{1}, z_{2}, z_{3}) \in \mathbf{2}^{3}_{\letkp} : z_{1} = 1\big\}.
\]

\begin{obs}\label{obs:ternas-valores}
The domain of $\mathcal{M}(\mathcal{B}_{2})$, namely $\mathbf{2}^{3}_{\letkp}$, consists of 6 triples and it holds that the matrices $\mathcal{M}_6$ and $\mathcal{M}(\mathcal{B}_{2})$ are the same by taking the following translation:

$$\begin{array}{rrr}
  \mathsf{T}=(1,0,1) & \quad \mathsf{t}=(1,0,0) & \quad \mathsf{b}=(1,1,0)\\
  \mathsf{f}=(0,1,0) & \quad \mathsf{F}=(0,1,1) & \quad \mathsf{n}=(0,0,0)
\end{array}$$ 
\end{obs}

The second case arises by letting $\mathcal{B}$ be the power-set algebra $\mathcal{B}_{X} = \langle 2^{X}, \cup, \cap, (\cdot)^{C}, \emptyset, X \rangle$, interpreting the inclusion relation $\subseteq$ as the implication and assuming $X$ is a non-empty set. This construction produces the twist structure $T_{\mathcal{B}_{X}}$, and consequently the logical matrix
\[
\mathcal{M}(\mathcal{B}_{X}) = \langle T_{\mathcal{B}_{X}}, D_{\mathcal{B}_{X}} \rangle,
\qquad
D_{\mathcal{B}_{X}} = \big\{(z_{1}, z_{2}, z_{3}) \in (2^{X})^{3}_{\letkp} : z_{1} = X\big\}.
\]

\begin{obs}\label{obs:twist-sigma-alg} \textbf{(Twist $\sigma$-algebras)} 
For any non-empty set $X$, the Boolean algebra  $\mathcal{B}_{X}=2^{X}$ is a $\sigma$-algebra. In fact, it is the archetypical  example of an algebra of events in classical probability theory. In this context, it will be useful to consider the twist  structure $T_{\mathcal{B}_{X}}$ as a {\em twist $\sigma$-algebra for  $\letkp$}, i.e., an algebra of events in a probability theory based on $\letkp$. Each event is now represented by a triple $(A,B,C)$ of subsets of $X$ (called {\em snapshot}, in the terminology of the more general swap structures --- see, for instance, \cite{coniglio2024belnap}). A snapshot  $(A,B,C)$ is an informational state describing an event. It  is composed by three sets: $A$, the positive extension of the event (i.e., the states or `reasons' supporting the event); $B$, the negative extension of the event (the states or `reasons' rejecting the event, or supporting its De Morgan negation); and $C$, the reliability  extension of the event (i.e., as discussed in the Introduction, 
some states or `reasons' supporting the reliability of the epistemic attitude towards the event, whether acceptance or rejection). The latter can be explained as follows. First, it may happen that $A \cup B \subsetneq X$, meaning that some states neither support nor reject the event (a gap). Dually, some states may both support and reject it, i.e., $A \cap B \neq \emptyset$ (a glut). The set $C$ is contained in $A \cup B$ and disjoint from $A \cap B$. Thus, $C = (A \cap C) \cup (B \cap C)$ is a partition of $C$ into states that firmly support the event (and thus do not reject it) and states that firmly reject it (and thus do not support it). Consequently, $C$ is a set of states in which the event is classical: there is no doubt about its truth or falsity, and no contradiction can arise in this regard. In region $C$, the epistemic attitude towards the event is unequivocal and unrevisable.

Besides the usual expressiveness of the power-set algebra $2^{X}$ (meets as conjunction of events and so on), twist $\sigma$-algebras can express by means of an operator the (De Morgan) complement of an event, which `reads' (by `swapping' to the first coordinate) the second component of the snapshot representing the  given event, as well as its reliability (by means of another operator), which `reads' (in that sense) the third coordinate of the snapshot. The twist $\sigma$-algebra for \fde\ over $X$ is obtained from $T_{\mathcal{B}_{X}}$ by `forgetting' (or deleting) the third coordinate. These structures, formed by snapshots of the form $(A,B)$ where $A$ and $B$ are as above, constitute  a well-known sound and complete class of models for \fde\, since they represent the variety of De Morgan lattices, as Dunn stated in his PhD thesis (see~\cite[Theorem~2]{dunn2010contradictory}).
\end{obs}

Two fundamental concepts in matrix semantics are valuations and the consequence relations they induce. Based on the matrix $\mathcal{M}(\mathcal{B})$ introduced in Definition~\ref{def:matriz-M(B)}, the corresponding definitions are as follows.

\begin{definition}\label{def:valuación-en-M(B)}
        \textbf{(Valuation in $\mathcal{M}(\mathcal{B})$)} 
    Let $V$ be a set of propositional variables and $\Sigma$ a propositional signature. Given a Boolean algebra $\mathcal{B}$, a valuation in the logical matrix $\mathcal{M}(\mathcal{B})$ is a homomorphism $v$ from $\forsigma$ to $T_{\mathcal{B}}$
\end{definition}

\begin{definition}\label{def:relacion-de-consecuencia-en-M(B)}
    \textbf{(Consequence relation induced by $\mathcal{M}(\mathcal{B})$)}
    Given a set of formulas $\Gamma\cup\{\varphi\}$, we say that $\Gamma$ has $\varphi$ as a consequence in $\mathcal{M}(\mathcal{B})$, denoted by $\Gamma\models_{\mathcal{M}(\mathcal{B})}\varphi$, if for every valuation $v$ in $\mathcal{M}(\mathcal{B})$ it holds that $v(\varphi)\in D_{\mathcal{B}}$ whenever $v(\Gamma)\subset D_{\mathcal{B}}$, where $v(\Gamma)= \{v(\gamma):\gamma\in \Gamma\}$.
\end{definition}

In \cite{coniglio2024belnap} (Theorem 23) it is shown that if $\textit{Mat}(\letkp)$ is the class of logical matrices of the form $\mathcal{M}(\mathcal{B})$ and $\models_{\textit{Mat}(\letkp)}$ the associated consequence relation, then this semantics is sound and complete with respect to the proof system presented for $\letkp$. The logical matrix $\mathcal{M}(\mathcal{B}_X)$ is an element of $\textit{Mat}(\letkp)$ and therefore the consequence relation it induces is also sound and complete with respect to the proof system presented for $\letkp$.

\section{Twist Models} \label{sect:twist-models}

In \cite{klein2021probabilities}, the authors define non-standard models, grounded in a semantics for $\fde$ that separates positive and negative evidence. We extend this approach by employing the matrix semantics $\mathcal{M}(\mathcal{B}_{X})$ for $\letkp$ to introduce a new class of twist models. Like their counterparts, these models maintain a separation between positive and negative evidence; their distinctive contribution is the incorporation of reliable evidence as a fundamental component.

\begin{definition}\label{def:modelo_twist} \textbf{(Twist Model for $\letkp$)} Let $X$ be a non-empty set of elements, which we shall refer to as states. Let $T_{\mathcal{B}_X}$ denote the twist structure for $\letkp$ induced by $\mathcal{B}_X$, and let $v: \forsigma \longrightarrow T_{\mathcal{B}_X}$ be a valuation in $\mathcal{M}(\mathcal{B}_X)$, such that for every formula $\varphi$, $v(\varphi) = \pa{v_1(\varphi), v_2(\varphi), v_3(\varphi) }$. A twist model for $\letkp$ is then defined as the ordered pair $\mathsf{M} = \langle X, v \rangle$.
\end{definition}

\begin{definition}\label{def:extensiones:pos_neg_conf} \textbf{(Extension of a Formula in a Twist Model for $\letkp$ )}  
Let \(\varphi \in \forsigma\) be a formula and \(\mathsf{M} = \langle X, v \rangle\) a twist model for the logic $\letkp$.  Whenever \(v(\varphi) = \bigl( v_1(\varphi), v_2(\varphi), v_3(\varphi) \bigr)\), we define the following sets:
\begin{itemize}
\item $|\varphi|_\mathsf{M}^+ = v_1(\varphi)$, referred to as the positive extension of $\varphi$ in the model $\mathsf{M}$.
\item $|\varphi|_\mathsf{M}^- = v_2(\varphi)$, referred to as the negative extension of $\varphi$ in the model $\mathsf{M}$.
\item $|\varphi|_\mathsf{M}^\circ = v_3(\varphi)$, referred to as the reliable extension of $\varphi$ in the model $\mathsf{M}$.
\end{itemize}
\end{definition}

\begin{figure}[H]
\begin{center}
\begin{tikzpicture}[pencildraw/.style={decoration={random steps,segment length=10pt, amplitude=3pt}},scale=.55]
\tikzset{myroundednode/.style={fill=white, rounded corners=3mm, text centered }}
\draw (0,0) -- (8,0) --(8,8) -- (0,8)-- cycle;
\draw[pattern={Lines[angle=-45,distance=5pt]},pattern color=blue]   (0,0)--(4,0) decorate[pencildraw]{--(4,8)}--(0,8) --cycle;
\draw[pattern={Lines[angle=45, distance=5pt]},pattern color=red]  (0,4) decorate[pencildraw]{--(8,4)}--(8,8)--(0,8)   --cycle;
\draw[pattern={Lines[angle=0,  distance=5pt]},pattern color=green] (0,4) decorate[pencildraw]{--(4,0)}--(0,0)   --cycle;
\draw[pattern={Lines[angle=0,  distance=5pt]},pattern color=green] (4,8) decorate[pencildraw]{--(8,4)}--(8,8)   --cycle;
\draw[pattern={Lines[angle=-45,distance=5pt]},pattern color=blue]  (9,6)--(9,6.7)--(9.7,6.7)--(9.7,6) --cycle;
\draw[pattern={Lines[angle=45, distance=5pt]},pattern color=red]   (9,3.7)--(9,4.4)--(9.7,4.4)--(9.7,3.7) --cycle;
\draw[pattern={Lines[angle=0,  distance=5pt]},pattern color=green] (9,1.4)--(9,2.1)--(9.7,2.1)--(9.7,1.4) --cycle;
\node[fill=white] at (12,6.3)  (Pos) {$v_1(\varphi)=|\varphi|^+_\mathsf{M}$};
\node[fill=white] at (12,4)    (Neg) {$v_2(\varphi)=|\varphi|^-_\mathsf{M}$};
\node[fill=white] at (12,1.7)  (Conf) {$v_3(\varphi)=|\varphi|^{\circ}_\mathsf{M}$};
\end{tikzpicture}
\caption{Positive, negative, and reliable extensions of a formula $\varphi$ in a  twist model $\mathsf{M} = \langle X, v \rangle$ for the logic $\letkp$.}
\label{fig:extensiones}
\end{center}
\end{figure}
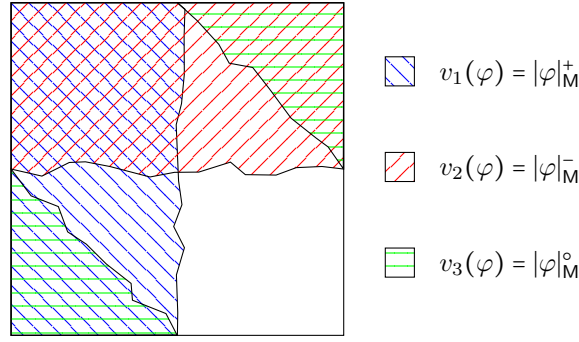

In Figure \ref{fig:extensiones}, one can observe an example of a distribution of the three extensions given in Definition \ref{def:extensiones:pos_neg_conf} for a particular formula $\varphi$. Given that $v$ constitutes a valuation from $\forsigma$ into $T_{\mathcal{B}_X}$, the following properties are satisfied:

\begin{proposition}\label{prop:propiedades-extensiones}
For any $\varphi, \psi \in \forsigma$ and for any  twist model $\mathsf{M} = \langle X, v \rangle$ for $\letkp$, it holds that:

\begin{enumerate}
\setlength{\columnsep}{1.5cm}
\begin{multicols}{2}
\item $|\varphi|_\mathsf{M}^\circ \subseteq |\varphi|_\mathsf{M}^+ \cup |\varphi|_\mathsf{M}^-$. 
\item $|\varphi|_\mathsf{M}^+ \cap |\varphi|_\mathsf{M}^- \cap |\varphi|_\mathsf{M}^\circ = \emptyset$.    
\end{multicols}

\begin{multicols}{2}
\item $|\varphi \land \psi|_\mathsf{M}^+     = |\varphi|_\mathsf{M}^+ \cap |\psi|_\mathsf{M}^+$.
\item $|\varphi \land \psi|_\mathsf{M}^-     = |\varphi|_\mathsf{M}^- \cup |\psi|_\mathsf{M}^-$.   
\end{multicols}

\item $|\varphi \land \psi|_\mathsf{M}^\circ = \pa{|\varphi|_\mathsf{M}^+ \cap |\varphi|_\mathsf{M}^\circ  \cap |\psi|_\mathsf{M}^+ \cap |\psi|_\mathsf{M}^\circ} \cup \pa{|\varphi|_\mathsf{M}^- \cap |\varphi|_\mathsf{M}^\circ} \cup \pa{|\psi|_\mathsf{M}^- \cap |\psi|_\mathsf{M}^\circ}$.

\begin{multicols}{2}
\item $|\varphi \lor \psi|_\mathsf{M}^+     = |\varphi|_\mathsf{M}^+ \cup |\psi|_\mathsf{M}^+$.
\item $|\varphi \lor \psi|_\mathsf{M}^-     = |\varphi|_\mathsf{M}^- \cap |\psi|_\mathsf{M}^-$.
\end{multicols}
\item $|\varphi \lor \psi|_\mathsf{M}^\circ = \pa{|\varphi|_\mathsf{M}^- \cap |\varphi|_\mathsf{M}^\circ  \cap |\psi|_\mathsf{M}^- \cap |\psi|_\mathsf{M}^\circ} \cup \pa{|\varphi|_\mathsf{M}^+ \cap |\varphi|_\mathsf{M}^\circ} \cup \pa{|\psi|_\mathsf{M}^+ \cap |\psi|_\mathsf{M}^\circ} $.

\begin{multicols}{2}
\item $|\varphi \rightarrow \psi|_\mathsf{M}^+  = \pa{|\varphi|_\mathsf{M}^+ \Rightarrow |\psi|_\mathsf{M}^+} = \pa{X\setminus|\varphi|_\mathsf{M}^+} \cup |\psi|_\mathsf{M}^+$.
\item $|\varphi \rightarrow \psi|_\mathsf{M}^-  = |\varphi|_\mathsf{M}^+ \cap |\psi|_\mathsf{M}^-$.
\end{multicols}
\item $|\varphi \rightarrow \psi|_\mathsf{M}^\circ = \pa{|\varphi|_\mathsf{M}^+ \cap |\psi|_\mathsf{M}^- \cap |\psi|_\mathsf{M}^\circ} \cup \pa{|\varphi|_\mathsf{M}^- \cap |\varphi|_\mathsf{M}^\circ} \cup \pa{|\psi|_\mathsf{M}^+ \cap |\psi|_\mathsf{M}^\circ}$.

\begin{multicols}{3}
\item $|\neg \varphi|_\mathsf{M}^+ = |\varphi|_\mathsf{M}^-$.
\item $|\neg \varphi|_\mathsf{M}^- = |\varphi|_\mathsf{M}^+$.
\item $|\neg \varphi|_\mathsf{M}^\circ = |\varphi|_\mathsf{M}^\circ.$
\end{multicols}

\begin{multicols}{3}
\item $| {\circ} \varphi|_\mathsf{M}^+ = |\varphi|_\mathsf{M}^\circ$. 
\item $| {\circ} \varphi|_\mathsf{M}^- = X \setminus |\varphi|_\mathsf{M}^\circ$.
\item $| {\circ} \varphi|_\mathsf{M}^\circ  = X$.
\end{multicols}

\begin{multicols}{3}
\item $|\neg \neg \varphi|_\mathsf{M}^+ = |\varphi|_\mathsf{M}^+$.
\item $|\neg \neg \varphi|_\mathsf{M}^- = |\varphi|_\mathsf{M}^-$.
\item $|\neg \neg \varphi|_\mathsf{M}^\circ = |\varphi|_\mathsf{M}^\circ$.
\end{multicols}

\begin{multicols}{3}
\item $|{\circ\circ}\varphi|_\mathsf{M}^+  = X$.
\item $|{\circ\circ}\varphi|_\mathsf{M}^-  = \emptyset$.
\item $|{\circ\circ}\varphi|_\mathsf{M}^\circ  = X$.
\end{multicols}

\begin{multicols}{3}
\item $|\varphi\land\neg\varphi\land\circ\varphi|_\mathsf{M}^+  = \emptyset$.
\item $|\varphi\land\neg\varphi\land\circ\varphi|_\mathsf{M}^-  = X$.
\item $|\varphi\land\neg\varphi\land\circ\varphi|_\mathsf{M}^\circ  = X$.
\end{multicols}

\begin{multicols}{3}
\item $|{\sim} \varphi|_\mathsf{M}^+ = X\setminus |\varphi|_\mathsf{M}^+ $.
\item $|{\sim} \varphi|_\mathsf{M}^- = |\varphi|_\mathsf{M}^+$.
\item $|{\sim} \varphi|_\mathsf{M}^\circ = |\varphi|_\mathsf{M}^+\cup \pa{|\varphi|_\mathsf{M}^-\cap |\varphi|_\mathsf{M}^\circ}.$
\end{multicols}
\end{enumerate}
\end{proposition}

\begin{proof}
Properties 1-17 follow directly from the definition of the domain and the operations in the twist structure $T_{\mathcal{B}_X}$. The remaining properties are readily obtained from the first seventeen properties together with the statements in Observation \ref{obs:twist-constantes-y-negación}.
\end{proof}

Properties 1 and 2 of Proposition \ref{prop:propiedades-extensiones} preclude the existence of elements in $|\varphi|_\mathsf{M}^\circ \setminus \pa{|\varphi|_\mathsf{M}^+ \cup |\varphi|_\mathsf{M}^-}$ and in $|\varphi|_\mathsf{M}^+ \cap |\varphi|_\mathsf{M}^- \cap |\varphi|_\mathsf{M}^\circ$. Consequently, the set $X$ is partitioned into exactly six disjoint regions. Employing the sets from Definition \ref{def:extensiones:pos_neg_conf}, we can label each of the six disjoint regions in Figure \ref{fig:extensiones} as indicated in the following definition:

\begin{definition}\label{def:particion6valores}
\textbf{(Sets Induced by a Formula $\varphi$ and a  Twist Model \(\mathsf{M} = \langle X, v \rangle\))}
Let \(\varphi \in \forsigma\) be a formula and \(\mathsf{M} = \langle X, v \rangle\) be a  twist model for $\letkp$. Whenever $v(\varphi) = \pa{|\varphi|_\mathsf{M}^+, |\varphi|_\mathsf{M}^-, |\varphi|_\mathsf{M}^\circ}$, we define the following subsets of $X$ induced by $\varphi$ and by $\mathsf{M}$: 
\begin{itemize}[nosep]
\begin{multicols}{2}
\item[] $B^\circ_\varphi = |\varphi|^+_\mathsf{M} \cap |\varphi|^\circ_\mathsf{M}$ reliable belief for $\varphi$ in $\mathsf{M}$.
\item[] $C_\varphi = |\varphi|^+_\mathsf{M} \cap |\varphi|^-_\mathsf{M}$  conflict for $\varphi$ in $\mathsf{M}$.
\item[] $D_\varphi^\circ = |\varphi|^-_\mathsf{M} \cap |\varphi|^\circ_\mathsf{M}$  reliable disbelief for $\varphi$ in $\mathsf{M}$.
\item[] $B_\varphi = |\varphi|^+_\mathsf{M} \setminus \pa{|\varphi|^-_\mathsf{M} \cup |\varphi|^\circ_\mathsf{M}}$ belief for $\varphi$ in $\mathsf{M}$.
\item[] $D_\varphi = |\varphi|^-_\mathsf{M} \setminus \pa{|\varphi|^+_\mathsf{M} \cup |\varphi|^\circ_\mathsf{M}}$ disbelief for $\varphi$ in $\mathsf{M}$.
\item[] $U_\varphi = X \setminus \pa{|\varphi|^+_\mathsf{M} \cup |\varphi|^-_\mathsf{M}}$ uncertainty for $\varphi$ in $\mathsf{M}$.
\end{multicols}
\end{itemize}
\end{definition}

In Figure~\ref{fig:nuevas_regiones}, one can identify the six regions from Definition \ref{def:particion6valores}.

\begin{figure}[H]
\begin{center}
\begin{tikzpicture}[pencildraw/.style={decoration={random steps,segment length=10pt, amplitude=3pt}},scale=.60]
\tikzset{myroundednode/.style={fill=white, rounded corners=3mm, text centered}}
\draw (0,0) -- (8,0) --(8,8) -- (0,8)-- cycle;
\draw[pattern={Lines[angle=-45,distance=5pt]},pattern color=blue] (0,0)--(4,0) decorate[pencildraw]{--(4,8)}--(0,8) --cycle;
\draw[pattern={Lines[angle=45,distance=5pt]},pattern color=red] (0,4) decorate[pencildraw]{--(8,4)}--(8,8)--(0,8)   --cycle;
\draw[pattern={Lines[angle=0,distance=5pt]},pattern color=green] (0,4) decorate[pencildraw]{--(4,0)}--(0,0)   --cycle;
\draw[pattern={Lines[angle=0,distance=5pt]},pattern color=green] (4,8) decorate[pencildraw]{--(8,4)}--(8,8)   --cycle;
\node[myroundednode] at (1.2,1.1)  (Do) {$B^{\circ}_\varphi$};
\node[myroundednode] at (2.8,2.8)  (D) {$B_\varphi$};
\node[myroundednode] at (2,6)      (C) {$C_\varphi$};
\node[myroundednode] at (6,2)      (U) {$U_\varphi$};
\node[myroundednode] at (5.2,5.1)  (B) {$D_\varphi$};
\node[myroundednode] at (6.8,6.8)  (Bo) {$D^{\circ}_\varphi$};
\draw[pattern={Lines[angle=-45,distance=5pt]},pattern color=blue] (9,6)--(9,6.7)--(9.7,6.7)--(9.7,6) --cycle;
\draw[pattern={Lines[angle=45,distance=5pt]},pattern color=red] (9,3.7)--(9,4.4)--(9.7,4.4)--(9.7,3.7) --cycle;
\draw[pattern={Lines[angle=0,distance=5pt]},pattern color=green] (9,1.4)--(9,2.1)--(9.7,2.1)--(9.7,1.4) --cycle;
\node[fill=white] at (12,6.3)  (Pos) {$v_1(\varphi)=|\varphi|^+_\mathsf{M}$};
\node[fill=white] at (12,4)  (Neg) {$v_2(\varphi)=|\varphi|^-_\mathsf{M}$};
\node[fill=white] at (12,1.7)  (Conf) {$v_3(\varphi)=|\varphi|^{\circ}_\mathsf{M}$};
\end{tikzpicture}
\caption{Subsets of $X$ induced by a  twist model and a formula $\varphi$ in $\letkp$.}
\label{fig:nuevas_regiones}
\end{center}
\end{figure}
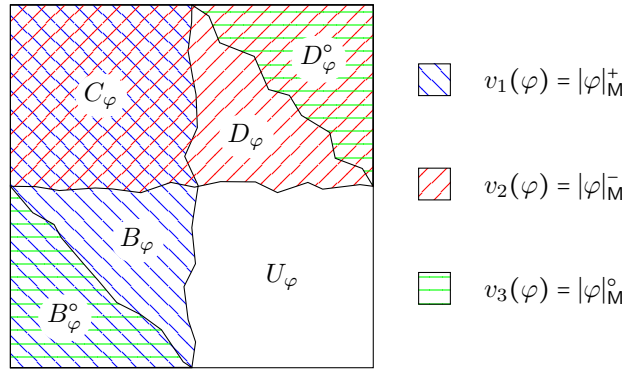

\begin{proposition}\label{prop:6-particion}
Given a formula $\varphi \in \forsigma$ and a  twist model $\mathsf{M} = \langle X, v \rangle$ for $\letkp$, the set $\mathcal{P}_\varphi=\{B^{\circ}_\varphi, B_\varphi, C_\varphi, D_\varphi, D^{\circ}_\varphi,  U_\varphi\}$, induced by $\mathsf{M}$ and $\varphi$ is indeed a partition of $X$.
\end{proposition}
\begin{proof}
    The proof follows from Definition \ref{def:particion6valores} by means of set algebra properties and Proposition \ref{prop:propiedades-extensiones}.
\end{proof}

\begin{obs}\label{obs:particion-con-extensones-positivas}
Thanks to the properties of the positive, negative, and reliable extensions given in Proposition \ref{prop:propiedades-extensiones}, it can be proved that each of the six sets in Definition \ref{def:particion6valores} can be expressed as the positive extension of a suitable formula, as indicated in the following equalities.

$$\begin{array}{lll}
B^\circ_\varphi=|\varphi\land\circ\varphi|^+
&
\quad B_\varphi =|\varphi\land{\sim}(\neg\varphi\lor\circ\varphi)|^+
&
\quad C_\varphi = |\varphi\land\neg\varphi|^+\\
D_\varphi =|\neg\varphi\land{\sim}(\varphi\lor\circ\varphi)|^+
&
\quad D_\varphi^\circ =|\neg\varphi\land\circ\varphi|^+        
&
\quad U_\varphi =|{\sim}(\varphi\lor\neg\varphi)|^+ \\
\end{array}$$
For the sake of simplicity we will denote such suitable formulas as: 
\begin{align*}
\varphi_{B^\circ} &= \varphi \land\circ \varphi &  \varphi_B &= \varphi\land{\sim}\pa{\neg \varphi\lor {\circ} \varphi }&
\varphi_C &= \varphi\land\neg \varphi \\   
\varphi_D &= \neg\varphi\land {\sim}\pa{\varphi\lor{\circ} \varphi }&\varphi_{D^\circ} &= \neg \varphi\land\circ \varphi&  \varphi_U &= {\sim}\pa{\varphi\lor\neg \varphi}
\end{align*}
We will refer to these formulas as the region-identification formulas.
\end{obs}

If $\varphi$ is a formula in $\letkp$, according to Definition \ref{def:particion6valores}, $\varphi$ partitions the space $X$ into six regions.
Before to continue with more properties for these regions, it will be useful to point out the relation between these regions and the truth values in $\mathcal{M}_6$. Given a  twist model $\mathsf{M} = \langle X, v \rangle$ and a formula $\varphi$, we have that $v(\varphi) = \pa{|\varphi|_\mathsf{M}^+, |\varphi|_\mathsf{M}^-, |\varphi|_\mathsf{M}^\circ}$. For an specific state $x \in X$ we define the function $v_x:For(\Sigma)\longrightarrow T_{\mathcal{B}_2}$ as $v_x\pa{\varphi}= \pa{\chi_{|\varphi|_\mathsf{M}^+}(x),\chi_{|\varphi|_\mathsf{M}^-}(x),\chi_{|\varphi|_\mathsf{M}^\circ}(x)}$, where $\chi_A$ denotes the characteristic (or indicator) function of the set $A$ (i.e. $\chi_{A}(x)=1$ iff $x\in A$ and $\chi_{A}(x)=0$ iff $x\notin A$). That is, $v_x$ assigns a triple of ones and zeros to $\varphi$, depending on whether $x$ belongs to the respective set in the triple $\pa{|\varphi|_\mathsf{M}^+, |\varphi|_\mathsf{M}^-, |\varphi|_\mathsf{M}^\circ}$. Due to Property 1 of Proposition \ref{prop:propiedades-extensiones} it holds that  $|\varphi|_\mathsf{M}^+ \cap |\varphi|_\mathsf{M}^- \cap |\varphi|_\mathsf{M}^\circ=\emptyset$, so $v_x\pa{\varphi}\neq(1,1,1)$. Moreover, by Property 2 of the same proposition it holds that $|\varphi|_\mathsf{M}^\circ \setminus \pa{|\varphi|_\mathsf{M}^+ \cup |\varphi|_\mathsf{M}^-}=\emptyset$, causing that $v_x\pa{\varphi}\neq(0,0,1)$. In fact $v_x$ is a valuation in $\mathcal{M}(\mathcal{B}_2)$ and given the identification among triples in $\mathcal{M}(\mathcal{B}_2)$ and values in $\mathcal{M}_6$ given in Observation \ref{obs:ternas-valores}, it is also a valuation in $\mathcal{M}_6$. 

\begin{obs}\label{obs:identificacion-regiones-valores}
Regarding the sets of Definition \ref{def:particion6valores} and the function $v_x$  we have that:
\begin{center}
\begin{multicols}{2}
$x\in B_\varphi^\circ$  iff $v_x\pa{\varphi}=(1,0,1)=\mathsf{T}$\\
$x\in C_\varphi$        iff $v_x\pa{\varphi}=(1,1,0)=\mathsf{b}$\\
$x\in D_\varphi^\circ$  iff $v_x\pa{\varphi}=(0,1,1)=\mathsf{F}$\\
$x\in B_\varphi$        iff $v_x\pa{\varphi}=(1,0,0)=\mathsf{t}$\\
$x\in D_\varphi$        iff $v_x\pa{\varphi}=(0,1,0)=\mathsf{f}$\\
$x\in U_\varphi$        iff $v_x\pa{\varphi}=(0,0,0)=\mathsf{n}$    
\end{multicols}
\end{center}

Hence each truth value in $\mathcal{M}_6$ can be identified with exactly one region and vice versa.    
\end{obs}

\begin{proposition}\label{prop:propiedades-6-regiones}
For any $\varphi, \psi \in \forsigma$ and for any  twist model $\mathsf{M} = \langle X, v \rangle$ for $\letkp$, it holds that:

General properties
\begin{enumerate}
\begin{multicols}{2}
\item $B^\circ_\varphi  \cup B_\varphi  \cup C_\varphi          = |\varphi|^+_\mathsf{M}$.
\item $C_\varphi        \cup D_\varphi  \cup D^\circ_\varphi    = |\varphi|^-_\mathsf{M}$.
\end{multicols}

\begin{multicols}{2}
\item $B_{\varphi}^\circ \cup  D_{\varphi}^\circ  = |\varphi|_\mathsf{M}^\circ$.
\item $B^\circ_\varphi = B^\circ_{\varphi\land\circ\varphi}\cup B_{\varphi\land\circ\varphi}\cup C_{\varphi\land\circ\varphi}$.
\end{multicols}
\end{enumerate}

Negation properties
\begin{enumerate}
\setcounter{enumi}{4}
\begin{multicols}{3}
\item $B^\circ_{\neg \varphi}=D^\circ_{\varphi}$.
\item $B_{\neg\varphi} =D_\varphi$.
\item $C_{\neg \varphi}= C_{\varphi}$.

\end{multicols}
\begin{multicols}{3}
\item $D_{\neg\varphi} =B_\varphi$.
\item $D^\circ_{\neg \varphi}=B^\circ_{\varphi}$.
\item $U_{\neg \varphi}= U_{\varphi}$.
\end{multicols}
\end{enumerate}

Conjunction properties
\begin{enumerate}
\setcounter{enumi}{10}
\item $B_{\varphi \land \psi}^\circ = B_\varphi^\circ \cap B_\psi^\circ$.

\item $B_{\varphi \land \psi}       = \pa{B_\varphi \cap \pa{B_\psi^\circ \cup B_\psi}} \cup
                                      \pa{B_\psi \cap \pa{B_\varphi^\circ \cup B_\varphi}}$.
                                      
\item $C_{\varphi \land \psi}       = \pa{C_\varphi \cap \pa{B_\psi^\circ \cup B_\psi \cup C_\psi}} \cup
                                      \pa{C_\psi \cap \pa{B_\varphi^\circ \cup B_\varphi \cup C_\varphi}}$.

\item $D_{\varphi \land \psi}       = \pa{D_\varphi \setminus D_\psi^\circ} \cup
                                      \pa{D_\psi \setminus D_\varphi^\circ} \cup 
                                      \pa{C_\varphi \cap U_\psi} \cup
                                      \pa{C_\psi \cap U_\varphi}$.
\item $D_{\varphi \land \psi}^\circ = D_{\varphi}^\circ \cup D_{\psi}^\circ$.

\item $U_{\varphi \land \psi}       = \pa{U_\varphi \cap \pa{B_\psi^\circ \cup B_\psi \cup U_\psi}} \cup
                                      \pa{U_\psi \cap \pa{B_\varphi^\circ \cup B_\varphi \cup U_\varphi}}$ .
\end{enumerate}

Disjunction properties

\begin{enumerate}
\setcounter{enumi}{16}
\item $B_{\varphi \lor \psi}^\circ = B_\varphi^\circ \cup B_\psi^\circ$.

\item $B_{\varphi \lor \psi}       = \pa{B_\varphi \setminus B_\psi^\circ} \cup
                                     \pa{B_\psi \setminus B_\varphi^\circ} \cup 
                                     \pa{C_\varphi \cap U_\psi} \cup \pa{C_\psi \cap U_\varphi}$.

\item $C_{\varphi \lor \psi}       = \pa{C_\varphi \cap \pa{C_\psi \cup D_\psi \cup D_\psi^\circ}} \cup
                                     \pa{C_\psi    \cap \pa{C_\varphi \cup D_\varphi \cup D_\varphi^\circ}}$

\item $D_{\varphi \lor \psi}       = \pa{D_\varphi \cap \pa{D_\psi \cup D_\psi^\circ}} \cup
                                     \pa{D_\psi    \cap \pa{D_\varphi \cup D_\varphi^\circ}}$.    

\item $D_{\varphi \lor \psi}^\circ = D_{\varphi}^\circ \cap D_{\psi}^\circ$.

\item $U_{\varphi \lor \psi}       = \pa{U_\varphi \cap \pa{D_\psi \cup D_\psi^\circ \cup U_\psi}} \cup
                                     \pa{U_\psi    \cap \pa{D_\varphi \cup D_\varphi^\circ \cup U_\psi}}$.
\end{enumerate}
\newpage
Implication properties

\begin{enumerate}
\setlength{\multicolsep}{8pt}
\setcounter{enumi}{22}
\begin{multicols}{2}
\item $B_{\varphi \rightarrow \psi}^\circ = D_{\varphi}^\circ \cup B_{\psi}^\circ$.
\item $B_{\varphi \rightarrow \psi}       = 
\pa{\pa{D_\varphi \cup U_\varphi} \setminus  B_\psi^\circ} \cup  \pa{B_\psi \setminus D_\varphi^\circ }$.
\end{multicols}

\begin{multicols}{2}
\item $C_{\varphi \rightarrow \psi}       = \pa{B_\varphi^\circ \cup B_\varphi \cup C_\varphi} \cap C_\psi$.
\item $D_{\varphi \rightarrow \psi}       = \pa{B_\varphi^\circ \cup B_\varphi \cup C_\varphi} \cap D_\psi$.   
\end{multicols}

\begin{multicols}{2}
\item $D_{\varphi \rightarrow \psi}^\circ = \pa{B_\varphi^\circ \cup B_\varphi \cup C_\varphi} \cap D_\psi^\circ$.
\item $U_{\varphi \rightarrow \psi}       = \pa{B_\varphi^\circ \cup B_\varphi \cup C_\varphi} \cap U_\psi$.
\end{multicols}
\end{enumerate}

Reliability properties

\begin{enumerate}
\setlength{\multicolsep}{8pt}
\setcounter{enumi}{28}
\begin{multicols}{3}
\item $B_{\circ\varphi}^\circ = B_\varphi^\circ \cup D_\varphi^\circ$.
\item $B_{\circ\varphi} = \emptyset$.
\item $C_{\circ\varphi} = \emptyset$.
\end{multicols}
\begin{multicols}{3}
\item $D_{\circ\varphi} = \emptyset$.
\item $D_{\circ\varphi}^\circ = X \setminus \pa{B_\varphi^\circ \cup D_\varphi^\circ}$.
\item $U_{\circ\varphi} = \emptyset$.
\end{multicols}
\end{enumerate}

Classic negation properties

\begin{enumerate}
\setlength{\multicolsep}{8pt}
\setcounter{enumi}{34}
\begin{multicols}{3}
\item $B_{\sim\varphi}^\circ= D_{\varphi}^\circ$.
\item $B_{\sim\varphi} = D_{\varphi} \cup U_{\varphi}$.
\item $C_{\sim\varphi} = \emptyset$.
\end{multicols}
\begin{multicols}{3}
\item $D_{\sim\varphi} = \emptyset$.
\item $D_{\sim\varphi}^\circ = B_{\varphi}^\circ \cup B_{\varphi} \cup C_{\varphi}$.
\item $U_{\sim\varphi} = \emptyset$.
\end{multicols}
\end{enumerate}

Contradiction properties

\begin{enumerate}
\setlength{\multicolsep}{8pt}
\setcounter{enumi}{40}
\begin{multicols}{3}
\item $B^\circ_{\varphi \land \neg \varphi} = \emptyset$.
\item $B_{\varphi \land \neg \varphi} = \emptyset$.
\item $C_{\varphi \land \neg \varphi}= C_\varphi$.

\end{multicols}
\begin{multicols}{3}
\item $D_{\varphi \land \neg \varphi}= B_\varphi \cup D_\varphi$.
\item $D_{\varphi \land \neg \varphi}^\circ= B_{\varphi}^\circ \cup D^\circ _{\varphi}$.
\item $U_{\varphi \land \neg \varphi}= U_\varphi$.
\end{multicols}
\end{enumerate}

Tautology properties

\begin{enumerate}
\setlength{\multicolsep}{8pt}
\setcounter{enumi}{46}
\begin{multicols}{3}
\item $B_{{{\circ\circ} } \varphi}^\circ= X$.
\item $B_{{{\circ\circ} } \varphi} = \emptyset$.
\item $C_{{{\circ\circ} } \varphi} = \emptyset$.
\end{multicols}
\begin{multicols}{3}
\item $D_{{{\circ\circ} } \varphi} = \emptyset$.
\item $D_{{{\circ\circ} } \varphi}^\circ = \emptyset$. 
\item $U_{{{\circ\circ} } \varphi} = \emptyset$.
\end{multicols} 
\end{enumerate}
\end{proposition}\bigskip 

\begin{proof}~
\begin{enumerate}

\item $\begin{aligned}[t]
B^\circ_\varphi\cup B_\varphi\cup C_\varphi 
& = 
\pa{|\varphi|^+_\mathsf{M} \cap |\varphi|^\circ_\mathsf{M} }\cup
\pa{|\varphi|^+_\mathsf{M} \setminus \pa{ |\varphi|^-_\mathsf{M} \cup |\varphi|^\circ_\mathsf{M}} }\cup 
\pa{|\varphi|^+_\mathsf{M} \cap |\varphi|^-_\mathsf{M} }\\
& = 
\pa{|\varphi|^+_\mathsf{M} \cap |\varphi|^\circ_\mathsf{M} }\cup
\pa{|\varphi|^+_\mathsf{M} \cap |\varphi|^-_\mathsf{M} } \cup
\pa{|\varphi|^+_\mathsf{M} \setminus \pa{ |\varphi|^-_\mathsf{M} \cup |\varphi|^\circ_\mathsf{M}} } 
\\
& = 
\pa{|\varphi|^+_\mathsf{M} \cap \pa{|\varphi|^-_\mathsf{M} \cup |\varphi|^\circ_\mathsf{M}}}\cup
\pa{|\varphi|^+_\mathsf{M} \cap \pa{X \setminus \pa{|\varphi|^-_\mathsf{M} \cup |\varphi|^\circ_\mathsf{M}}}}\\
&=
|\varphi|^+_\mathsf{M} \cap 
\pa{
    \pa{|\varphi|^-_\mathsf{M} \cup |\varphi|^\circ_\mathsf{M}} \cup
    \pa{ X \setminus \pa{|\varphi|^-_\mathsf{M} \cup |\varphi|^\circ_\mathsf{M}}}
    }\\
&=
|\varphi|^+_\mathsf{M} \cap X\\
&=
|\varphi|^+_\mathsf{M}.
\end{aligned}$    
   
\item $\begin{aligned}[t]
C_\varphi  \cup D_\varphi\cup D^\circ_\varphi
& = 
D^\circ_\varphi\cup D_\varphi\cup C_\varphi\\
& = 
\pa{|\varphi|^-_\mathsf{M} \cap |\varphi|^\circ_\mathsf{M} }\cup
\pa{|\varphi|^-_\mathsf{M} \setminus \pa{|\varphi|^+_\mathsf{M} \cup |\varphi|^\circ_\mathsf{M}} }\cup 
\pa{|\varphi|^+_\mathsf{M} \cap |\varphi|^-_\mathsf{M} } \\
& = 
\pa{|\varphi|^-_\mathsf{M} \cap \pa{|\varphi|^+_\mathsf{M} \cup |\varphi|^\circ_\mathsf{M}}   }\cup
\pa{|\varphi|^-_\mathsf{M} \cap \pa{X \setminus \pa{|\varphi|^+_\mathsf{M} \cup |\varphi|^\circ_\mathsf{M}}}}\\
& =
|\varphi|^-_\mathsf{M} \cap 
\pa{
    \pa{|\varphi|^+_\mathsf{M}\cup |\varphi|^\circ_\mathsf{M}} \cup
    \pa{X \setminus \pa{|\varphi|^+_\mathsf{M}\cup |\varphi|^\circ_\mathsf{M}}}
    }\\
& =
|\varphi|^-_\mathsf{M} \cap X\\
&=
|\varphi|^-_\mathsf{M}.  
\end{aligned}$

\item $\begin{aligned}[t]
B_{\varphi}^\circ \cup  D_{\varphi}^\circ
& = 
\pa{|{\varphi}|_\mathsf{M}^+ \cap |{\varphi}|_\mathsf{M}^\circ} \cup
\pa{|{\varphi}|_\mathsf{M}^-\cap|{\varphi}|_\mathsf{M}^\circ }\\
& = 
\pa{|{\varphi}|_\mathsf{M}^+ \cup |{\varphi}|_\mathsf{M}^-} \cap
|\varphi|_\mathsf{M}^\circ\\
& = 
|\varphi|_\mathsf{M}^\circ. 
\end{aligned}$

\item $\begin{aligned}[t]
B^\circ_\varphi
& = 
|\varphi|_\mathsf{M}^+\cap|\varphi|_\mathsf{M}^\circ\\
& = 
|\varphi|_\mathsf{M}^+\cap|{\circ}\varphi|_\mathsf{M}^+\\
& =
|\varphi\land\circ\varphi|_\mathsf{M}^+\\
& = 
B^\circ_{\varphi\land{\circ}\varphi}\cup 
B_{\varphi\land{\circ}\varphi}\cup 
C_{\varphi\land{\circ}\varphi}.
\end{aligned}$
\end{enumerate}

For the proof of the remaining properties, it suffices to consider the truth tables of the connectives and to observe that the expression $\mathcal{R}_\gamma$ represents the region where the formula $\gamma$ takes the truth value assigned for the region $\mathcal{R}$. Note that the letter $\mathcal{R}$ can be replaced by any of the regions into which the space is partitioned given any formula, namely $B^\circ, B, C, D, D^\circ, U$. Moreover, as seen in Observation \ref{obs:identificacion-regiones-valores}, each of these regions is identified with one and only one truth value: $B^\circ$ with $\mathsf{T}$, $B$ with $\mathsf{t}$, $C$ with $\mathsf{b}$, $D$ with $\mathsf{f}$, $D^\circ$ with $\mathsf{F}$, and $U$ with $\mathsf{n}$. Thus, for instance in the item 12, when we write $B_{\varphi \land \psi}$, we refer to the region in which the formula $\varphi \land \psi$ takes the truth value $\mathsf{t}$. Directly from the truth table, one can see that this occurs in two ways: either when $\varphi$ takes the value $\mathsf{t}$ and $\psi$ takes either $\mathsf{T}$ or $\mathsf{t}$, or when $\psi$ takes the value $\mathsf{t}$ and $\varphi$ takes either $\mathsf{T}$ or $\mathsf{t}$. This is expressed as:
\[
\bigl( B_\varphi \cap (B_\psi^\circ \cup B_\psi) \bigr) \cup \bigl( B_\psi \cap (B_\varphi^\circ \cup B_\varphi) \bigr).
\]

\end{proof}

\begin{example} \label{ex:relig} \textbf{[Theistic Belief]}
Let $X$ be a sample space consisting of the residents of a specific country, and consider the following event $\varphi:=$ ``An individual believes in the existence of God''. Let $\mathsf{M} = \langle X, v \rangle$ be  a  twist model such that $v(\varphi) = \pa{|\varphi|_\mathsf{M}^+, |\varphi|_\mathsf{M}^-, |\varphi|_\mathsf{M}^\circ}$. Then, recalling Definition~\ref{def:particion6valores} and Figure~\ref{fig:nuevas_regiones}:

\begin{itemize}

\item[-] $B_\varphi^\circ$ consists of all the  steadfast theists (i.e., they claim to have absolute certainty about God)

\item[-] $B_\varphi$ consists of all the theists without firm belief (i.e., they believe in God but they are not fully convinced)

\item[-] $C_\varphi$ consists of all the agnostics (i.e.: they have reasons both for and against belief in God, but remain uncertain)

\item[-] $D_\varphi$ consists of all the agnostics atheists (i.e.: they do not believe in God, but admit that they are not certain and could be wrong)

\item[-] $D_\varphi^\circ$ consists of all the  convinced atheists (i.e., they claim absolute certainty that God does not exist)

\item[-] $U_\varphi$ consists of all the persons who have no reasons for or against God's existence (for instance, ignostics or non-cognitivists).

\end{itemize}

Note that:

\begin{itemize}

\item[-]   $|\varphi|_\mathsf{M}^+=B_\varphi^\circ \cup B_\varphi \cup C_\varphi$  is the set of people who have reasons for believing in God

\item[-]   $|\varphi|_\mathsf{M}^-=D_\varphi^\circ \cup D_\varphi \cup C_\varphi$  is the set of people who have reasons not to believe in God

\item[-]  $|\varphi|_\mathsf{M}^\circ=B_\varphi^\circ \cup D_\varphi^\circ$  is the set of people with firm convictions about God's existence (both theists and atheists).

\end{itemize}

\end{example}

\section{Probability functions using \texorpdfstring{$\letkp$}{}}  \label{sect:prob-functions}

The previous sections were dedicated to describing the logical aspects of \letkp, with emphasis on the role of twist models. In this section, we show how to induce probability functions for reasoning based on \letkp. This will be done in two different, though equivalent, ways: first, a syntactical (or Carnapian) approach, based on axioms or principles that guide probability functions for formulas (which describe events) in \letkp; second, a semantical approach, in which probability functions over formulas in \letkp\ are induced by twist models. Each type of probability function will be defined in three variants: one-dimensional (assigning a probability value in the real interval $[0,1]$ to each formula); three-dimensional (assigning a value in $[0,1]^3$); and six-dimensional (assigning a value in $[0,1]^6$). From the semantical perspective, the latter two have a very natural interpretation, in terms of the (three or six) regions associated with each formula by a valuation in a twist model. In turn, the axiomatization of this kind of probability function is more complex, especially in the six-dimensional case, as will be shown below. As discussed in the Introduction, the present developments constitute a generalization of the proposal in~\cite{klein2021probabilities}, moving from the four-valued logic \fde\ to the six-valued logic $\letkp$. The main novelty is the introduction of a natural notion of three-dimensional probability functions based on twist structures, which makes the definition of six-valued probability functions more evident.

\subsection{Syntactical definition}
In \cite{weatherson2003probability}, Weatherson proposes a generalization of Kolmogorov's axioms for probability calculus to define probability functions for an arbitrary logic. When applied to classical logic in particular, the proposed axioms allow one to recover standard or classical probability functions.

Using the same axioms, probability functions have been studied for various non-classical logics. For instance, for intuitionistic logic in \cite{weatherson2003probability}, for the paraconsistent logics $Ci$ in \cite{bueno2016paraconsistent}, for $Cie$ in \cite{carnielli2017paraconsistent}, and more recently for $\fde$ in \cite{klein2021probabilities}, and for an extension of $\fde$ called $\LETF$ in \cite{rodrigues2021measuring}. Here we shall address the case of the logic $\letkp$.

\begin{definition}\label{def:funcion_de_probabilidad}
\textbf{(Probability function for $\letkp$)}
A probability function $p$ for the logic $\letkp$ is a function $p: \forsigma \longrightarrow [0,1]$ such that for any $\varphi, \psi \in \forsigma$:
\begin{description}
\item[(Ax1.1)] If $\vdash_{\letkp} \varphi$, then $p \pa{\varphi} = 1$.
\item[(Ax1.2)] If for every $\psi$ it holds that $\varphi \vdash_{\letkp} \psi$, then $p \pa{\varphi} = 0$.
\item[(Ax1.3)] $\varphi \vdash_{\letkp} \psi$ implies that $p \pa{\varphi} \leq p \pa{\psi}$.
\item[(Ax1.4)] $p \pa{\varphi} + p \pa{\psi} = p \pa{\varphi \land \psi} + p \pa{\varphi \lor \psi}$.
\end{description}
We denote by $\mathbb{P}_{1}$ the class of probability functions for $\letkp$. 
\end{definition}

\begin{proposition}\label{prop:propiedades-funcion-probabilidad}
Let $p$ be a probability function for $\letkp$.
Some immediate consequences of the previous definition are the following:
\begin{enumerate}
\item If $\varphi$ is a bottom particle in $\letkp$, then $p\pa{\varphi} = 0$.
\item If $\varphi \dashv \vdash_{\letkp} \psi$, then $p \pa{\varphi} = p \pa{\psi}$.\footnote{$\varphi \dashv \vdash_{L} \psi$ denotes that $\varphi \vdash_{L} \psi$ and $\psi \vdash_{L} \varphi$.}
\item $p \pa{\varphi \land \neg\varphi \land \circ\varphi} = 0$.
\item $p \pa{\circ\varphi} \leq p \pa{\varphi \lor \neg\varphi}$.
\item $p \pa{\varphi \lor {\sim}\varphi} = 1$.
\item $p \pa{\varphi \land {\sim}\varphi} = 0$.
\item $p\pa{{\sim}\varphi} = 1 - p \pa{\varphi}$.
\item If $\psi \vdash_{\letkp} \varphi$, then $p \pa{\varphi \land \psi} = p \pa{\psi}$.
\item In general, $p \pa{\varphi \land {\sim}\psi} = p \pa{\varphi} - p \pa{\varphi\land\psi}$, and in particular, if $\varphi$ is derivable from $\psi$ in $\letkp$, then $p \pa{\varphi \land {\sim}\psi} = p \pa{\varphi} - p \pa{\psi}$.
\item If $\psi_1 \vdash_{\letkp} \varphi$ and $\psi_2 \vdash_{\letkp} \varphi$, then $p \pa{\varphi \land {\sim}\pa{\psi_1 \lor \psi_2}} = p \pa{\varphi} - p\pa{\psi_1} - p\pa{\psi_2}+p\pa{\psi_1\land \psi_2}$.
\end{enumerate}
\end{proposition}

\begin{proof}
By employing the properties from Definition \ref{def:funcion_de_probabilidad} together with the following observations, each statement follows respectively:
\begin{enumerate}
\item As a consequence of $(Ax1.2)$.
\item As a consequence of $(Ax1.3)$.
\item As a consequence of rule $EXP^\circ$, $\varphi \land \neg\varphi \land \circ\varphi$ is a bottom particle.
\item As a consequence of rule $PEM^\circ$.
\item $\vdash_{\letkp} \varphi \lor {\sim}\varphi$ due to rule $\rightarrow CL$.
\item As a consequence of rule $\rightarrow E$, $\varphi \land {\sim}\varphi$ is a bottom particle.
\item Using items 5 and 6 of this proposition together with $(Ax1.4)$ of the definition.
\item In one hand we have that $\psi \vdash_{\letkp} \psi$, and by hypothesis $\psi \vdash_{\letkp} \varphi$, then $\psi \vdash_{\letkp} \varphi \land \psi$. On the other hand $\varphi \land \psi \vdash_{\letkp} \psi$. As a consequence ot this $\varphi \land \psi \dashv \vdash \psi$, therefore $p \pa{\varphi \land \psi} = p \pa{\psi}$.
\item It holds that $\varphi \dashv \vdash \pa{\varphi \land \psi} \lor \pa{\varphi \land {\sim}\psi}$, then:
$$\begin{aligned}[t]
p \pa{\varphi} 
& = 
p \pa{\pa{\varphi \land \psi} \lor \pa{\varphi \land {\sim}\psi} }\\
& = 
p \pa{\varphi \land \psi} + p \pa{\varphi \land {\sim}\psi}- p \pa{\pa{\varphi \land \psi} \land \pa{\varphi \land {\sim}\psi} }\\
& =
p \pa{\varphi \land \psi} + p \pa{\varphi \land {\sim}\psi}.
\end{aligned}$$     
Therefore $p \pa{\varphi \land {\sim}\psi} = p \pa{\varphi} - p \pa{\varphi\land\psi}$, and if $\psi \vdash_{\letkp} \varphi$, then $p \pa{\varphi \land \psi} = p \pa{\psi}$, so in this case $p \pa{\varphi \land {\sim}\psi} = p \pa{\varphi} - p \pa{\psi}$.
\item By hypothesis $\psi_1 \vdash_{\letkp} \varphi$ and $\psi_2 \vdash_{\letkp} \varphi$, then also holds that $\psi_1 \lor\psi_2\vdash_{\letkp} \varphi$. Using the previous item of this proposition as well as $Ax1.4$ we have the following:
$$\begin{aligned}[t]
p \pa{\varphi \land {\sim}\pa{\psi_1 \lor \psi_2}} 
& = 
p \pa{\varphi}-p\pa{\psi_1 \lor \psi_2} \\
& =
p \pa{\varphi} - p\pa{\psi_1} - p\pa{\psi_2}+p\pa{\psi_1\land \psi_2}.
\end{aligned}$$    
\end{enumerate}
\end{proof}

Although Definition \ref{def:funcion_de_probabilidad} allows us to have a notion of probability function using $\letkp$ logic, the fact is that in this case, unlike the classical case, having information about the probability of a statement does not provide information about the probability of its negation nor about the probability of its reliability. More detailed information enables a deeper understanding of the nature of the statement. To achieve a more comprehensive perspective, we can extend our function from one to three dimensions, thereby providing specific information for these two aspects: negation and reliability. This leads to the following definition.

\begin{definition}\label{def:funcion_de_probabilidad-twist}
\textbf{(Twist probability function for $\letkp$)}
A twist probability function $\widetilde{p}$ for $\letkp$ is a function $\widetilde{p}: \forsigma \longrightarrow [0,1]^3$, whose components, $\widetilde{p}_1,\ \widetilde{p}_2$ and $\widetilde{p}_3$, are simply functions with domain $\forsigma$ and codomain $[0,1]$. More explicitly $\widetilde{p}\pa{\varphi} = \pa{\widetilde{p}_1\pa{\varphi}, \widetilde{p}_2\pa{\varphi}, \widetilde{p}_3\pa{\varphi}}$, satisfying that for any $\varphi, \psi \in \forsigma$:
\begin{description}
    \item[(Ax3.1)] $\widetilde{p}_1\pa{\circ{\circ\varphi}} = 1$.
    \item[(Ax3.2)] $\widetilde{p}_1\pa{\varphi \land \neg \varphi \land \circ \varphi} = 0$.
    \item[(Ax3.3)] $\widetilde{p}_1\pa{\neg \varphi}  = \widetilde{p}_2\pa{\varphi}$.
    \item[(Ax3.4)] $\widetilde{p}_1\pa{\circ \varphi} = \widetilde{p}_3\pa{\varphi}$.
    \item[(Ax3.5)] $\varphi \vdash_{\letkp} \psi$ implies that $\widetilde{p}_1\pa{\varphi} \leq \widetilde{p}_1\pa{\psi}$.
    \item[(Ax3.6)] $\widetilde{p}_1\pa{\varphi} + \widetilde{p}_1\pa{\psi} = \widetilde{p}_1\pa{\varphi \land \psi} + \widetilde{p}_1\pa{ \varphi \lor \psi}$.
\end{description}
We denote by $\mathbb{P}_{3}$ the class of twist probability functions for $\letkp$.
\end{definition}

\begin{obs}\label{obs:propiedades de p-1}
From Definition \ref{def:funcion_de_probabilidad-twist}, it follows that $\vdash_{\letkp} \varphi$ implies $\widetilde{p}_1(\varphi) = 1$. Moreover, $\widetilde{p}_1$ satisfies all remaining axioms in Definition \ref{def:funcion_de_probabilidad} and therefore constitutes a probability function for $\letkp$. Consequently, $\widetilde{p}_1$ also fulfills all properties established in Proposition \ref{prop:propiedades-funcion-probabilidad}.
\end{obs}

\begin{definition}\label{def:funcion_probabilidad_6V}
\textbf{(6-valued probability function for $\letkp$)}
A 6-valued probability function for $\letkp$ is a function $\widehat{p}: \forsigma \longrightarrow [0,1]^6$, whose components, $b^\circ_\varphi$, $b^{~}_\varphi$, $c^{~}_\varphi$, $d^{~}_\varphi$, $d^\circ_\varphi$, and $u^{~}_\varphi$, are simply functions with domain $\forsigma$ and codomain $[0,1]$, more explicitly 
$\widehat{p}\pa{\varphi} = \pa{b^\circ_\varphi, b^{~}_\varphi,c^{~}_\varphi, d^{~}_\varphi, d^\circ_\varphi,  u^{~}_\varphi}$, satisfying that for any $\varphi, \psi \in \forsigma$ the following conditions hold:
 
\begin{description}
\item[(Ax6.1)] $b^\circ_{\circ\circ\varphi} = 1$.
\item[(Ax6.2)] $d^\circ_{\varphi \land \neg\varphi \land \circ\varphi} = 1$.
\item[(Ax6.3)] $b^\circ_{\varphi} + b^{~}_{\varphi} + c^{~}_{\varphi} + d^{~}_{\varphi} + d^\circ_{\varphi}  + u^{~}_{\varphi} = 1$.
\item[(Ax6.4)] $b^\circ_{\varphi \land \circ\varphi} + b^{~}_{\varphi \land \circ\varphi} + c^{~}_{\varphi \land \circ\varphi} = b^\circ_{\varphi}$.
\item[(Ax6.5)] $b^\circ_{\neg \varphi} = d^\circ_\varphi$, $b^{~}_{\neg \varphi} = d^{~}_\varphi$, and $c^{~}_{\neg \varphi} = c^{~}_\varphi$.
\item[(Ax6.6)] $b^\circ_{\varphi \land \neg\varphi} = 0$, $b^{~}_{\varphi \land \neg\varphi} = 0$, and $c^{~}_{\varphi \land \neg\varphi} = c^{~}_{\varphi}$.
\item[(Ax6.7)] $b^\circ_{\circ \varphi} = b^\circ_\varphi + d^\circ_\varphi$, $b^{~}_{\circ \varphi} = 0$, and $c^{~}_{\circ \varphi} = 0$.
\item[(Ax6.8)] If $\varphi \vdash_{\letkp} \psi$, then $b^\circ_\varphi + b^{~}_\varphi + c^{~}_\varphi \leq b^\circ_\psi + b^{~}_\psi + c^{~}_\psi$.
\item[(Ax6.9)] $\pa{b^\circ_\varphi + b^{~}_\varphi + c^{~}_\varphi} + \pa{b^\circ_\psi + b^{~}_\psi + c^{~}_\psi} = \pa{b^\circ_{\varphi \land \psi} + b^{~}_{\varphi \land \psi} + c^{~}_{\varphi \land \psi}} + \pa{b^\circ_{\varphi \lor \psi} + b^{~}_{\varphi \lor \psi} + c^{~}_{\varphi \lor \psi}}$.
\end{description}

We denote by $\mathbb{P}_{6}$ the class of 6-valued probability functions for $\letkp$.
\end{definition}

Up to this point we have defined three different probability functions for $\letkp$ of 1, 3 and 6 dimensions. At first glance they have similarities but in fact they are equivalent. In order to prove this it is necessary to define bijective translations among them. We first prove in Lemmas \ref{lema:AP_es_APtwist}-\ref{lema:AP6V-induce-AP} that the six translations to be proposed are well defined. 
\begin{lemma}\label{lema:AP_es_APtwist}
    Let $p$ be a probability function for $\letkp$. Then $\widetilde{p}: \forsigma \longrightarrow [0,1]^3$, given by $\widetilde{p}\pa{\varphi} = \pa{p \pa{\varphi}, p(\neg \varphi), p(\circ \varphi)}$, is a twist probability function for $\letkp$.
\end{lemma}

\begin{proof}
Since $p$ is a probability function for $\letkp$, it follows that for any $\psi \in \forsigma$, $p \pa{\psi} \in [0,1]$. Consequently, $p \pa{\varphi}, p \pa{\neg\varphi}, p \pa{\circ\varphi} \in [0,1]$. Let us verify that $\widetilde{p}$ satisfies the six points of Definition \ref{def:funcion_de_probabilidad-twist}.
\begin{enumerate}
    \item We have that $\vdash_{\letkp} {\circ{\circ\varphi}}$. Therefore, by Definition \ref{def:funcion_de_probabilidad}, $p\pa{{\circ{\circ\varphi}}} = 1$, and thus $\widetilde{p}_1\pa{{\circ{\circ\varphi}}} = 1$.
    \item We have that $\varphi \land \neg\varphi \land \circ \varphi \vdash_{\letkp} \psi$ for any $\psi$. Hence, from Definition \ref{def:funcion_de_probabilidad}, it follows that $p \pa{\varphi \land \neg\varphi \land \circ \varphi} = 0$, and therefore $\widetilde{p}_1\pa{\varphi \land \neg\varphi \land \circ \varphi} = 0$.
    \item We have that $\widetilde{p}\pa{\varphi} = \pa{p \pa{\varphi}, p \pa{\neg\varphi}, p \pa{\circ \varphi}}$ and also that $\widetilde{p}\pa{\neg\varphi} = \pa{p \pa{\neg\varphi}, p \pa{\neg\neg\varphi}, p \pa{\circ \neg\varphi}}$. Consequently, $\widetilde{p}_1\pa{\neg\varphi} = p \pa{\neg\varphi} = \widetilde{p}_2\pa{\varphi}$.
    \item We have that $\widetilde{p}\pa{\varphi} = \pa{p \pa{\varphi}, p \pa{\neg\varphi}, p \pa{\circ \varphi}}$ and also that $\widetilde{p}\pa{\circ\varphi} = \pa{p \pa{\circ\varphi}, p \pa{\neg\circ\varphi}, p \pa{ {\circ{\circ\varphi}}}}$. Therefore, $\widetilde{p}_1 \pa{\circ \varphi} = p \pa{\circ \varphi} = \widetilde{p}_3\pa{\varphi}$.
    \item Assume that $\varphi \vdash_{\letkp} \psi$. Then, by Definition \ref{def:funcion_de_probabilidad}, we have that $p \pa{\varphi} \leq p \pa{\psi}$, and thus $\widetilde{p}_1\pa{\varphi} \leq \widetilde{p}_1\pa{\psi}$.
    \item From Definition \ref{def:funcion_de_probabilidad}, we have that $p \pa{\varphi} + p \pa{\psi} = p \pa{\varphi \land \psi} + p \pa{\varphi \lor \psi}$, and hence $\widetilde{p}_1 \pa{\varphi} + \widetilde{p}_1\pa{\psi} = \widetilde{p}_1\pa{\varphi \land \psi} + \widetilde{p}_1\pa{\varphi \lor \psi}$.
\end{enumerate}
\end{proof}

\begin{lemma}\label{lema:APNETwist-induce-AP6V}
Let $\widetilde{p}$ be a twist probability function for $\letkp$. Then $\widehat{p}: \forsigma \longrightarrow [0,1]^6$ defined by $\widehat{p}\pa{\varphi}=\pa{\widehat{p}_1\pa{\varphi}, \widehat{p}_2\pa{\varphi},\ldots, \widehat{p}_6\pa{\varphi}}$, where:
\begin{center}
\begin{itemize}[nosep]
\begin{multicols}{2}
\item[] $\widehat{p}_1\pa{\varphi} = \widetilde{p}_1\pa{\varphi \land \circ \varphi}$
\item[] $\widehat{p}_3\pa{\varphi} = \widetilde{p}_1\pa{\varphi \land \neg \varphi}$
\item[] $\widehat{p}_5\pa{\varphi} = \widetilde{p}_1\pa{\neg \varphi \land \circ \varphi}$
\item[] $\widehat{p}_2\pa{\varphi} = \widetilde{p}_1\pa{\varphi} - \widetilde{p}_1\pa{\varphi \land \circ \varphi} - \widetilde{p}_1\pa{\varphi \land \neg \varphi}$
\item[] $\widehat{p}_4\pa{\varphi} = \widetilde{p}_1\pa{\neg \varphi} - \widetilde{p}_1\pa{\neg \varphi \land \circ \varphi} - \widetilde{p}_1\pa{\varphi \land \neg \varphi}$
\item[] $\widehat{p}_6\pa{\varphi} = 1 - \widetilde{p}_1\pa{\varphi \lor \neg \varphi}$
\end{multicols}
\end{itemize}
\end{center}
is a 6-valued probability function for $\letkp$. Hence $\widehat{p}\pa{\varphi} = \pa{b_\varphi^\circ,b_\varphi, c_\varphi, d_\varphi, d_\varphi^\circ, u_\varphi}$.
\end{lemma}

\begin{proof}
First, let us show that each component of $\widehat{p}$ takes values in $[0,1]$. Since $\widetilde{p}_1$ is the first component of $\widetilde{p}$, for any formula $\psi \in \forsigma$ we have $\widetilde{p}_1\pa{\psi}\in [0,1]$ and $-\widetilde{p}_1\pa{\psi}\in [-1,0]$. In particular:

\begin{center}
\begin{itemize}[nosep]
\begin{multicols}{2}
\item[] $\widehat{p}_1\pa{\varphi} =  \widetilde{p}_1\pa{\varphi \land \circ \varphi} \in [0,1] $
\item[] $\widehat{p}_5\pa{\varphi} =  \widetilde{p}_1\pa{\neg \varphi \land \circ \varphi}\in [0,1] $
\item[] $\widehat{p}_3\pa{\varphi} =  \widetilde{p}_1\pa{\varphi \land \neg \varphi}\in [0,1] $
\item[] $\widehat{p}_6\pa{\varphi} =  1 - \widetilde{p}_1\pa{\varphi \lor \neg \varphi}\in [0,1]$
\end{multicols}
\end{itemize}
\end{center}

On the other hand, using Observation \ref{obs:propiedades de p-1} and the facts that $\varphi\land\circ\varphi\vdash_{\letkp}\varphi$, $\varphi\land\neg\varphi\vdash_{\letkp}\varphi$, $\neg\varphi\land\circ\varphi\vdash_{\letkp}\neg\varphi$, and $\varphi\land\neg\varphi\vdash_{\letkp}\neg\varphi$, we have:

\begin{center}
\begin{itemize}[nosep]
\item[] $\widehat{p}_2\pa{\varphi} = \widetilde{p}_1\pa{\varphi} - \widetilde{p}_1\pa{\varphi \land \circ \varphi} - \widetilde{p}_1\pa{\varphi \land \neg \varphi}=\widetilde{p}_1\pa{\varphi \land \sim\pa{\pa{\varphi\land\circ\varphi}\lor\pa{\varphi\land\neg\varphi}}}\in [0,1]$
\item[] $\widehat{p}_4\pa{\varphi} = \widetilde{p}_1\pa{\neg\varphi} - \widetilde{p}_1\pa{\neg\varphi \land \circ \varphi} - \widetilde{p}_1\pa{\varphi \land \neg \varphi}=\widetilde{p}_1\pa{\neg\varphi \land \sim\pa{\pa{\neg\varphi\land\circ\varphi}\lor\pa{\varphi\land\neg\varphi}}}\in [0,1]$
\end{itemize}
\end{center}

Now we verify that $\widehat{p}$ satisfies the 9 conditions of Definition \ref{def:funcion_probabilidad_6V}.
\begin{enumerate}
\item By definition, $\widehat{p}_1\pa{\circ{\circ{\varphi}}}= \widetilde{p}_1\pa{\circ{\circ{\varphi}}\land \circ{\circ{\circ{\varphi}}}}$. Moreover, we have $\vdash_{\letkp} \circ{\circ{\varphi}}$ and $\vdash_{\letkp} \circ{\circ\pa{{\circ{\varphi}}}}$, hence $\vdash_{\letkp} \circ{\circ{\varphi}}\land \circ{\circ{\circ{\varphi}}}$. Using Observation \ref{obs:propiedades de p-1}, $\widetilde{p}_1\pa{\circ{\circ{\varphi}}\land \circ{\circ{\circ{\varphi}}}}=1$, so $\widehat{p}_1\pa{\circ{\circ{\varphi}}}=1$.

\item Let $\psi=\varphi\land \neg\varphi\land\circ\varphi$. Then $\widehat{p}_5\pa{\psi}= \widetilde{p}_1\pa{\neg\psi\land\circ\psi}$. Since $\vdash_{\letkp} \neg\psi\land\circ\psi$, again by Observation \ref{obs:propiedades de p-1} we have $\widetilde{p}_1\pa{\neg\psi\land\circ\psi}=1$, and thus $\widehat{p}_5\pa{\psi}=\widehat{p}_5\pa{\varphi\land \neg\varphi\land\circ\varphi}=1$.

\item We show that the sum of the six components equals 1.

$\begin{aligned}[t]
\widehat{p}_1\pa{\varphi}+\widehat{p}_2\pa{\varphi}+\widehat{p}_3\pa{\varphi}+\widehat{p}_4\pa{\varphi}+\widehat{p}_5\pa{\varphi}+\widehat{p}_6\pa{\varphi}
& =
\cancel{\widetilde{p}_1\pa{\varphi \land \circ \varphi}} +
\pa{
\widetilde{p}_1\pa{\varphi} -
\cancel{\widetilde{p}_1\pa{\varphi \land \circ \varphi}} -
\cancel{\widetilde{p}_1\pa{\varphi \land \neg \varphi}}
}+ 
\\
& \quad 
\cancel{\widetilde{p}_1\pa{\varphi \land \neg \varphi}} +
\pa{
\widetilde{p}_1\pa{\neg \varphi} -
\cancel{\widetilde{p}_1\pa{\neg \varphi \land \circ \varphi}} -
\widetilde{p}_1\pa{\varphi \land \neg \varphi}}+
\\
& \quad
\cancel{\widetilde{p}_1\pa{\neg \varphi \land \circ \varphi}} +
\pa{1 - \widetilde{p}_1(\varphi \lor \neg \varphi)}\\
& =
\widetilde{p}_1\pa{\varphi} +
\widetilde{p}_1\pa{\neg\varphi}-
\widetilde{p}_1\pa{\varphi \land \neg \varphi} +
1-\widetilde{p}_1\pa{\varphi \lor \neg \varphi}
\\
& =
\cancel{\widetilde{p}_1(\varphi \lor \neg \varphi)}+1-\cancel{\widetilde{p}_1(\varphi \lor \neg \varphi)}
\\
&=1.
\end{aligned}$

\item We have:

{\small
$\begin{aligned}[t] 
\widehat{p}_1\pa{\varphi\land\circ\varphi}+ \widehat{p}_2\pa{\varphi\land\circ\varphi}+ \widehat{p}_3\pa{\varphi\land\circ\varphi}
& = 
\widetilde{p}_1\pa{\pa{\varphi \land \circ \varphi} \land \circ\pa{\varphi \land \circ \varphi}} +\\
& 
\pa{
\widetilde{p}_1\pa{\varphi \land \circ \varphi}-
\widetilde{p}_1\pa{\pa{\varphi \land \circ \varphi} \land \circ \pa{\varphi \land \circ \varphi}}-
\widetilde{p}_1\pa{\pa{\varphi \land \circ \varphi} \land \neg \pa{\varphi \land \circ \varphi} }
} \\
& +
\widetilde{p}_1\pa{\pa{\varphi \land \circ \varphi} \land \neg\pa{\varphi \land \circ \varphi}}\\
& =
\cancel{\widetilde{p}_1\pa{\pa{\varphi \land \circ \varphi} \land \circ\pa{\varphi \land \circ \varphi}}}+\\
& 
\pa{
\widetilde{p}_1\pa{\varphi \land \circ \varphi} -
\cancel{\widetilde{p}_1\pa{\pa{\varphi \land \circ \varphi} \land \circ \pa{\varphi \land \circ \varphi}}} -
\cancel{\widetilde{p}_1\pa{\pa{\varphi \land \circ \varphi} \land \neg  \pa{\varphi \land \circ \varphi}}}
}\\
& +
\cancel{\widetilde{p}_1\pa{\pa{\varphi \land \circ \varphi} \land \neg \pa{\varphi \land \circ \varphi}}}\\
& =
\widetilde{p}_1\pa{\varphi \land \circ \varphi}\\
& =
\widehat{p}_1\pa{\varphi}.
\end{aligned}$}

\item Three claims must be proved: 
$\widehat{p}_1\pa{\neg\varphi}=\widehat{p}_5\pa{\varphi}$, $\widehat{p}_2\pa{\neg\varphi}=\widehat{p}_4\pa{\varphi}$, and $\widehat{p}_3\pa{\neg\varphi}=\widehat{p}_3\pa{\varphi}$.
\begin{itemize}
\item We have $\widehat{p}_1\pa{\neg\varphi}= \widetilde{p}_1\pa{\neg\varphi\land\circ\neg\varphi}$ and $\widehat{p}_5\pa{\varphi}= \widetilde{p}_1\pa{\neg\varphi\land\circ\varphi}$. Since $\circ\varphi \dashv \vdash_{\letkp} \circ \neg\varphi$, it follows that $\neg\varphi\land\circ\varphi \dashv \vdash_{\letkp} \neg\varphi\land\circ \neg\varphi$, hence $\widetilde{p}_1\pa{\neg\varphi\land\circ\neg\varphi}= \widetilde{p}_1\pa{\neg\varphi\land\circ\varphi}$ and thus $\widehat{p}_1\pa{\neg\varphi}=\widehat{p}_5\pa{\varphi}$.

\item We have $\widehat{p}_2\pa{\neg\varphi}= \widetilde{p}_1\pa{\neg\varphi} - \widetilde{p}_1\pa{\neg\varphi \land \circ \neg\varphi} - \widetilde{p}_1\pa{\neg\varphi \land \neg \neg\varphi}$ and $\widehat{p}_4\pa{\varphi} = \widetilde{p}_1\pa{\neg \varphi} - \widetilde{p}_1\pa{\neg \varphi \land \circ \varphi} - \widetilde{p}_1\pa{\varphi \land \neg \varphi}$. From the previous point, $\widetilde{p}_1\pa{\neg\varphi\land\circ\neg\varphi}= \widetilde{p}_1\pa{\neg\varphi\land\circ\varphi}$. Moreover, $\neg\neg\varphi \dashv \vdash_{\letkp}\varphi$, hence $\neg\varphi\land\neg\neg\varphi \dashv \vdash_{\letkp} \neg\varphi\land\varphi$, so $\widetilde{p}_1\pa{\neg\varphi\land\neg\neg\varphi}= \widetilde{p}_1\pa{\neg\varphi\land\varphi}$. Therefore $\widehat{p}_2\pa{\neg\varphi}=\widehat{p}_4\pa{\varphi}$.

\item We have $\widehat{p}_3\pa{\neg\varphi}=\widetilde{p}_1\pa{\neg\varphi \land \neg \neg\varphi}$ and $\widehat{p}_3\pa{\varphi} =\widetilde{p}_1\pa{\varphi \land \neg \varphi}$. From the previous point, $\widetilde{p}_1\pa{\neg\varphi\land\varphi}$ $=\widetilde{p}_1\pa{\neg\varphi \land\neg\neg\varphi}$, hence $\widehat{p}_3\pa{\neg\varphi}=\widehat{p}_3\pa{\varphi}$.
\end{itemize}

\item Three claims must be proved: $\widehat{p}_1\pa{\varphi\land\neg\varphi}=0$, $\widehat{p}_2\pa{\varphi\land\neg\varphi}=0$, and $\widehat{p}_3\pa{\varphi\land\neg\varphi}=\widehat{p}_3\pa{\varphi}$.

\begin{itemize}

\item As a consequence of the rule $\cor{{T}{\land}E_l}$ we have $\pa{\varphi\land\neg\varphi}^T \vdash_{\letkp} \varphi^T$, this written without abbreviations becomes $\pa{\varphi\land\neg\varphi}\land\circ\pa{\varphi\land\neg\varphi} \vdash_{\letkp} \varphi\land\circ\varphi$. Starting from this we can prove $\pa{\varphi\land\neg\varphi}\land\circ\pa{\varphi\land\neg\varphi} \vdash_{\letkp} \varphi\land\neg\varphi\land\circ\varphi$. As a consequence of Definition \ref{def:funcion_de_probabilidad-twist} it follows that $\widetilde{p}_1\pa{\pa{\varphi\land\neg\varphi}\land\circ\pa{\varphi\land\neg\varphi}} \leq \widetilde{p}_1\pa{\varphi\land\neg\varphi\land\circ\varphi} = 0$. But $\widehat{p}_1\pa{\varphi\land\neg\varphi} = \widetilde{p}_1\pa{\pa{\varphi\land\neg\varphi}\land\circ\pa{\varphi\land\neg\varphi}}$, then $\widehat{p}_1\pa{\varphi\land\neg\varphi} =0$.

\item Using De Morgan's laws, double‑negation elimination, and simplifying, one can see that $\pa{\varphi\land\neg\varphi}\land\neg\pa{\varphi\land\neg\varphi}$ $\dashv \vdash_{\letkp}\varphi\land\neg\varphi$, hence $\widetilde{p}_1\pa{\pa{\varphi\land\neg\varphi}\land\neg\pa{\varphi\land\neg\varphi}}=\widetilde{p}_1\pa{\varphi\land\neg\varphi}$. Then

$\begin{aligned}[t]
\widehat{p}_2\pa{\varphi\land\neg\varphi}
& = 
\widetilde{p}_1\pa{\varphi\land\neg\varphi}-
\widetilde{p}_1\pa{\pa{\varphi\land\neg\varphi}\land\circ\pa{\varphi\land\neg\varphi}}-
\widetilde{p}_1\pa{\pa{\varphi\land\neg\varphi}\land\neg\pa{\varphi\land\neg\varphi}}\\
& = 
\cancel{\widetilde{p}_1\pa{\varphi\land\neg\varphi}}-\widehat{p}_1\pa{\varphi\land\neg\varphi}-\cancel{\widetilde{p}_1\pa{\varphi\land\neg\varphi}}\\
& =
-\widehat{p}_1\pa{\varphi\land\neg\varphi}\\
& =
0.
\end{aligned}$

\item Since $\widetilde{p}_1\pa{\pa{\varphi\land\neg\varphi}\land\neg\pa{\varphi\land\neg\varphi}}=\widetilde{p}_1\pa{\varphi\land\neg\varphi}$, hence:

$\begin{aligned}[t]
\widehat{p}_3\pa{\varphi\land\neg\varphi}
& = 
\widetilde{p}_1\pa{\pa{\varphi\land\neg\varphi}\land\neg\pa{\varphi\land\neg\varphi}}\\
& = 
\widetilde{p}_1\pa{\varphi\land\neg\varphi}\\
& =
\widehat{p}_3\pa{\varphi}.
\end{aligned}$
\end{itemize}

\item Three claims must be proved: $\widehat{p}_1\pa{\circ\varphi}= \widehat{p}_1\pa{\varphi}+\widehat{p}_5\pa{\varphi}$, $\widehat{p}_2\pa{\circ\varphi}=0$, and $\widehat{p}_3\pa{\circ\varphi}=0$.
\begin{itemize}
\item Since $\vdash_{\letkp}\circ{\circ\varphi}$ and $\vdash_{\letkp}\circ\varphi \lor\circ{\circ\varphi}$, we have $\widetilde{p}_1\pa{\circ{\circ\varphi}}=1$ and $\widetilde{p}_1\pa{\circ\varphi \lor\circ{\circ\varphi}}=1$, so:

$\begin{aligned}[t]
\widehat{p}_1\pa{\circ\varphi}
& =
\widetilde{p}_1\pa{\circ\varphi \land \circ{\circ\varphi}}\\
& =
\widetilde{p}_1\pa{\circ\varphi}+\widetilde{p}_1\pa{\circ{\circ\varphi}}-\widetilde{p}_1(\circ\varphi \lor\circ{\circ\varphi})\\
& =
\widetilde{p}_1\pa{\circ\varphi}+1-1\\
& =
\widetilde{p}_1\pa{\circ\varphi}.
\end{aligned}$

It can be shown that $\varphi\lor\neg\varphi \dashv \vdash_{\letkp} \varphi\lor\neg\varphi\lor{\circ{\varphi}}$, hence $\widetilde{p}_1(\varphi\lor\neg\varphi)= \widetilde{p}_1(\varphi\lor\neg\varphi\lor{\circ{\varphi}})$. 
Therefore:

$\begin{aligned}[t]
\widehat{p}_1\pa{\varphi}+\widehat{p}_5\pa{\varphi}
&=
\widetilde{p}_1\pa{{\varphi}\land {\circ{\varphi}}}+\widetilde{p}_1\pa{\neg\varphi \land {\circ\varphi}}\\
&=
\widetilde{p}_1\pa{\pa{{\varphi}\land {\circ{\varphi}}}\land \pa{\neg\varphi \land {\circ\varphi}}} +
\widetilde{p}_1\pa{\pa{{\varphi}\land {\circ{\varphi}}}\lor \pa{\neg\varphi \land {\circ\varphi}}}\\
&=
\widetilde{p}_1\pa{{\varphi}\land\neg\varphi \land {\circ\varphi}}+\widetilde{p}_1\pa{{\circ\varphi}\land \pa{\varphi \lor {\neg\varphi}}}\\
&=
0+ \widetilde{p}_1\pa{{\circ\varphi}}+\cancel{\widetilde{p}_1\pa{\varphi\lor\neg\varphi}}-\cancel{\widetilde{p}_1\pa{{\circ\varphi}\lor \pa{\varphi \lor {\neg\varphi}}}}\\
&=
\widetilde{p}_1\pa{{\circ\varphi}}.
\end{aligned}$

Thus $\widehat{p}_1\pa{\circ\varphi}= \widetilde{p}_1\pa{\circ\varphi}=\widehat{p}_1\pa{\varphi}+\widehat{p}_5\pa{\varphi}$.

\item In the first part of this item we implicitly proved that $\widetilde{p}_1\pa{\circ\varphi \land \circ{\circ\varphi}}= \widetilde{p}_1\pa{\circ\varphi}$. Moreover, it can be shown that $\circ\varphi \land \neg\circ\varphi \dashv \vdash_{\letkp}\circ\varphi \land \neg{\circ}\varphi\land\circ{\circ\varphi}$; note that the right-hand side of this equivalence is a bottom particle, hence it must be that $\widetilde{p}_1\pa{\circ\varphi \land \neg \circ\varphi} = 0$. Therefore:

$\begin{aligned}[t]
\widehat{p}_2\pa{\circ\varphi}
&=
\widetilde{p}_1\pa{\circ\varphi} - \widetilde{p}_1\pa{\circ\varphi \land \circ{\circ\varphi}} - \widetilde{p}_1\pa{\circ\varphi \land \neg \circ\varphi}\\
&=
\cancel{\widetilde{p}_1\pa{\circ\varphi}} - \cancel{\widetilde{p}_1\pa{\circ\varphi}} - 0\\
&=
0.
\end{aligned}$

\item Finally, $\widehat{p}_3\pa{\circ\varphi} = \widetilde{p}_1\pa{\circ\varphi \land \neg \circ\varphi}= \widetilde{p}_1\pa{\circ\varphi \land \neg \circ\varphi\land\circ{\circ\varphi}}=0$.
\end{itemize}

\item If $\varphi\vdash_{\letkp} \psi$, then $\widetilde{p}_1\pa{\varphi}\leq \widetilde{p}_1\pa{\psi}$. Moreover, for any formula $\varphi$ we have:
$$
\widehat{p}_1\pa{\varphi}+ \widehat{p}_2\pa{\varphi}+ \widehat{p}_3\pa{\varphi}=\pa{\cancel{\widetilde{p}_1\pa{\varphi \land \circ \varphi}}}+ \pa{\widetilde{p}_1\pa{\varphi} - \cancel{\widetilde{p}_1(\varphi \land \circ \varphi)} - \cancel{\widetilde{p}_1(\varphi \land \neg \varphi)}}+\pa{\cancel{\widetilde{p}_1(\varphi \land \neg \varphi)}}=\widetilde{p}_1\pa{\varphi},
$$
so $\widehat{p}_1\pa{\varphi}+ \widehat{p}_2\pa{\varphi}+ \widehat{p}_3\pa{\varphi}= \widetilde{p}_1\pa{\varphi}\leq \widetilde{p}_1\pa{\psi}=\widehat{p}_1\pa{\psi}+ \widehat{p}_2\pa{\psi}+ \widehat{p}_3\pa{\psi}$.

\item As implicitly shown in the previous point, for any formula $\varphi$, we have that $\widehat{p}_1\pa{\varphi}+ \widehat{p}_2\pa{\varphi}+ \widehat{p}_3\pa{\varphi}=\widetilde{p}_1\pa{\varphi}$. By hypothesis, $\widetilde{p}_1\pa{\varphi} + \widetilde{p}_1\pa{\psi} = \widetilde{p}_1\pa{\varphi \land \psi} + \widetilde{p}_1\pa{ \varphi \lor \psi}$. Hence it follows immediately that:

$\pa{\widehat{p}_1\pa{\varphi}+ \widehat{p}_2\pa{\varphi}+ \widehat{p}_3\pa{\varphi}} + \pa{\widehat{p}_1\pa{\psi}+ \widehat{p}_2\pa{\psi}+ \widehat{p}_3\pa{\psi}} =\\
\pa{\widehat{p}_1\pa{\varphi\land\psi} + \widehat{p}_2\pa{\varphi\land\psi}+ \widehat{p}_3\pa{\varphi\land\psi}} + \pa{\widehat{p}_1\pa{\varphi\lor\psi}+ \widehat{p}_2\pa{\varphi\lor\psi}+ \widehat{p}_3\pa{\varphi\lor\psi}}$.
\end{enumerate}
\end{proof}

\begin{lemma}\label{lema:APNE-induce-AP6V}
Let $p$ be a probability function for $\letkp$. Then $\widehat{p}: \forsigma \longrightarrow [0,1]^6$ defined by $\widehat{p}\pa{\varphi}= \pa{\widehat{p}_1\pa{\varphi}, \widehat{p}_2\pa{\varphi},\ldots, \widehat{p}_6\pa{\varphi}}$, where:\\

$
\begin{array}{ll}
\widehat{p}_1\pa{\varphi} = p\pa{\varphi \land \circ \varphi} & \hspace{2cm} \widehat{p}_2\pa{\varphi} = p\pa{\varphi} - p\pa{\varphi \land \circ \varphi} - p\pa{\varphi \land \neg \varphi}\\
\widehat{p}_3\pa{\varphi} = p\pa{\varphi \land \neg \varphi} & \hspace{2cm} \widehat{p}_4\pa{\varphi} = p\pa{\neg \varphi} - p\pa{\neg \varphi \land \circ \varphi} - p\pa{\varphi \land \neg \varphi}\\
\widehat{p}_5\pa{\varphi} = p\pa{\neg \varphi \land \circ \varphi} & \hspace{2cm} \widehat{p}_6\pa{\varphi} = 1 - p\pa{\varphi \lor \neg \varphi}
\end{array}$

\ \\

\noindent
is a 6-valued probability function for $\letkp$. Hence $\widehat{p}\pa{\varphi} = \pa{b_\varphi^\circ,b_\varphi,c_\varphi,d_\varphi, d_\varphi^\circ, u_\varphi}$.
\end{lemma}

\begin{proof}
    Direct consequence of Lemmas \ref{lema:AP_es_APtwist} and \ref{lema:APNETwist-induce-AP6V}.
\end{proof}

\begin{lemma}\label{lema:AP6V-induce-APNETwist}
Let $\widehat{p}$ be a 6-valued probability function for $\letkp$. If $\widehat{p}\pa{\varphi} = \pa{b_\varphi^\circ,b_\varphi,c_\varphi,d_\varphi, d_\varphi^\circ, u_\varphi}$, then the function $\widetilde{p}:\forsigma \longrightarrow [0,1]^3$ satisfying $\widetilde{p}\pa{\varphi}=\pa{\widetilde{p}_1\pa{\varphi},\widetilde{p}_2\pa{\varphi},\widetilde{p}_3\pa{\varphi}}$, where:
\begin{itemize}
    \item[] $\widetilde{p}_1\pa{\varphi}=b_\varphi^\circ + b_\varphi + c_\varphi$
    \item[] $\widetilde{p}_2\pa{\varphi}=d_\varphi^\circ + d_\varphi + c_\varphi$
    \item[] $\widetilde{p}_3\pa{\varphi}=b_\varphi^\circ + d_\varphi^\circ$
\end{itemize}
is a twist probability function for $\letkp$.
\end{lemma}

\begin{proof}
Since $\widehat{p}$ is a 6-valued probability function for $\letkp$, for any $\varphi\in \forsigma$ we have $b_\varphi^\circ,$ $b_\varphi,$ $c_\varphi,$ $d_\varphi,$ $d_\varphi^\circ,$  $u_\varphi\in[0,1]$ and $b_\varphi^\circ + b_\varphi + c_\varphi + d_\varphi + d_\varphi^\circ+ u_\varphi=1$. From this we can conclude that $0\leq b_\varphi^\circ + b_\varphi +c_\varphi\leq 1$, $0\leq d_\varphi^\circ + d_\varphi +c_\varphi\leq 1$, and $0\leq b_\varphi^\circ + d_\varphi^\circ\leq 1$. Therefore $\widetilde{p}_1\pa{\varphi},\widetilde{p}_2\pa{\varphi},\widetilde{p}_3\pa{\varphi}\in[0,1]$. We now verify that $\widetilde{p}$ satisfies the 6 conditions of Definition \ref{def:funcion_de_probabilidad-twist}. 
\begin{enumerate}
\item By Definition \ref{def:funcion_probabilidad_6V}, we have $b_{\circ{\circ \varphi}}^\circ=1$. Moreover, each of the values $b_{\circ{\circ \varphi}}$, $c_{\circ{\circ \varphi}}$, $d_{\circ{\circ \varphi}}$, $d_{\circ{\circ \varphi}}^\circ$,  $u_{\circ{\circ \varphi}}$ is non-negative and the sum of all of the six components equals 1. It follows that $b_{\circ{\circ \varphi}}=0$, $c_{\circ{\circ \varphi}}=0$, and hence $\widetilde{p}_1\pa{{\circ{\circ \varphi}}}=b_{\circ{\circ \varphi}}^\circ + b_{\circ{\circ \varphi}} + c_{\circ{\circ \varphi}}=1$.

\item By Definition \ref{def:funcion_probabilidad_6V}, if $\psi = \varphi \land \neg \varphi \land \circ \varphi$, then $d_{\psi}^\circ=1$. Also, each of the values $b_{\psi}^\circ, b_{\psi}, c_{\psi}, d_{\psi}, u_{\psi}$ is non-negative and the sum of all the six components equals 1. Thus $b_{\psi}^\circ=0$, $b_{\psi}=0$, $c_{\psi}=0$, and therefore $\widetilde{p}_1\pa{\psi}=\widetilde{p}_1(\varphi \land \neg \varphi \land \circ \varphi)=0$.

\item We have $\widetilde{p}_1\pa{\neg \varphi} = b_{\neg\varphi}^\circ + b_{\neg\varphi} + c_{\neg\varphi}= d_{\varphi}^\circ + d_{\varphi} + c_{\varphi}= \widetilde{p}_2\pa{\varphi}$.

\item We have $\widetilde{p}_1\pa{\circ \varphi} = b_{\circ\varphi}^\circ + b_{\circ\varphi} + c_{\circ\varphi} = b_{\varphi}^\circ + d_{\varphi}^\circ + b_{\circ\varphi} + c_{\circ\varphi} = b_{\varphi}^\circ + d_{\varphi}^\circ = \widetilde{p}_3\pa{\varphi}$.

\item Assume that $\varphi\vdash_{\letkp} \psi$. Then by Definition \ref{def:funcion_probabilidad_6V} we have $b_\varphi^\circ + b_\varphi +c_\varphi \leq b_\psi^\circ + b_\psi +c_\psi$, and hence $\widetilde{p}_1\pa{\varphi} \leq \widetilde{p}_1\pa{\psi}$.

\item By Definition \ref{def:funcion_probabilidad_6V} we have:
\begin{align*}
    \pa{b_\varphi^\circ + b_\varphi +c_\varphi} + \pa{b_\psi^\circ + b_\psi +c_\psi}  
    &= \pa{b_{\varphi \land \psi}^\circ + b_{\varphi \land \psi} +c_{\varphi \land \psi}} + \pa{b_{\varphi \lor \psi}^\circ + b_{\varphi \lor \psi} +c_{\varphi \lor \psi}},
\end{align*} 
that is, $\widetilde{p}_1\pa{\varphi} + \widetilde{p}_1\pa{\psi} = \widetilde{p}_1\pa{\varphi \land \psi} + \widetilde{p}_1\pa{\varphi \lor \psi}$.
\end{enumerate}
\end{proof}

\begin{lemma}\label{lema:APtwist_es_AP}
    Let $\widetilde{p}$ be a twist probability function for $\letkp$. Then $p: \forsigma \longrightarrow [0,1]$ defined by $p \pa{\varphi}=\widetilde{p}_1\pa{\varphi}$ is a probability function for $\letkp$.
\end{lemma}

\begin{proof}
Since $\widetilde{p}$ is a twist probability function for $\letkp$, for any $\varphi\in \forsigma$ we have $\widetilde{p}_1\pa{\varphi}\in[0,1]$ and therefore $p \pa{\varphi}\in[0,1]$. We now verify that $p$ satisfies the 4 conditions of Definition \ref{def:funcion_de_probabilidad}.
\begin{enumerate}
    \item Assume that $\vdash_{\letkp}\varphi$. Then we have ${\circ{\circ\varphi}}\vdash_{\letkp}\varphi$. By Definition \ref{def:funcion_de_probabilidad-twist}, $\widetilde{p}_1\pa{{\circ{\circ\varphi}}}=1$ and $\widetilde{p}_1\pa{{\circ{\circ\varphi}}}\leq \widetilde{p}_1\pa{\varphi}$. Hence $\widetilde{p}_1\pa{\varphi}=1$ and $p \pa{\varphi}=1$.
    
    \item Assume that $\varphi\vdash_{\letkp} \psi$ for every $\psi$. In particular, $\varphi\vdash_{\letkp} \varphi \land \neg\varphi \land \circ \varphi$. By Definition \ref{def:funcion_de_probabilidad-twist}, $\widetilde{p}_1\pa{\varphi} \leq \widetilde{p}_1\pa{\varphi \land \neg\varphi \land \circ \varphi}$ and $\widetilde{p}_1\pa{\varphi \land \neg\varphi \land \circ \varphi}=0$. Therefore $\widetilde{p}_1\pa{\varphi}=0$ and thus $p \pa{\varphi}=0$.
    
    \item Assume that $\varphi\vdash_{\letkp} \psi$. Then by Definition \ref{def:funcion_de_probabilidad-twist}, $\widetilde{p}_1\pa{\varphi} \leq \widetilde{p}_1\pa{\psi}$ and hence $p \pa{\varphi} \leq p \pa{\psi}$.
    
    \item From Definition \ref{def:funcion_de_probabilidad-twist}, $\widetilde{p}_1\pa{\varphi} + \widetilde{p}_1\pa{\psi} = \widetilde{p}_1\pa{\varphi \land \psi} + \widetilde{p}_1\pa{\varphi \lor \psi}$. Consequently, $p \pa{\varphi} + p \pa{\psi} = p\pa{\varphi \land \psi} + p\pa{\varphi \lor \psi}$.
\end{enumerate}
\end{proof}

\begin{lemma}\label{lema:AP6V-induce-AP}
Let $\widehat{p}$ be a 6-valued probability function for $\letkp$. If $\widehat{p}\pa{\varphi}=\pa{b_\varphi^\circ,b_\varphi,c_\varphi,d_\varphi, d_\varphi^\circ, u_\varphi}$, then the function $p:\forsigma \longrightarrow [0,1]$ defined by:
\[
p \pa{\varphi}=b_\varphi^\circ + b_\varphi + c_\varphi
\]
is a probability function for $\letkp$.
\end{lemma}

\begin{proof}
    Direct consequence of Lemmas \ref{lema:AP6V-induce-APNETwist} and \ref{lema:APtwist_es_AP}.
\end{proof}

The preceding lemmas allow us to establish the following translations.

\begin{corolario}\label{coro:traduccionNST-NS}
    Let $\tratresuno:\mathbb{P}_{3} \longrightarrow \mathbb{P}_{1}$ be defined by $[\tratresuno \pa{\widetilde{p}}]\pa{\varphi}=\widetilde{p}_1\pa{\varphi}$ and $\traunotres:\mathbb{P}_{1} \longrightarrow \mathbb{P}_{3}$ by $[\traunotres\pa{p}]\pa{\varphi}=\pa{p \pa{\varphi},p \pa{\neg\varphi},p \pa{\circ\varphi}}$. Then:
    \begin{enumerate}
        \item $\traunotres\pa{\tratresuno\pa{\widetilde{p}}}=\widetilde{p}$.
        \item $\tratresuno\pa{\traunotres(p)}=p$.
    \end{enumerate}
    That is, $\tratresuno$ and $\traunotres$ are inverses of each other, and their composition (depending on the order) yields the identity on $\mathbb{P}_{3}$ or on $\mathbb{P}_{1}$.
\end{corolario}

\begin{proof}
By Lemmas \ref{lema:APtwist_es_AP} and \ref{lema:AP_es_APtwist}, the functions $\tratresuno$ and $\traunotres$ are well-defined. Moreover:

\begin{enumerate}
\item Given $\widetilde{p}\in \mathbb{P}_{3}$, let $p=\tratresuno(\widetilde{p})$. Then $p \pa{\varphi}=\widetilde{p}_1\pa{\varphi}$, and thus:
\begin{align*}
    \traunotres(p)\pa{\varphi} &= \pa{ p \pa{\varphi}, p(\neg \varphi), p(\circ \varphi) }\\
    &= \pa{ \widetilde{p}_1\pa{\varphi}, \widetilde{p}_1(\neg \varphi),\widetilde{p}_1(\circ \varphi) }\\
    &= \pa{ \widetilde{p}_1\pa{\varphi}, \widetilde{p}_2\pa{\varphi}, \widetilde{p}_3\pa{\varphi} }\\
    &= \widetilde{p}\pa{\varphi}.
\end{align*}
Therefore $\traunotres\pa{\tratresuno\pa{\widetilde{p}}}=\widetilde{p}$.

\item Given $p\in\mathbb{P}_{1}$, let $\widetilde{p}=\traunotres\pa{p}$. Then $\widetilde{p}\pa{\varphi} = \pa{p \pa{\varphi}, p\pa{\neg \varphi}, p\pa{\circ \varphi}}$, and so:
\begin{align*}
    \tratresuno\pa{\widetilde{p}}\pa{\varphi} &= \widetilde{p}_1\pa{\varphi}=p \pa{\varphi}.
\end{align*}
Therefore $\tratresuno\pa{\traunotres\pa{p}}=p$.
\end{enumerate}
\end{proof}

\begin{corolario}\label{coro:traduccionNST-6}
Define $\traseistres:\mathbb{P}_{6} \longrightarrow \mathbb{P}_{3}$ as follows: for $\widehat{p}\in \mathbb{P}_6$ with $\widehat{p}\pa{\varphi} = \pa{b_\varphi^\circ,b_\varphi,c_\varphi,d_\varphi, d_\varphi^\circ, u_\varphi}$, set
$$
[\traseistres\pa{\widehat{p}}]\pa{\varphi}=\pa{b_\varphi^\circ + b_\varphi +c_\varphi,\; d_\varphi^\circ + d_\varphi +c_\varphi,\; b_\varphi^\circ + d_\varphi^\circ}.
$$
Define $\tratresseis: \mathbb{P}_{3} \longrightarrow \mathbb{P}_{6}$ by $[\tratresseis\pa{\widetilde{p}}]\pa{\varphi}=\widehat{p}\pa{\varphi} = \pa{\widehat{p}_1\pa{\varphi}, \widehat{p}_2\pa{\varphi},\ldots, \widehat{p}_6\pa{\varphi}}$, where:
\begin{center}
\begin{itemize}[nosep]
\begin{multicols}{2}
\item[] $\widehat{p}_1\pa{\varphi} = \widetilde{p}_1\pa{\varphi \land \circ \varphi}$
\item[] $\widehat{p}_3\pa{\varphi} = \widetilde{p}_1\pa{\varphi \land \neg \varphi}$
\item[] $\widehat{p}_5\pa{\varphi} = \widetilde{p}_1\pa{\neg \varphi \land \circ \varphi}$
\item[] $\widehat{p}_2\pa{\varphi} = \widetilde{p}_1\pa{\varphi} - \widetilde{p}_1\pa{\varphi \land \circ \varphi} - \widetilde{p}_1\pa{\varphi \land \neg \varphi}$
\item[] $\widehat{p}_4\pa{\varphi} = \widetilde{p}_1\pa{\neg \varphi} - \widetilde{p}_1\pa{\neg \varphi \land \circ \varphi} - \widetilde{p}_1\pa{\varphi \land \neg \varphi}$
\item[] $\widehat{p}_6\pa{\varphi} = 1 - \widetilde{p}_1\pa{\varphi \lor \neg \varphi}$
\end{multicols}
\end{itemize}
\end{center}

Then:
\begin{enumerate}
    \item $\tratresseis\pa{\traseistres\pa{\widehat{p}}}   =\widehat{p}$.
    \item $\traseistres\pa{\tratresseis\pa{\widetilde{p}}} =\widetilde{p}$.
\end{enumerate}
That is, $\traseistres$ and $\tratresseis$ are inverses of each other, and their composition (depending on the order) yields the identity on $\mathbb{P}_{6}$ or on $\mathbb{P}_{3}$.
\end{corolario}

\begin{proof}
By Lemmas \ref{lema:AP6V-induce-APNETwist} and \ref{lema:APNETwist-induce-AP6V}, the functions $\traseistres$ and $\tratresseis$ are well-defined. Moreover:

\begin{enumerate}
\item Given $\widehat{p}\in \mathbb{P}_{6}$, let $\widetilde{p}=\traseistres\pa{\widehat{p}}$. Then $\widetilde{p}\pa{\varphi} = \pa{b_\varphi^\circ + b_\varphi +c_\varphi,\; d_\varphi^\circ + d_\varphi +c_\varphi,\; b_\varphi^\circ + d_\varphi^\circ}$ and the following equations hold:

\begin{enumerate}
\item By Ax6.4; $b_{\varphi\land\circ\varphi}^\circ + b_{\varphi\land\circ\varphi} +c_{\varphi\land\circ\varphi} = b_\varphi^\circ $.
\item By Ax6.6: $b^\circ_{\varphi \land \neg\varphi} = 0$, $b^{~}_{\varphi \land \neg\varphi} = 0$, and $c^{~}_{\varphi \land \neg\varphi} = c^{~}_{\varphi}$, therefore:
$$b_\varphi^\circ + b_\varphi +c_\varphi - \pa{b_{\varphi\land\circ\varphi}^\circ + b_{\varphi\land\circ\varphi} +c_{\varphi\land\circ\varphi}} - \pa{b_{\varphi\land\neg\varphi}^\circ + b_{\varphi\land\neg\varphi} +c_{\varphi\land\neg\varphi}}
= b_\varphi^\circ + b_\varphi +c_\varphi - b_\varphi^\circ - c_\varphi 
= b_\varphi.$$

\item Since $\neg\varphi\land\circ\varphi\dashv \vdash_{\letkp} \neg\varphi\land\circ\neg\varphi$: $
b_{\neg\varphi\land\circ\varphi}^\circ + b_{\neg\varphi\land\circ\varphi} +c_{\neg\varphi\land\circ\varphi}= b_{\neg\varphi\land\circ\neg\varphi}^\circ + b_{\neg\varphi\land\circ\neg\varphi} +c_{\neg\varphi\land\circ\neg\varphi}=b_{\neg\varphi}^\circ$, then it holds:

$\begin{aligned}[t]
b_{\neg\varphi}^\circ + b_{\neg\varphi} + c_{\neg\varphi} - \pa{b_{\neg\varphi\land\circ\varphi}^\circ + b_{\neg\varphi\land\circ\varphi} +c_{\neg\varphi\land\circ\varphi}} - \pa{b_{\varphi\land\neg\varphi}^\circ + b_{\varphi\land\neg\varphi} +c_{\varphi\land\neg\varphi}}  
&=
b_{\neg\varphi}^\circ + b_{\neg\varphi} + c_{\neg\varphi} - b_{\neg\varphi}^\circ - c_\varphi\\
&= b_{\neg\varphi} + c_{\neg\varphi} - c_\varphi\\
&= d_\varphi + c_\varphi - c_\varphi\\
&= d_\varphi
\end{aligned}$

\item  $b_{\neg\varphi\land\circ\neg\varphi}^\circ + b_{\neg\varphi\land\circ\neg\varphi} +c_{\neg\varphi\land\circ\neg\varphi}=b_{\neg\varphi}^\circ=d_{\varphi}^\circ
$.

\item By Ax6.6;  $b_{\varphi\land\neg\varphi}^\circ + b_{\varphi\land\neg\varphi} + c_{\varphi\land\neg\varphi} = c^{~}_{\varphi}$.

\item Finally:

$\begin{aligned}[t]
1- \pa{b_{\varphi\lor\neg\varphi}^\circ + b_{\varphi\lor\neg\varphi} +c_{\varphi\lor\neg\varphi}}
& = 1- \pa{b_{\varphi}^\circ + b_{\varphi} +c_{\varphi}}-\pa{b_{\neg\varphi}^\circ + b_{\neg\varphi} +c_{\neg\varphi}} + \pa{b_{\varphi\land\neg\varphi}^\circ + b_{\varphi\land\neg\varphi} +c_{\varphi\land\neg\varphi}}\\
& = 1- \pa{b_{\varphi}^\circ + b_{\varphi} +c_{\varphi}} - \pa{d_{\varphi}^\circ + d_{\varphi} +c_{\varphi}} + c_\varphi\\
& = 1- \pa{b_{\varphi}^\circ + b_{\varphi} +d_{\varphi}^\circ + d_{\varphi} +2c_{\varphi}} + c_{\varphi}\\
& = 1- \pa{b_{\varphi}^\circ + b_{\varphi} +d_{\varphi}^\circ + d_{\varphi} +c_{\varphi}}\\
& = u_{\varphi}.
\end{aligned}$
\end{enumerate}

From the above equations we obtain:
\begin{align*}
[\tratresseis(\widetilde{p})]\pa{\varphi} &= \bigl(
b_{\varphi\land\circ\varphi}^\circ + b_{\varphi\land\circ\varphi} +c_{\varphi\land\circ\varphi},\\
&\quad (b_\varphi^\circ + b_\varphi +c_\varphi) - (b_{\varphi\land\circ\varphi}^\circ + b_{\varphi\land\circ\varphi} +c_{\varphi\land\circ\varphi}) - (b_{\varphi\land\neg\varphi}^\circ + b_{\varphi\land\neg\varphi} +c_{\varphi\land\neg\varphi}),\\ 
&\quad b_{\varphi\land\neg\varphi}^\circ + b_{\varphi\land\neg\varphi} +c_{\varphi\land\neg\varphi},\\
&\quad (b_{\neg\varphi}^\circ + b_{\neg\varphi} +c_{\neg\varphi}) - (b_{\neg\varphi\land\circ\varphi}^\circ + b_{\neg\varphi\land\circ\varphi} +c_{\neg\varphi\land\circ\varphi}) - (b_{\varphi\land\neg\varphi}^\circ + b_{\varphi\land\neg\varphi} +c_{\varphi\land\neg\varphi}),\\
&\quad b_{\neg\varphi\land\circ\varphi}^\circ + b_{\neg\varphi\land\circ\varphi} +c_{\neg\varphi\land\circ\varphi},\\
&\quad 1-(b_{\varphi\lor\neg\varphi}^\circ + b_{\varphi\lor\neg\varphi} +c_{\varphi\lor\neg\varphi}) 
\bigr)\\
&= (b_\varphi^\circ, b_\varphi,c_\varphi, d_\varphi, d_\varphi^\circ, u_\varphi)\\
&= \widehat{p}\pa{\varphi}.
\end{align*}

Therefore $\tratresseis\pa{\traseistres\pa{\widehat{p}}}=\widehat{p}$.

\item Given $\widetilde{p}\in \mathbb{P}_{3}$, let $\widehat{p}=\tratresseis\pa{\widetilde{p}}$. Then $\widehat{p}\pa{\varphi} = \pa{\widehat{p}_1\pa{\varphi}, \widehat{p}_2\pa{\varphi},\widehat{p}_3\pa{\varphi}, \widehat{p}_4\pa{\varphi},\widehat{p}_5\pa{\varphi}, \widehat{p}_6\pa{\varphi}}$ with
\begin{center}
\begin{itemize}[nosep]
\begin{multicols}{2}
\item[] $\widehat{p}_1\pa{\varphi} = \widetilde{p}_1\pa{\varphi \land \circ \varphi}$
\item[] $\widehat{p}_3\pa{\varphi} = \widetilde{p}_1\pa{\varphi \land \neg \varphi}$
\item[] $\widehat{p}_5\pa{\varphi} = \widetilde{p}_1\pa{\neg \varphi \land \circ \varphi}$
\item[] $\widehat{p}_2\pa{\varphi} = \widetilde{p}_1\pa{\varphi} - \widetilde{p}_1\pa{\varphi \land \circ \varphi} - \widetilde{p}_1\pa{\varphi \land \neg \varphi}$
\item[] $\widehat{p}_4\pa{\varphi} = \widetilde{p}_1\pa{\neg \varphi} - \widetilde{p}_1\pa{\neg \varphi \land \circ \varphi} - \widetilde{p}_1\pa{\varphi \land \neg \varphi}$
\item[] $\widehat{p}_6\pa{\varphi} = 1 - \widetilde{p}_1\pa{\varphi \lor \neg \varphi}$
\end{multicols}
\end{itemize}
\end{center}

Consequently:

$\begin{aligned}[t]
[\traseistres\pa{\widehat{p}}]\pa{\varphi} 
&=
\pa{
\widehat{p}_1\pa{\varphi}+ \widehat{p}_2\pa{\varphi}+ \widehat{p}_3\pa{\varphi},
\widehat{p}_3\pa{\varphi}+ \widehat{p}_4\pa{\varphi}+ \widehat{p}_5\pa{\varphi},
\widehat{p}_1\pa{\varphi}+ \widehat{p}_5\pa{\varphi}}\\
&=
\Bigl(
\cancel{\widetilde{p}_1\pa{\varphi \land \circ \varphi}} +
\pa{\widetilde{p}_1\pa{\varphi} - \cancel{\widetilde{p}_1\pa{\varphi \land \circ \varphi}} - \cancel{\widetilde{p}_1\pa{\varphi \land \neg \varphi}}} +
\cancel{\widetilde{p}_1\pa{\varphi \land \neg \varphi}},\\
&\quad
\cancel{\widetilde{p}_1\pa{\varphi \land \neg \varphi}} +
\pa{\widetilde{p}_1\pa{\neg \varphi} - \cancel{\widetilde{p}_1\pa{\neg \varphi \land \circ \varphi}} - \cancel{\widetilde{p}_1\pa{\varphi \land \neg \varphi}}} +
\cancel{\widetilde{p}_1\pa{\neg \varphi \land \circ \varphi}},\\
&\quad
\widetilde{p}_1\pa{\varphi \land \circ \varphi} + \widetilde{p}_1\pa{\neg \varphi \land \circ \varphi}\Bigr)
\\
&=
\pa{\widetilde{p}_1\pa{\varphi}, \widetilde{p}_1\pa{\neg \varphi}, \widetilde{p}_1\pa{\circ \varphi}}\\
&=
\pa{\widetilde{p}_1\pa{\varphi}, \widetilde{p}_2\pa{\varphi}, \widetilde{p}_3\pa{\varphi}}\\
&=
\widetilde{p}\pa{\varphi}.
\end{aligned}$

Therefore $\traseistres\pa{\tratresseis\pa{\widetilde{p}}}=\widetilde{p}$.
\end{enumerate}
\end{proof}

\begin{corolario}\label{coro:traduccionNS-6}
Define $\traseisuno:\mathbb{P}_{6} \longrightarrow \mathbb{P}_{1}$ as follows: for $\widehat{p}\in \mathbb{P}_6$ with $\widehat{p}\pa{\varphi}=\pa{b_\varphi^\circ,b_\varphi, c_\varphi, d_\varphi, d_\varphi^\circ, u_\varphi}$, set
$$[\traseisuno\pa{\widehat{p}}]\pa{\varphi}=b_\varphi^\circ + b_\varphi + c_\varphi.$$
Define $\traunoseis: \mathbb{P}_{1} \longrightarrow \mathbb{P}_{6}$ by $[\traunoseis(p)]\pa{\varphi}=\widehat{p}\pa{\varphi} = \pa{\widehat{p}_1\pa{\varphi}, \widehat{p}_2\pa{\varphi},\ldots, \widehat{p}_6\pa{\varphi}}$, where:
\begin{center}
\begin{itemize}[nosep]
\begin{multicols}{2}
\item[] $\widehat{p}_1\pa{\varphi} = p\pa{\varphi \land \circ \varphi}$
\item[] $\widehat{p}_3\pa{\varphi} = p\pa{\varphi \land \neg \varphi}$
\item[] $\widehat{p}_5\pa{\varphi} = p\pa{\neg \varphi \land \circ \varphi}$
\item[] $\widehat{p}_2\pa{\varphi} = p\pa{\varphi} - p\pa{\varphi \land \circ \varphi} - p\pa{\varphi \land \neg \varphi}$
\item[] $\widehat{p}_4\pa{\varphi} = p\pa{\neg \varphi} - p\pa{\neg \varphi \land \circ \varphi} - p\pa{\varphi \land \neg \varphi}$
\item[] $\widehat{p}_6\pa{\varphi} = 1 - p\pa{\varphi \lor \neg \varphi}$
\end{multicols}
\end{itemize}
\end{center}

Then:
\begin{enumerate}
    \item $\traunoseis\pa{\traseisuno\pa{\widehat{p}}}=\widehat{p}$.
    \item $\traseisuno\pa{\traunoseis\pa{p}}=p$.
\end{enumerate}
That is, $\traseisuno$ and $\traunoseis$ are inverses of each other, and their composition (depending on the order) yields the identity on $\mathbb{P}_{6}$ or on $\mathbb{P}_{1}$.
\end{corolario}

\begin{proof}
By Lemmas \ref{lema:AP6V-induce-AP} and \ref{lema:APNE-induce-AP6V}, the functions $\traseisuno$ and $\traunoseis$ are well-defined. Moreover, since $\traseisuno=\tratresuno\circ \traseistres$ and $\traunoseis=\tratresseis\circ \traunotres$, using Corollaries \ref{coro:traduccionNST-NS} and \ref{coro:traduccionNST-6} we obtain the desired results.
\end{proof}

A consequence of Corollaries \ref{coro:traduccionNST-NS}, \ref{coro:traduccionNST-6}, and \ref{coro:traduccionNS-6} is that the diagram presented in Figure \ref{fig:equivalencia-enfoques} commutes. This allows the functions in $\mathbb{P}_{1}$, $\mathbb{P}_{3}$, and $\mathbb{P}_{6}$ to be translated into one another, enabling probability functions to be viewed with 1, 3, or 6 components, thereby gaining expressivity or losing simplicity, as the case may be.

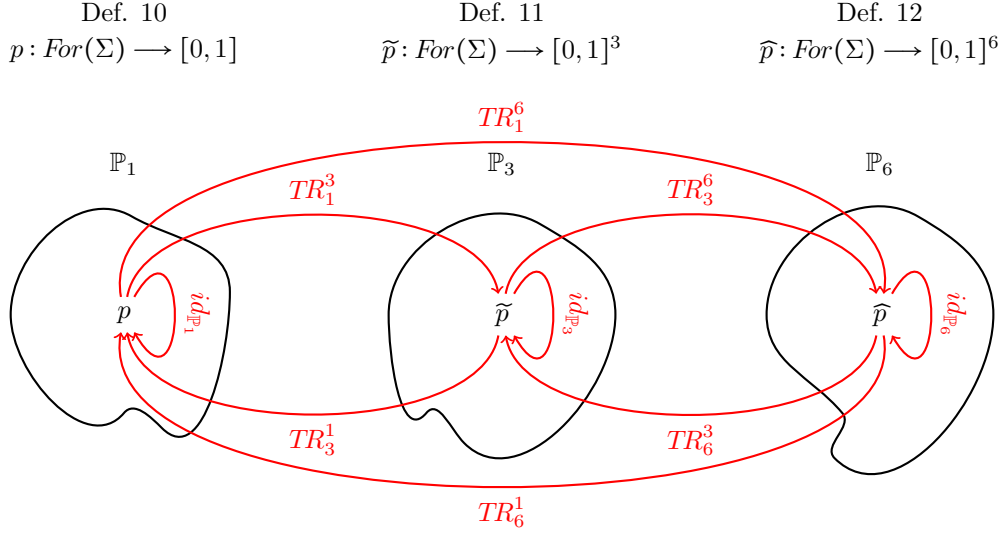
\begin{figure}
\begin{center}
\begin{tikzpicture}
        [vertex_style/.style={circle,shading=ball,scale=0.6,ball color=red,draw=red!80!white,drop shadow={opacity=0.4}}]
        \tikzset{myroundednode/.style={
            fill=white, 
            rounded corners=3mm,
            text centered 
        }}
      \begin{scope}[shift={(0,0)}, rotate=0, scale=1]
            \path[draw,black,thick]  (0,0)to [out=90,in=150] (1.6,1.3)to [out=-30,in=90] (2.9,0.3)to [out=-90,in=-45] (2,-1.5)to [out=135,in=50] (1.5,-1.3)to [out=-130,in=-90] (0,0);
            \path[draw,black, thick]  (5,0)to [out=90,in=150] (7,1.2)to [out=-30,in=90] (8,0)to [out=-90,in=-45] (6,-1.7)to [out=135,in=50] (5.5,-1.3)to [out=-130,in=-90] (5,0);
            \path[draw,black,thick]  (10,0)to [out=90,in=150] (12,1.3)to [out=-30,in=90] (13,0)to [out=-90,in=-45] (10.9,-1.9)to [out=135,in=-50] (11,-1.3)to [out=130,in=-90] (10,0);
            
            \node[myroundednode] at (1.5,2) (PNS) {$\mathbb{P}_{1}$};
            \node[myroundednode] at (6.5,2) (PNST) {$\mathbb{P}_{3}$};
            \node[myroundednode] at (11.5,2) (P6) {$\mathbb{P}_{6}$};

            \node[myroundednode] at (1.5,4) (def-PNS) {Def. \ref{def:funcion_de_probabilidad}};
            \node[myroundednode] at (6.5,4) (def-PNST) {Def. \ref{def:funcion_de_probabilidad-twist}};
            \node[myroundednode] at (11.5,4) (def-P6) {Def. \ref{def:funcion_probabilidad_6V}};

            \node[myroundednode] at (1.5,0) (p) {$p$};
            \node[myroundednode] at (6.5,0) (ps) {$\widetilde{p}$};
            \node[myroundednode] at (11.5,0) (ph) {$\widehat{p}$};

            \node[myroundednode] at (1.5,3.5) (PNS) {$p:\forsigma \longrightarrow [0,1 ]$};
            \node[myroundednode] at (6.5,3.5) (PNST) {$\widetilde{p}:\forsigma \longrightarrow [0,1 ]^3$};
            \node[myroundednode] at (11.5,3.5) (P6) {$\widehat{p}:\forsigma\longrightarrow [0,1 ]^6$};

        \draw[<-,thick, red] (p) .. controls +(260:3cm) and +(280:3cm) .. node[below,sloped] {$\traseisuno$} (ph);
        \draw[->,thick, red] (p) .. controls +(100:3cm) and +(80:3cm) .. node[above,sloped] {$\traunoseis$} (ph);        
        \draw[<-,thick, red] (ps) .. controls +(280:1.7cm) and +(260:1.7cm) .. node[below,sloped] {$\traseistres$} (ph);
        \draw[<-,thick, red] (p) .. controls +(280:1.7cm) and +(260:1.7cm) .. node[below,sloped] {$\tratresuno$} (ps);
        \draw[->,thick, red] (ps) .. controls +(80:1.7cm) and +(100:1.7cm) .. node[above,sloped] {$\tratresseis$} (ph);
        \draw[->,thick, red] (p) .. controls +(80:1.7cm) and +(100:1.7cm) .. node[above,sloped] {$\traunotres$} (ps);
        \draw[->,thick, red] (ps) .. controls +(60:1.7cm) and +(-60:1.7cm) .. node[above,sloped] {$id_{\mathbb{P}_{3}}$} (ps);
        \draw[->,thick, red] (ph) .. controls +(60:1.7cm) and +(-60:1.7cm) .. node[above,sloped] {$id_{\mathbb{P}_{6}}$} (ph);
        \draw[->,thick, red] (p) .. controls +(60:1.7cm) and +(-60:1.7cm) .. node[above,sloped] {$id_{\mathbb{P}_{1}}$} (p);
        \end{scope}
    \end{tikzpicture}
    \caption{Equivalence between probability functions for $\letkp$ in $\mathbb{P}_{1}$, $\mathbb{P}_{3}$ and $\mathbb{P}_{6}$}
    \label{fig:equivalencia-enfoques}
\end{center}
\end{figure}

\subsection{Semantical definition}

Now, the semantical version of each kind of \letkp-probability function will be introduced. They will be based on the suitable notion of twist structures which, as discussed in Observation~\ref{obs:twist-sigma-alg}, can be seen as a transposition of the notion of $\sigma$-algebras, from classical logic to \letkp\ (or even to \fde).

\begin{definition}\label{def:modelo-de-probabilidad}
    \textbf{(Probability Model for $\letkp$)}
    Given a model $\mathsf{M} = \langle X, v \rangle$ for $\letkp$ and a probability measure $\mu$ on $\mathcal{B}_X$, the triple $\mathsf{M}_\mu = \langle X, \mu, v \rangle$ is called a probability model for $\letkp$ based on $\mathsf{M}$ and on $\mu$.
\end{definition}

\begin{definition}\label{def:funciones-probabilidad-inducidas}
    \textbf{(Probability Functions Induced by a Probability Model for $\letkp$)}
    Given a probability model $\mathsf{M}_\mu = \langle X, \mu, v \rangle$ for $\letkp$ based on $\mathsf{M}$ and on $\mu$, three distinct probability functions induced by $\mathsf{M}_\mu$ are defined:
    \begin{enumerate}
        \item Probability function induced by $\mathsf{M}_\mu$, $\pmu: \forsigma \longrightarrow [0,1]$ given by: $\pmu\pa{\varphi} = \mu\pa{|\varphi|_\mathsf{M}^+}$.
        \item Twist probability function induced by $\mathsf{M}_\mu$, $\pmutilde: \forsigma \longrightarrow [0,1]^3$ given by: 
        $$\pmutilde\pa{\varphi} = \pa{\mu\pa{|\varphi|_\mathsf{M}^+}, \mu\pa{|\varphi|_\mathsf{M}^-}, \mu\pa{|\varphi|_\mathsf{M}^\circ}}.$$
        \item 6-valued probability function induced by $\mathsf{M}_\mu$, $\pmuhat: \forsigma \longrightarrow [0,1]^6$ given by:
        $$\pmuhat\pa{\varphi} = \pa{\mu\pa{B^\circ_\varphi}, \mu\pa{B_\varphi}, \mu\pa{C_\varphi}, \mu\pa{D_\varphi}, \mu\pa{D_\varphi^\circ},  \mu\pa{U_\varphi}}.$$
    \end{enumerate}
\end{definition}

The previous definition gave rise to functions of 1, 3, and 6 dimensions by means of a probabilistic model. Although the origin of these functions differs from the syntactical approach (Definitions \ref{def:funcion_de_probabilidad}, \ref{def:funcion_de_probabilidad-twist} and \ref{def:funcion_probabilidad_6V}), they behave in a remarkably similar manner. In fact, the induced functions of 1, 3, and 6 dimensions from Definition \ref{def:funciones-probabilidad-inducidas} belong to $\mathbb{P}_1$, $\mathbb{P}_3$, and $\mathbb{P}_6$, respectively, as stated in Propositions \ref{prop:robustez-1}, \ref{prop:robustez-3} and \ref{prop:robustez-6}. 

\begin{proposition}\label{prop:robustez-1}
    Given $\mathsf{M}_\mu=\langle X, \mu, v \rangle$ a probabilistic model for $\letkp$ based on $\mathsf{M}=\langle X, v\rangle$ and $\mu$, then $\pmu$, (the probability function induced by $\mathsf{M}_\mu$)  is a probability function for $\letkp$, i.e. $\pmu\in\mathbb{P}_1$.
\end{proposition}

\begin{proof}
    We show that $\pmu$ is a probability function for $\letkp$ by verifying that $\pmu$ satisfies the four conditions of Definition \ref{def:funcion_de_probabilidad}. Let $\varphi,\psi \in For(\Sigma)$:
    \begin{enumerate}
        \item Assume $\vdash_{\letkp}\varphi$. Since $\mathsf{M}=\langle X, v\rangle$ is a twist model for $\letkp$, $v$ is a valuation in $\mathcal{M}(\mathcal{B}_X)$. Given that the consequence relation induced by $\mathcal{M}(\mathcal{B}_X)$ is sound and complete with respect to $\vdash_{\letkp}$, we have $v(\varphi)\in D_{\mathcal{B}_X}$, which implies $v_1(\varphi)=|\varphi|^+=X$. As $\mu$ is a probability measure, $\mu(X)=1$, hence $\pmu(\varphi)=\mu(|\varphi|^+)=\mu(X)=1$.
        \item Suppose that for every $\psi$, it holds that $\varphi\vdash_{\letkp} \psi$. In particular, $\varphi\vdash_{\letkp} \psi\land\neg\psi\land\circ\psi$ and therefore $\vdash_{\letkp} \varphi\rightarrow\pa{\psi\land\neg\psi\land\circ\psi}$. This implies that $v_1\pa{\varphi\rightarrow\pa{\psi\land\neg\psi\land\circ\psi}}=X$, this entails $v_1(\varphi)\subseteq v_1(\psi\land\neg\psi\land\circ\psi)$, and thus $v_1(\varphi)=v_1(\psi\land\neg\psi\land\circ\psi)=\emptyset$. Since $\mu(\emptyset)=0$, we obtain $\pmu(\varphi)=0$.
        \item From $\varphi\vdash_{\letkp} \psi$, we obtain that $\vdash_{\letkp} \varphi\rightarrow\psi$ which implies that $v_1(\varphi)\subseteq v_1(\psi)$. Since $\mu$ is a probability measure, it is monotone, so $\mu(v_1(\varphi))\leq \mu(v_1(\psi))$, yielding $\pmu(\varphi)\leq \pmu(\psi)$.
        \item We have $v_1(\varphi\land\psi)=|\varphi|^+_\mathsf{M}\cap |\psi|^+_\mathsf{M}$ and $v_1(\varphi\lor\psi)=|\varphi|^+_\mathsf{M}\cup |\psi|^+_\mathsf{M}$. Moreover, $\mu$ being a probability measure implies $\mu(|\varphi|^+_\mathsf{M}\cup |\psi|^+_\mathsf{M})=\mu(|\varphi|^+_\mathsf{M})+\mu(|\psi|^+_\mathsf{M})-\mu(|\varphi|^+_\mathsf{M}\cap |\psi|^+_\mathsf{M})$, equivalently $\mu(v_1(\varphi\lor\psi))=\mu(v_1(\varphi))+\mu(v_1(\psi))-\mu(v_1(\varphi\land\psi))$, i.e., $\pmu(\varphi) + \pmu(\psi) = \pmu(\varphi \land \psi) + \pmu(\varphi \lor \psi)$.
    \end{enumerate}
    Therefore, $\pmu\in\mathbb{P}_1$.
\end{proof}

 To prove the three dimensional and the six dimensional cases we can use some of the translations among the classes $\mathbb{P}_1$, $\mathbb{P}_3$ and $\mathbb{P}_6$. So we have the following lemma.

\begin{lemma}\label{lema:trad-funciones-inducidas}
Given $\mathsf{M}_\mu $ a probability model for $\letkp$ and its induced functions $\pmu$, $\pmutilde$ and $\pmuhat$, it holds that: 

\begin{tasks}[label=\arabic*.](2)
\task $\traunotres \pa{\raisebox{0.13em}{$\pmu$}} =\pmutilde$ 
\task $\tratresuno \pa{\pmutilde}                 =\pmu$
\task $\traunoseis \pa{\pmu}                      =\pmuhat$
\task $\traseisuno \pa{\pmuhat}                   =\pmu$
\end{tasks}
\end{lemma}

\begin{proof}
Using Definition \ref{def:funciones-probabilidad-inducidas} we have that:
\begin{enumerate}
\item Using items 12  and 15 from Proposition \ref{prop:propiedades-extensiones} it holds that:\\ $\traunotres\pa{\pmu\pa{\varphi}}=\pa{\pmu\pa{\varphi},\pmu\pa{\neg\varphi},\pmu\pa{\circ\varphi}}=$ $\pa{\mu\pa{|\varphi|_\mathsf{M}^+},\mu\pa{|\neg\varphi|_\mathsf{M}^+},\mu\pa{|{\circ}\varphi|_\mathsf{M}^+}}=$  $\pa{\mu\pa{|\varphi|_\mathsf{M}^+},\mu\pa{|\varphi|_\mathsf{M}^-},\mu\pa{|\varphi|_\mathsf{M}^\circ}}=\pmutilde\pa{\varphi}$.

\item $\tratresuno\pa{\pmutilde\pa{\varphi}}=\tratresuno\pa{\pa{\mu\pa{|\varphi|_\mathsf{M}^+},\mu\pa{|\varphi|_\mathsf{M}^-},\mu\pa{|\varphi|_\mathsf{M}^\circ}}}=\mu\pa{|\varphi|_\mathsf{M}^+}=\pmu\pa{\varphi}$.

\item $\traunoseis\pa{\pmu\pa{\varphi}}=$
$\Bigl(
{\pmu}\pa{\varphi\land\circ\varphi},\
{\pmu}\pa{\varphi}-{\pmu}\pa{\varphi\land\neg\varphi}-{\pmu}\pa{\varphi\land\circ\varphi},\ 
{\pmu}\pa{\varphi\land\neg\varphi},\\ \hspace*{2.8cm}
{\pmu}\pa{\neg\varphi}-{\pmu}\pa{\varphi\land\neg\varphi}-{\pmu}\pa{\varphi\land\circ\varphi},\
{\pmu}\pa{\neg\varphi\land\circ\varphi},\
1-{\pmu}\pa{\varphi\lor\neg\varphi}
\Bigr)$

let us check each component. First we have that ${\pmu}\pa{\varphi\land\circ\varphi}=\mu\pa{|\varphi\land\circ\varphi|^+}=\mu\pa{B^\circ_\varphi}$, analogously ${\pmu}\pa{\neg\varphi\land\circ\varphi}=\mu\pa{|\neg\varphi\land\circ\varphi|^+}=\mu\pa{D^\circ_\varphi}$. We also have that ${\pmu}\pa{\varphi\land\neg\varphi}=\mu\pa{|\varphi\land\neg\varphi|^+}=\mu\pa{C_\varphi}$ and $1-{\pmu}\pa{\varphi\lor\neg\varphi}=1-\mu\pa{|\varphi\lor\neg\varphi|^+}=\mu\pa{X\setminus|\varphi\lor\neg\varphi|^+}=\mu\pa{|{\sim}\pa{\varphi\lor\neg\varphi}|^+}=\mu\pa{U_\varphi}$. We have that ${\pmu}\in\mathbb{P}_1$ and since $\varphi\land\neg\varphi\vdash_{\letkp}\varphi$ and $\varphi\land\circ\varphi\vdash_{\letkp}\varphi$ using item 10 of Proposition \ref{prop:propiedades-funcion-probabilidad} we have that ${\pmu}\pa{\varphi}-{\pmu}\pa{\varphi\land\neg\varphi}-{\pmu}\pa{\varphi\land\circ\varphi}= {\pmu}\pa{\varphi\land{\sim}\pa{\pa{\varphi\land\neg\varphi}\lor\pa{\varphi\land\circ\varphi}}}$. Using interderivability of formulas in $\letkp$ as well as the fact that ${\pmu}\in\mathbb{P}_1$, the previous probability can be simplified as follows:

{\small
$\begin{aligned}[t]
    \pmu\pa{\varphi\land\sim\pa{\pa{\varphi\land\neg\varphi}\lor\pa{\varphi\land\circ\varphi}}}
    &=\pmu\pa{\varphi\land\sim\pa{\varphi\land\pa{\neg\varphi\lor\circ\varphi}}}\\
    &=\pmu\pa{\varphi\land\pa{{\sim}\varphi\lor{\sim}\pa{\neg\varphi\lor\circ\varphi}}}\\
    &=\pmu\pa{\pa{\varphi\land{\sim}\varphi}\lor\pa{\varphi\land{\sim}\pa{\neg\varphi\lor\circ\varphi}}}\\
    &=\cancel{\pmu\pa{\varphi\land{\sim}\varphi}}+\pmu\pa{\varphi\land{\sim}\pa{\neg\varphi\lor\circ\varphi}}-\cancel{\pmu\pa{\pa{\varphi\land{\sim}\varphi}\land\pa{\varphi\land{\sim}\pa{\neg\varphi\lor\circ\varphi}}}}\\
    &=\mu\pa{|\varphi\land{\sim}\pa{\neg\varphi\lor\circ\varphi}|^+}\\
    &=\mu\pa{B_\varphi}
\end{aligned}$}

Analogously $\pmu\pa{\varphi}-\pmu\pa{\varphi\land\neg\varphi}-\pmu\pa{\varphi\land\circ\varphi}= \pmu\pa{\neg\varphi\land{\sim}\pa{\varphi\lor\circ\varphi}}= \mu\pa{D_\varphi}$.
After analyzing the six components we can conclude that $\traunoseis\pa{\pmu\pa{\varphi}}=\pa{\mu\pa{B^\circ_\varphi},\mu\pa{B,_\varphi}, \mu\pa{C_\varphi} ,\mu\pa{D_\varphi},\mu\pa{D^\circ_\varphi},\mu\pa{U_\varphi}}=\pmuhat\pa{\varphi}$.

\item By Proposition \ref{prop:6-particion}, the fact that $\mu$  is a probability measure, as well as item 1 of Proposition \ref{prop:propiedades-6-regiones}, we have that $\traseisuno\pa{\pmuhat\pa{\varphi}}= \traseisuno\pa{\pa{\mu\pa{B^\circ_\varphi}, \mu\pa{B_\varphi}, \mu\pa{C_\varphi},  \mu\pa{D_\varphi}, \mu\pa{D_\varphi^\circ}, \mu\pa{U_\varphi}}}=\mu\pa{B^\circ_\varphi}+\mu\pa{B_\varphi}+\mu\pa{C_\varphi}=\mu\pa{B^\circ_\varphi\cup B_\varphi\cup C_\varphi}= \mu\pa{|\varphi|_\mathsf{M}^+}=\pmu\pa{\varphi}$.
\end{enumerate}
\end{proof}

We have confirmed that by applying an appropriate translation to an induced probability function we can get an induced twist probability function or an induced 6-valued probability function and viceversa. 

\begin{proposition}\label{prop:robustez-3}
    Given $\mathsf{M}_\mu=\langle X, \mu, v \rangle$ a probabilistic model for $\letkp$ based on $\mathsf{M}=\langle X, v\rangle$ and $\mu$, then $\pmutilde$ (the twist probability function induced by $\mathsf{M}_\mu$) is a twist probability function for $\letkp$, i.e. $\pmutilde\in\mathbb{P}_3$.
\end{proposition}
\begin{proof}
     By Proposition \ref{prop:robustez-1} it holds that $\pmu\in\mathbb{P}_1$ and by Lemma \ref{lema:trad-funciones-inducidas}, $\traunotres\pa{\pmu\pa{\varphi}}=\pmutilde\pa{\varphi} $. Hence, using Lemma \ref{lema:AP_es_APtwist} we conclude that $\pmutilde \in \mathbb{P}_3$. 
\end{proof}
\begin{proposition}\label{prop:robustez-6}
    Given $\mathsf{M}_\mu=\langle X, \mu, v \rangle$ a probabilistic model for $\letkp$ based on $\mathsf{M}=\langle X, v\rangle$ and $\mu$, then $\pmuhat$ (the 6-valued probability function induced by $\mathsf{M}_\mu$) is a 6-valued probability function for $\letkp$, i.e. $\pmuhat\in\mathbb{P}_6$.
\end{proposition}
  \begin{proof}
     By Proposition \ref{prop:robustez-1} it holds that $\pmu\in\mathbb{P}_1$ and by Lemma \ref{lema:trad-funciones-inducidas}, $\traunoseis\pa{\pmu\pa{\varphi}}=\pmuhat\pa{\varphi} $. Hence, using Lemma \ref{lema:APNE-induce-AP6V} we conclude that $\pmuhat \in \mathbb{P}_6$. 

\end{proof}

\begin{example} \label{ex:relig1} \textbf{[Theistic Belief, cont.]}
Let $\mathsf{M}_\mu=\langle X, \mu, v \rangle$ a probabilistic model for $\letkp$ based on the twist model $\mathsf{M}=\langle X, v\rangle$ given in Example~\ref {ex:relig}. Then, the 6-valued probability function $\pmuhat$ of Definition~\ref{def:funciones-probabilidad-inducidas}(3) assigns a measure to the six disjoint regions into which the population $X$ of the country was divided according to theistic beliefs. Thus, in a country where, for instance, a large portion of the population holds strong religious beliefs, the set $B_\varphi^\circ$ will have a high probability measure. In turn, if a significant portion of the population consists of strong atheists, the probability measure of region $D_\varphi^\circ$ will be high. In either cases, the set $|\varphi|_\mathsf{M}^\circ=B_\varphi^\circ \cup D_\varphi^\circ$  of people with strong convictions about God's existence will have a  high probability measure.
\end{example}

\section{Completeness}\label{sec:completeness}
Two approaches have been presented so far to define probability functions for $\letkp$ of 1, 3 and 6 dimensions, first from a syntactic approach by means of Definitions~\ref{def:funcion_de_probabilidad}, \ref{def:funcion_de_probabilidad-twist} and \ref{def:funcion_probabilidad_6V};  and later through a semantic approach by means of Definition~\ref{def:funciones-probabilidad-inducidas}. Thanks to the Propositions \ref{prop:robustez-1}, \ref{prop:robustez-3} and \ref{prop:robustez-6} we have that the induced functions are sound respect to the set of axioms that define the clases $\mathbb{P}_1$, $\mathbb{P}_3$ and $\mathbb{P}_6$. To prove that the approaches are equivalent, it is necessary to prove completeness for each of the clases.

We proceed similarly to the soundness case, i.e. we first guarantee the one dimensional case and later using the appropriate translations we obtain the result for the higher dimensional cases. So we first prove that given a probability function $p$ for $\letkp$, it is possible to construct a probabilistic model $\mathsf{M}_\mu$ for $\letkp$ that induces such a function. To this end, it is necessary to build an appropriate twist model $\mathsf{M}$ for $\letkp$ and a measure $\mu$. In order to build $\mathsf{M}$ it is necessary to recall some basic definitions, as well as to generalize some known concepts.
\begin{definition}
\textbf{(Assignment in $\mathcal{M}_6$)}
    Let $V=\{\mathsf{p}_1, \mathsf{p}_2,\ldots, \mathsf{p}_m\}$ be the set of propositional variables and let $\mathbf{6}=\{\mathsf{T,t,b,f,F,n}\}$ be the set of truth values of $\mathcal{M}_6$. Any function $a:V\longrightarrow\mathbf{6}$ will be called an assignment from $V$ to $\mathbf{6}$.
\end{definition}

\begin{proposition}
    Given $V=\{\mathsf{p}_1, \mathsf{p}_2,\ldots, \mathsf{p}_m\}$ of cardinality $m$ and $\mathbf{6}=\{\mathsf{T,t,b,f,F,n}\}$ of cardinality 6, let $X=\{a \mid a:V\longrightarrow\mathbf{6}\}$ be the set of all assignments from $V$ to $\mathbf{6}$. Then the cardinality of $X$ is $6^m$, i.e., $X$ is a finite set.
\end{proposition}

\begin{proposition}
    Given $V=\{\mathsf{p}_1, \mathsf{p}_2,\ldots,\mathsf{p}_m\}$ as the set of propositional variables, the signature $\Sigma=\{\land,\lor,\rightarrow,\neg,\circ\}$, and $X$ as the set of all assignments from $V$ to $\mathbf{6}$, then for any assignment $a\in X$ there exists a unique function $v_a:\forsigma\longrightarrow\langle \mathbf{6},\land,\lor,\rightarrow,\neg,\circ \rangle$ that extends $a$ and such that $v_a$ is a valuation in $\mathcal{M}_6$.
\end{proposition}

\begin{definition}
\textbf{(Set of models of a formula)}    
    Let $\varphi\in\forsigma$ be a formula and let $X$ be the set of all assignments from $V$ to $\mathbf{6}$. The set of models of $\varphi$ is the set $\Mod{\varphi}=\{a\in X \mid v_a(\varphi)\in D_6\}$.
\end{definition}

Since each of the truth values in $\mathbf{6}$ can be identified with a triple in $T_{B_2}$ and the classical negation $\sim$ is definable in $\letkp$, it is possible to encode these truth values. Given a truth value $w\in \mathbf{6}$, let $(z_1,z_2,z_3)\in \mathbf{2}_{\letkp}^3$ be its associated triple. Then its encoding will be the conjunction $t_{z_1}\land t_{z_2} \land t_{z_3}$ where:
$$
\begin{array}{ccccc}
t_{z_1}= \left\lbrace 
\begin{array}{lc}
    \textsf{p} & \text{if } z_1=1 \\[1.5ex]
    {\sim}\textsf{p} & \text{if } z_1=0
\end{array} \right.
&\quad&
t_{z_2}= \left\lbrace 
\begin{array}{lc}
    \neg\textsf{p} & \text{if } z_2=1 \\[1.5ex]
    {\sim}\neg\textsf{p} & \text{if } z_2=0
\end{array} \right.
&\quad&
t_{z_3}= \left\lbrace 
\begin{array}{lc}
    \circ\textsf{p} & \text{if } z_3=1 \\[1.5ex]
    {\sim}{\circ}\textsf{p} & \text{if } z_3=0
\end{array} \right.  
\end{array}
$$

Thus, for example, the truth value $\mathsf{T}$ is identified with the triple $(1,0,1)$ and its associated formula would be $\textsf{p}\land{\sim}{\neg}\textsf{p}\land \circ \textsf{p}$. However, in several cases, due to the restrictions on the triples in $\mathbf{2}_{\letkp}^3$, the encoding can be simplified as established by the following lemma, which allows us to verify that each of the truth values in $\mathbf{6}$ can be identified by means of a formula.

\begin{lemma}\label{lema:identificacion-valores-formulas}
Given $V=\{\mathsf{p}_1, \mathsf{p}_2,\ldots, \mathsf{p}_m\}$, $\Sigma=\{\land,\lor,\rightarrow,\neg,\circ\}$ and an assignment $a$ from $V$ to $\mathbf{6}$, the functions $f_w:V\longrightarrow\forsigma$ for $w\in\mathbf{6}$ given by:
\begin{align*}
f_\mathsf{T}(\mathsf{p}) &= \mathsf{p} \land \circ \mathsf{p} &  
f_\mathsf{t}(\mathsf{p}) &= \mathsf{p}\land{\sim}\neg \mathsf{p}\land {\sim}{\circ} \mathsf{p} &
f_\mathsf{b}(\mathsf{p}) &= \mathsf{p}\land\neg \mathsf{p}   \\
f_\mathsf{f}(\mathsf{p}) &= {\sim}\mathsf{p}\land\neg\mathsf{p}\land {\sim}{\circ} \mathsf{p} &  
f_\mathsf{F}(\mathsf{p}) &= \neg \mathsf{p}\land\circ \mathsf{p} &
f_\mathsf{n}(\mathsf{p}) &= {\sim}\mathsf{p}\land{\sim}\neg \mathsf{p}
\end{align*}
satisfy that $v_a\pa{f_w\pa{\mathsf{p}}}\in D_6$ if and only if $a(\mathsf{p})=w$ for every $\mathsf{p}\in V$.
\end{lemma}

\begin{proof}
It suffices to construct the truth table for each of the functions.
\[
\begin{array}{c|c|c|c|c|c|c}
\textsf{p} & f_\mathsf{T}(\textsf{p}) & f_{\mathsf{t}}(\textsf{p}) & f_\mathsf{b}(\textsf{p}) & f_{\mathsf{f}}(\textsf{p}) & f_\mathsf{F}(\textsf{p}) & f_n(\textsf{p}) \\
\hline
\mathsf{T} & \cellcolor{gray!25}\mathsf{T} & \mathsf{F} & \mathsf{F} & \mathsf{F} & \mathsf{F} & \mathsf{F} \\\hline
\mathsf{t} & \mathsf{F} & \cellcolor{gray!25}\mathsf{t} & \mathsf{f} & \mathsf{F} & \mathsf{F} & \mathsf{F} \\\hline
\mathsf{b} & \mathsf{F} & \mathsf{F} & \cellcolor{gray!25}\mathsf{b} & \mathsf{F} & \mathsf{F} & \mathsf{F} \\\hline
\mathsf{f} & \mathsf{F} & \mathsf{F} & \mathsf{f} & \cellcolor{gray!25}\mathsf{t} & \mathsf{F} & \mathsf{F} \\\hline
\mathsf{F} & \mathsf{F} & \mathsf{F} & \mathsf{F} & \mathsf{F} & \cellcolor{gray!25}\mathsf{T} & \mathsf{F} \\\hline
\mathsf{n} & \mathsf{F} & \mathsf{n} & \mathsf{n} & \mathsf{n} & \mathsf{F} & \cellcolor{gray!25}\mathsf{t} \\
\end{array}
\]
\end{proof}
\begin{obs}\label{obs:funciones-fw-identifican-regiones}
For our purposes, the functions $f_w$ are defined for propositional variables; however, they can be generalized to arbitrary formulas in $\forsigma$, for example $f_\mathsf{t}\pa{\varphi}= \varphi\land {\sim}\neg\varphi\land {\sim}\circ\varphi$. Given the region-identification formulas of Observation~\ref{obs:particion-con-extensones-positivas} we have that $\varphi_B = \varphi\land{\sim}\pa{\neg \varphi\lor {\circ} \varphi }$, then it can be easily proved that $\varphi\land {\sim}\neg\varphi\land {\sim}\circ\varphi\Dashv\vDash_{\mathcal{M}_6}\varphi\land{\sim}\pa{\neg \varphi\lor {\circ} \varphi }.$\footnote{$\varphi\Dashv\vDash\psi$ denotes that $f_\mathsf{f}\pa{\varphi}=$ $\varphi\vDash\psi$ and $\psi\vDash\varphi$.} In fact, we have that each of the following pair of formulas have the same models:\\

$\begin{array}{lll}
f_\mathsf{T}\pa{\varphi}\Dashv\vDash_{\mathcal{M}_6}\varphi_{B^\circ} & \hspace{1cm} f_\mathsf{t}\pa{\varphi}\Dashv\vDash_{\mathcal{M}_6}\varphi_{B} & \hspace{1cm} f_\mathsf{b}\pa{\varphi}\Dashv\vDash_{\mathcal{M}_6}\varphi_{C} \\[2mm]
f_\mathsf{f}\pa{\varphi}\Dashv\vDash_{\mathcal{M}_6}\varphi_{D} & \hspace{1cm} f_\mathsf{F}\pa{\varphi}\Dashv\vDash_{\mathcal{M}_6}\varphi_{D^\circ} & \hspace{1cm} f_\mathsf{n}\pa{\varphi}\Dashv\vDash_{\mathcal{M}_6}\varphi_{U^\circ}.
\end{array}
$
%

\

\noindent
Even more, each pair of formulas have exactly the same value under any assignment, then the formulas $f_w\pa{\varphi}$ also can be used to identify regions. 

\end{obs}

\begin{definition}
Let $V=\{\mathsf{p}_1,\mathsf{p}_2,\ldots,\mathsf{p}_m\}$. We define the following concepts:
\begin{enumerate}
    \item A \emph{literal} is a formula in the set: $$Lit:=V\cup\{\neg \mathsf{p}\mid \mathsf{p}\in V\}\cup\{\circ \mathsf{p}\mid \mathsf{p}\in V\}\cup\{{\sim} \mathsf{p}\mid \mathsf{p}\in V\}\cup\{{\sim}\neg \mathsf{p}\mid \mathsf{p}\in V\}\cup\{{\sim}{\circ} \mathsf{p}\mid \mathsf{p}\in V\}.$$
    \item A \emph{clause} is a conjunction of literals, i.e., it is of the form $\displaystyle C=\bigwedge_{i=1}^k l_i$ with $l_i\in Lit$.
    \item A formula is in \emph{Disjunctive Normal Form} if it is a disjunction of one or more clauses.
    \item The \emph{clause associated with an assignment} $a:V\longrightarrow\mathbf{6}$ is the clause $\displaystyle C(a)=\bigwedge_{i=1}^m f_{a(\mathsf{p}_i)}(\mathsf{p}_i)$.
\end{enumerate}
\end{definition}

\begin{lemma}\label{lema:DNF}
    Given $X=\{a \mid a:V\longrightarrow\mathbf{6}\}$ the set of assignments from $V$ to $\mathbf{6}$ and $A\subseteq X$, there exists a formula $\varphi_A$ in Disjunctive Normal Form such that $A=\Mod{\varphi_A}$.
\end{lemma}

\begin{proof}
Let $A\subseteq X$. We define the formula:
$$\varphi_A= \left\lbrace 
\begin{array}{lc}
    \displaystyle\bigvee_{a\in A}C(a) & \text{if } A\neq\emptyset \\[1.5ex]
    \textsf{p}_1\land\neg \textsf{p}_1\land \circ \textsf{p}_1 & \text{if } A=\emptyset 
\end{array} \right.$$
Note that the formula $\varphi_A$ is in Disjunctive Normal Form, then we consider whether $A=\emptyset$. If $A=\emptyset$, then we have $\varphi_\emptyset = \textsf{p}_1\land\neg \textsf{p}_1\land \circ \textsf{p}_1$, but $\Mod{\varphi_\emptyset}=\Mod{\textsf{p}_1\land\neg \textsf{p}_1\land \circ \textsf{p}_1}=\emptyset$. For the case where $A \neq \emptyset$, we show that $A\subseteq \Mod{\varphi_A}$ and $\Mod{\varphi_A}\subseteq A$ to guarantee equality.
Let $a'\in A$. Then, by Lemma \ref{lema:identificacion-valores-formulas}, $v_{a'}\pa{f_{a'\pa{\textsf{p}_j}}\pa{\textsf{p}_j}}\in D_6$ for every $j$ with $1\leq j\leq m$. Consequently, $v_{a'}\pa{C\pa{a'}}\in D_6$, and therefore $v_{a'}\pa{\bigvee_{a\in A}C\pa{a}}=v_{a'}\pa{\varphi_A}\in D_6$. Hence, $a'\in \Mod{\varphi_A}$. Now suppose that $a'\in \Mod{\varphi_A}$. Then $v_{a'}\pa{\bigvee_{a\in A}C\pa{a}}\in D_6$, so there exists an $a''\in A$ such that $v_{a'}\pa{C\pa{a''}}\in D_6$. This means that $v_{a'}\pa{f_{a''\pa{\textsf{p}_j}}\pa{\textsf{p}_j}}\in D_6$ for every $j$ with $1\leq j\leq m$, and by Lemma \ref{lema:identificacion-valores-formulas}, this occurs only if $a'\pa{\textsf{p}_j}=a''\pa{\textsf{p}_j}$ for all $j$ with $1\leq j\leq m$. That is, $a'=a''$, and therefore $a'\in A$.
\end{proof}

\begin{corolario}\label{coro:equivalencia-modelos}
Given $\psi \in \forsigma$ a formula, we can construct the set ${\Mod{\psi}}$. Let us denote by $\DNF\pa{\psi}$ the formula $\varphi_{\Mod{\psi}}$ constructed as in Lemma \ref{lema:DNF}.  We have that $\psi\Dashv\vDash_{\mathcal{M}_6} \DNF(\psi)$ and consequently also $\psi\dashv\vdash_{\letkp} \DNF(\psi)$.
\end{corolario}


Now we have enough elements to define: first, the set of states $X$; second, the twist structure based on it $T_{\mathcal{B}_X}$; and finally, a valuation on the matrix induced by the twist structure $\mathcal{M}\pa{\mathcal{B}_X}$ as stated by Lemma \ref{lema:valuacion-formulas-a-modelos}.

\begin{lemma}\label{lema:valuacion-formulas-a-modelos}
Let $X = \{ a \mid a : V \longrightarrow \mathbf{6} \}$ be the set of assignments from $V$ to $\mathbf{6}$, let $\mathcal{B}_{X}$ be the algebra of subsets of $X$, and let $T_{\mathcal{B}_{X}}$ be the twist structure for $\letkp$ induced by $\mathcal{B}_{X}$. Then the function $\overline{v}: \forsigma \longrightarrow T_{\mathcal{B}_{X}}$ given by $$\val{\varphi} = \pa{\Mod{\varphi}, \Mod{\neg\varphi}, \Mod{\circ\varphi}}$$ is a valuation in $\mathcal{M}\pa{\mathcal{B}_{X}}$.
\end{lemma}

\begin{proof}
We verify that each connective satisfies the operations in $T_{\mathcal{B}_{X}}$.

Let $\val{\varphi} = \pa{\Mod{\varphi}, \Mod{\neg\varphi}, \Mod{\circ\varphi}}$ and 
$\val{\psi} = \pa{\Mod{\psi}, \Mod{\neg\psi}, \Mod{\circ\psi}}$.

\begin{enumerate}
    \item $\val{\varphi \land \psi} = \pa{\Mod{\varphi \land \psi}, \Mod{\neg\pa{\varphi \land \psi}}, \Mod{\circ\pa{\varphi \land \psi}}}$.
    
    \begin{enumerate}[label=(\alph*)]
        \item 
        $\begin{aligned}[t]
            \Mod{\varphi \land \psi} 
            &= \lla{ a \in X \mid v_a\pa{\varphi \land \psi} \in \D{6} } \\
            &= \lla{ a \in X \mid v_a\pa{\varphi} \in \D{6} \text{ and } v_a\pa{\psi} \in \D{6} } \\
            &= \lla{ a \in X \mid v_a\pa{\varphi} \in \D{6} } \cap \lla{ a \in X \mid v_a\pa{\psi} \in \D{6} }\\
            &= \Mod{\varphi} \cap \Mod{\psi}.
        \end{aligned}$
        
        \item 
        $\begin{aligned}[t]
            \Mod{\neg\pa{\varphi \land \psi}} 
            &= \lla{ a \in X \mid v_a\pa{\neg\pa{\varphi \land \psi}} \in \D{6} } \\
            &= \lla{ a \in X \mid v_a\pa{\neg\varphi \lor \neg\psi} \in \D{6} } \\
            &= \lla{ a \in X \mid v_a\pa{\neg\varphi} \in \D{6} \text{ or } v_a\pa{\neg\psi} \in \D{6} } \\
            &= \lla{ a \in X \mid v_a\pa{\neg\varphi} \in \D{6} } \cup \lla{ a \in X \mid v_a\pa{\neg\psi} \in \D{6} } \\
            &= \Mod{\neg\varphi} \cup \Mod{\neg\psi}.
        \end{aligned}$
        
        \item 
        $\begin{aligned}[t]
            \Mod{\circ\pa{\varphi \land \psi}} 
            &= \lla{ a \in X \mid v_a\pa{\circ\pa{\varphi \land \psi}} \in \D{6} } \\
            &= \lla{ a \in X \mid v_a\pa{\varphi \land \psi} = \mathsf{T} \text{ or } v_a\pa{\varphi \land \psi} = \mathsf{F} } \\
            &= \lla{ a \in X \mid v_a\pa{\varphi}=\mathsf{T} \text{ and } v_a\pa{\psi}=\mathsf{T} } \\
            &\quad \cup 
               \lla{ a \in X \mid v_a\pa{\varphi}=\mathsf{F} \text{ or } v_a\pa{\psi}=\mathsf{F} } \\
            &= \pa{ \Mod{\varphi} \cap \Mod{\circ\varphi} \cap \Mod{\psi} \cap \Mod{\circ\psi} } \\
            &\quad \cup \pa{\Mod{\neg\varphi} \cap \Mod{\circ\varphi} } \cup \pa{\Mod{\neg\psi} \cap \Mod{\circ\psi}}.
        \end{aligned}$
    \end{enumerate}
    Therefore, $\val{\varphi \land \psi} = \val{\varphi} \tland \val{\psi}$.
    
    \item $\val{\varphi \lor \psi} = \pa{\Mod{\varphi \lor \psi}, \Mod{\neg\pa{\varphi \lor \psi}}, \Mod{\circ\pa{\varphi \lor \psi}}}$.
    
    \begin{enumerate}[label=(\alph*)]
        \item 
        $\begin{aligned}[t]
            \Mod{\varphi \lor \psi} 
            &= \lla{ a \in X \mid v_a\pa{\varphi \lor \psi} \in \D{6} } \\
            &= \lla{ a \in X \mid v_a\pa{\varphi} \in \D{6} \text{ or } v_a\pa{\psi} \in \D{6} } \\
            &= \lla{ a \in X \mid v_a\pa{\varphi} \in \D{6} } \cup \lla{ a \in X \mid v_a\pa{\psi} \in \D{6} } \\
            &= \Mod{\varphi} \cup \Mod{\psi}.
        \end{aligned}$
        
        \item 
        $\begin{aligned}[t]
            \Mod{\neg\pa{\varphi \lor \psi}} 
            &= \lla{ a \in X \mid v_a\pa{\neg\pa{\varphi \lor \psi}} \in \D{6} } \\
            &= \lla{ a \in X \mid v_a\pa{\neg\varphi \land \neg\psi} \in \D{6} } \\
            &= \lla{ a \in X \mid v_a\pa{\neg\varphi} \in \D{6} } \cap \lla{ a \in X \mid v_a\pa{\neg\psi} \in \D{6} } \\
            &= \Mod{\neg\varphi} \cap \Mod{\neg\psi}.
        \end{aligned}$
        
        \item 
        $\begin{aligned}[t]
            \Mod{\circ\pa{\varphi \lor \psi}} 
            &= \lla{ a \in X \mid v_a\pa{\circ\pa{\varphi \lor \psi}} \in \D{6} } \\
            &= \lla{ a \in X \mid v_a\pa{\varphi \lor \psi} = \mathsf{T} \text{ or } v_a\pa{\varphi \lor \psi} = \mathsf{F} } \\
            &= \lla{ a \in X \mid v_a\pa{\varphi} = \mathsf{T} \text{ or } v_a\pa{\psi} = \mathsf{T} } \\
            &\quad \cup 
               \lla{ a \in X \mid v_a\pa{\varphi} = \mathsf{F} \text{ and } v_a\pa{\psi} = \mathsf{F} } \\
            &= \pa{ \Mod{\varphi} \cap \Mod{\circ\varphi} } \cup \pa{ \Mod{\psi} \cap \Mod{\circ\psi} } \\
            &\quad \cup \pa{\Mod{\neg\varphi} \cap \Mod{\circ\varphi} \cap \Mod{\neg\psi} \cap \Mod{\circ\psi} }.
        \end{aligned}$
    \end{enumerate}
    Therefore, $\val{\varphi \lor \psi} = \val{\varphi} \tlor \val{\psi}$.
    
    \item $\val{\varphi \to \psi} = \pa{\Mod{\varphi \to \psi}, \Mod{\neg\pa{\varphi \to \psi}}, \Mod{\circ\pa{\varphi \to \psi}}}$.
    
    \begin{enumerate}[label=(\alph*)]
        \item 
        $\begin{aligned}[t]
            \Mod{\varphi \to \psi} 
            &= \lla{ a \in X \mid v_a\pa{\varphi \to \psi} \in \D{6} } \\
            &= \lla{ a \in X \mid v_a\pa{\varphi} \notin \D{6} \text{ or } v_a\pa{\psi} \in \D{6} } \\
            &= \lla{ a \in X \mid v_a\pa{\varphi} \notin \D{6} } \cup \lla{ a \in X \mid v_a\pa{\psi} \in \D{6} } \\
            &= \Mod{\varphi}^{\text{C}} \cup \Mod{\psi}.
        \end{aligned}$
        
        \item 
        $\begin{aligned}[t]
            \Mod{\neg\pa{\varphi \to \psi}} 
            &= \lla{ a \in X \mid v_a\pa{\neg\pa{\varphi \to \psi}} \in \D{6} } \\
            &= \lla{ a \in X \mid v_a\pa{\varphi} \in \D{6} \text{ and } v_a\pa{\neg\psi} \in \D{6} } \\
            &= \lla{ a \in X \mid v_a\pa{\varphi} \in \D{6} } \cap \lla{ a \in X \mid v_a\pa{\neg\psi} \in \D{6} } \\
            &= \Mod{\varphi} \cap \Mod{\neg\psi}.
        \end{aligned}$
        
        \item 
        $\begin{aligned}[t]
            \Mod{\circ\pa{\varphi \to \psi}} 
            &= \lla{ a \in X \mid v_a\pa{\circ\pa{\varphi \to \psi}} \in \D{6} } \\
            &= \lla{ a \in X \mid v_a\pa{\varphi \to \psi} = \mathsf{T} \text{ or } v_a\pa{\varphi \to \psi} = \mathsf{F} } \\
            &= \lla{ a \in X \mid v_a\pa{\varphi} = \mathsf{F} \text{ or } v_a\pa{\psi} = \mathsf{T} } \\
            &\quad \cup \lla{ a \in X \mid v_a\pa{\varphi} \in \D{6} \text{ and } v_a\pa{\psi} = \mathsf{F} } \\
            &= \pa{ \Mod{\neg \varphi} \cap \Mod{\circ\varphi} } \cup \pa{ \Mod{\psi} \cap \Mod{\circ\psi} } \\
            &\quad \cup \pa{ \Mod{\varphi} \cap \Mod{\neg\psi} \cap \Mod{\circ\psi} }.
        \end{aligned}$
    \end{enumerate}
    Therefore, $\val{\varphi \to \psi} = \val{\varphi} \tto \val{\psi}$.
    
    \item $\val{\neg\varphi} = \pa{\Mod{\neg\varphi}, \Mod{\neg\neg\varphi}, \Mod{\circ\neg\varphi}}$.
    
    \begin{enumerate}[label=(\alph*)]
        \item 
        $\begin{aligned}[t]
            \Mod{\neg\neg\varphi} 
            &= \lla{ a \in X \mid v_a\pa{\neg\neg\varphi} \in \D{6} } \\
            &= \lla{ a \in X \mid v_a\pa{\varphi} \in \D{6} } \\
            &= \Mod{\varphi}.
        \end{aligned}$
        
        \item 
        $\begin{aligned}[t]
            \Mod{\circ\neg\varphi} 
            &= \lla{ a \in X \mid v_a\pa{\circ\neg\varphi} \in \D{6} } \\
            &= \lla{ a \in X \mid v_a\pa{\neg\varphi} = \mathsf{T} \text{ or } v_a\pa{\neg\varphi} = \mathsf{F} } \\
            &= \lla{ a \in X \mid v_a\pa{\varphi} = \mathsf{F} \text{ or } v_a\pa{\varphi} = \mathsf{T} } \\
            &= \lla{ a \in X \mid v_a\pa{\circ\varphi} \in \D{6} } \\
            &= \Mod{\circ\varphi}.
        \end{aligned}$
    \end{enumerate}
    Therefore, $\val{\neg\varphi} = \tneg\val{\varphi}$.
    
    \item $\val{\circ\varphi} = \pa{\Mod{\circ\varphi}, \Mod{\neg{\circ}\varphi}, \Mod{\circ{\circ}\varphi}}$.
    
    \begin{enumerate}[label=(\alph*)]
        \item 
        $\begin{aligned}[t]
            \Mod{\neg{\circ}\varphi} 
            &= \lla{ a \in X \mid v_a\pa{\neg{\circ}\varphi} \in \D{6} } \\
            &= \lla{ a \in X \mid v_a\pa{\circ \varphi} \notin \D{6}   } \\
            &= \Mod{\circ \varphi}^{\text{C}}.
        \end{aligned}$
        
        \item 
        $\begin{aligned}[t]
            \Mod{{\circ}{\circ}\varphi} 
            &= \lla{ a \in X \mid v_a\pa{{\circ}{\circ}\varphi} \in \D{6} } = X.
        \end{aligned}$
    \end{enumerate}
    Therefore, $\val{\circ\varphi} = \tcirc\val{\varphi}$.
\end{enumerate}
\end{proof}
Thanks to Lemma \ref{lema:valuacion-formulas-a-modelos} we have the following corollary.

\begin{corolario}\label{coro:construccion-modelo}
 Let $X = \lla{ a \mid a : V \longrightarrow \mathbf{6} }$ be the set of assignments from $V$ to $\mathbf{6}$, let $\mathcal{B}_{X}$ be the algebra of subsets of $X$, and let $T_{\mathcal{B}_{X}}$ be the twist structure for $\letkp$ induced by $\mathcal{B}_{X}$. Let $\overline{v}: \forsigma \longrightarrow T_{\mathcal{B}_{X}}$ given by $\val{\varphi} = \pa{\Mod{\varphi}, \Mod{\neg\varphi}, \Mod{\circ\varphi}}$. Then the ordered pair $\mathsf{M} = \langle X, \overline{v} \rangle$ is a twist model for $\letkp$.
\end{corolario}

To conclude the construction of the desired probability model for $\letkp$ we need to define a probability measure $\mu$ on $\mathcal{B}_X$.
\begin{lemma}\label{lema:construccion-medida}
Let $p:\forsigma \longrightarrow [0,1]$ be a probability function for $\letkp$, let $\mathsf{M} = \langle X, \overline{v} \rangle$ be the twist model for $\letkp$ given in Corollary \ref{coro:construccion-modelo}, and for any set $A\in \mathcal{B}_X $ let $\varphi_A$ be the formula defined in Lemma \ref{lema:DNF}. Then $\mu:\mathcal{B}_X \longrightarrow [0,1]$ given by:
$$\mu(A)= p\pa{\varphi_A}$$
is well-defined and is a probability measure on $\mathcal{B}_X$.
\end{lemma}

\begin{proof}
To see that $\mu$ is well-defined, take $A\in \mathcal{B}_X$ and suppose there exist $\varphi , \psi\in \forsigma$ such that $\Mod{\varphi}=A=\Mod{\psi}$. Then $\DNF(\varphi)=\DNF(\psi)$ and by Corollary \ref{coro:equivalencia-modelos} we have $\varphi\Dashv\vDash_{\mathcal{M}_6}\DNF(\varphi)$ and $\psi\Dashv\vDash_{\mathcal{M}_6}\DNF(\psi)$; hence $\varphi\Dashv\vDash_{\mathcal{M}_6}\psi$, and therefore $\varphi\dashv\vdash_{\letkp}\psi$. Since $p$ is a probability function for $\letkp$, it follows that $p(\varphi)=p(\psi)$.

To verify that $\mu$ is a probability measure on $\mathcal{B}_X$, we proceed as follows:
\begin{enumerate}
\item Let $A\in \mathcal{B}_X$. Then there exists $\varphi_A\in \forsigma$ such that $A=v_1\pa{\varphi_A}$, hence $\mu(A)=p\pa{\varphi_A}$. Since $p$ is a probability function for $\letkp$, we have $0\leq p(\varphi_A)\leq 1$, and thus $0\leq \mu(A)\leq 1$.

\item We have $\vdash_{\letkp}{\circ}{\circ}\varphi$, so $\Mod{{\circ}{\circ}\varphi}=X$. Since $p$ is a probability function for $\letkp$, we obtain $\mu\pa{X}=\mu\pa{\Mod{{\circ}{\circ}\varphi}}=p\pa{{\circ}{\circ}\varphi}=1$.

\item For any $\varphi\in \forsigma$, we have $\Mod{\varphi\land\neg\varphi\land\circ\varphi}=\Mod{\varphi}\cap \Mod{\neg\varphi}\cap \Mod{\circ\varphi}=\emptyset$. Moreover, $\varphi\land\neg\varphi\land\circ\varphi\vdash_{\letkp}\psi$ for any $\psi$, hence $p\pa{\varphi\land\neg\varphi\land\circ\varphi}=0$. Therefore, $\mu\pa{\emptyset}=\mu\pa{\Mod{\varphi\land\neg\varphi\land\circ\varphi}}=p\pa{\varphi\land\neg\varphi\land\circ\varphi}=0$.

\item Since $X$ is finite, it suffices to analyze additivity for pairs of sets. Let $A,B \in \mathcal{B}_X$ such that $A\cap B=\emptyset$. Then there exist $\varphi,\psi\in \forsigma$ with $A=\Mod{\varphi}$, $B=\Mod{\psi}$, and $\Mod{\varphi}\cap \Mod{\psi}=\emptyset$. Consequently, $\varphi,\psi \models_{\mathcal{M}_6}\bot$, and therefore $\varphi,\psi \vdash_{\letkp}\bot$, from which we obtain $0 \leq p\pa{\varphi\land\psi} \leq p\pa{\bot}=0$, so $p\pa{\varphi\land\psi}=0$. On the other hand, we have:
\[
\begin{aligned}
\mu(A\cup B) &=\mu\pa{\Mod{\varphi}\cup \Mod{\psi}}
               =\mu\pa{\Mod{\varphi\lor\psi}}
               =p\pa{\varphi\lor\psi}  \\
             &=p\pa{\varphi}+p\pa{\psi}-\cancel{p\pa{\varphi\land\psi}}
               =p\pa{\varphi}+p\pa{\psi}  \\
             &=\mu\pa{\Mod{\varphi}}+\mu\pa{\Mod{\psi}}
               =\mu(A)+\mu(B).
\end{aligned}
\]
\end{enumerate}
\end{proof}
   
\begin{proposition}\label{prop:completez-1}
Given a probability function $p:\forsigma \longrightarrow [0,1]$ for $\letkp$, there exists $\mathsf{M}_\mu=\langle X, \mu, v \rangle$, a probabilistic model for $\letkp$ and ${\pmu}$ the probability function induced by $\mathsf{M}_\mu$, such that $p=\pmu$.
\end{proposition}
\begin{proof}
    In the first place we have that $\mathsf{M}=\langle X, \overline{v}\rangle$ defined as in Corollary \ref{coro:construccion-modelo} is a twist model for $\letkp$, in the second place $\mu$ defined as in Lemma \ref{lema:construccion-medida} is a measure in $\mathcal{B}_X$, then $\mathsf{M}_\mu=\langle X, \mu, \overline{v}\rangle$ is a probabilistic model for $\letkp$ based on $\mathsf{M}$ and $\mu$, and finally by construction $p=\pmu$.
\end{proof}
Up to this point from Propositions \ref{prop:robustez-1} and \ref{prop:completez-1} we can conclude that:

\begin{corolario}\label{coro:robustez-completitud-1}
 Axioms $Ax1.1-Ax1.4$ of Definition \ref{def:funcion_de_probabilidad} are sound and complete with respect to the class of probability functions induced by probabilistic models for $\letkp$.
\end{corolario}
In order to conclude the completeness of the 3 and 6 dimensional probabilities we need to prove the following propositions:

\begin{proposition}\label{prop:completez-3}
Given a twist probability function $\widetilde{p}:\forsigma \longrightarrow [0,1]^3$ for $\letkp$, there exists $\mathsf{M}_\mu=\langle X, \mu, v \rangle$, a probabilistic model for $\letkp$ and ${\pmutilde}$ the twist probability function induced by $\mathsf{M}_\mu$ such that $\widetilde{p}=\pmutilde$.
\end{proposition}
\begin{proof}
Let $\widetilde{p}$ be a twist probability function for $\letkp$, using Corollary \ref{coro:traduccionNST-NS} it can be translated into a one dimensional probability function $p=\tratresuno\pa{\widetilde{p}}$. Then construct the appropriate probability model for $\letkp$ such that $p=p_\mu$ using Proposition \ref{prop:completez-1} and later using Lemma \ref{lema:trad-funciones-inducidas} translate $p_\mu$ into a induced probability function of 3 dimensions $\pmutilde=\traunotres\pa{\pmu}$. So we have that $\pmutilde=\traunotres\pa{\pmu}=\traunotres\pa{p}=\traunotres\pa{\tratresuno\pa{\widetilde{p}}}$ and by Corollary \ref{coro:traduccionNST-NS} it holds that $\traunotres\pa{\tratresuno\pa{\widetilde{p}}}=\widetilde{p}$, therefore $\widetilde{p}=\pmutilde$ as desired.
\end{proof}

\begin{proposition}\label{prop:completez-6}
Given a 6-valued probability function $\widehat{p}:\forsigma \longrightarrow [0,1]^6$ for $\letkp$, there exists $\mathsf{M}_\mu=\langle X, \mu, v \rangle$, a probabilistic model for $\letkp$ and ${\pmuhat}$ the 6-valued probability function induced by $\mathsf{M}_\mu$ such that $\widehat{p}=\pmuhat$.
\end{proposition}
The proof of Proposition \ref{prop:completez-6} is analogous to the proof of Proposition \ref{prop:completez-3} but using Corollary  \ref{coro:traduccionNS-6} instead of Corollary \ref{coro:traduccionNST-NS}. As an immediate consequence of Propositions \ref{prop:robustez-3} and \ref{prop:completez-3} we obtain that:

\begin{corolario}\label{coro:robustez-completitud-3}
Axioms $Ax3.1-Ax3.6$ of Definition \ref{def:funcion_de_probabilidad-twist} are sound and complete with respect to the class of twist probability functions induced by  probabilistic models for $\letkp$.
\end{corolario}

Analogously from Propositions \ref{prop:robustez-6} and \ref{prop:completez-6} we obtain that:
\begin{corolario}\label{coro:robustez-completitud-6}
    Axioms $Ax6.1-Ax6.9$ of Definition \ref{def:funcion_probabilidad_6V} are sound and complete with respect to the class of 6-valued probability functions induced by  probabilistic models for $\letkp$.

\end{corolario}

\section{Conditional Probabilities}  \label{sect:condit-prob}

Jeffrey's update rule (also known as Jeffrey's conditionalization) is a generalization of Bayes' rule for updating beliefs when the new evidence is not known with total certainty (i.e., the new evidence is not an event with probability 1). Therefore, Jeffrey's update rule is a method for adjusting probabilities when evidence is uncertain yet quantifiable, while certain conditionals remain fixed. It serves as a bridge between classical Bayesian updating and more general approaches to belief revision in artificial intelligence and philosophy.

In the classical context, Bayes' rule updates a prior probability $p\pa{\psi}$ to a posterior or updated probability $\overline{p}\pa{\psi}=p\pa{\psi|\varphi}$ once it is observed that the probability of $\varphi$ is 1 (i.e., the updated probability of formula $\varphi$ is 1, $\overline{p}\pa{\varphi}=1$). The original idea of Jeffrey (in his book \textit{The Logic of Decision}, 1965) is that sometimes evidence is uncertain, leading to an updated probability where $\varphi$ may attain a probability degree $  \uplambda \in [0,1]$, so that $\overline{p}\pa{\varphi}= \uplambda$.

Assuming that there is a finite partition of the sample space: $\{\varphi_1, \varphi_2, \ldots ,\varphi_k\}$, according to Jeffrey's rule, if the new information changes the probabilities of the partition elements from $p\pa{\varphi_i}$ to $\overline{p}\pa{\varphi_i}=  \uplambda_i$, under the assumption of rigidity or invariance of the conditionals which prevents altering the conditional probabilities given each $\varphi_i$ (i.e. $\overline{p}\pa{\psi|\varphi_i}=p\pa{\psi|\varphi_i}$ for $1\leq i\leq k$), then the updated probability is created from a linear combination that expands or contracts the original probability measure over the partition:

$$\displaystyle \overline{p}\pa{\psi}=\sum_{i=1}^{k} {p\pa{\psi|\varphi_i}  \uplambda_i}$$
where $\displaystyle p\pa{\psi|\varphi_i}=\frac{p\pa{\psi\land \varphi_i}}{p\pa{\varphi_i}}$. Notice that using the rigidity assumption we can replace $p\pa{\psi |\varphi_i}$ by $\overline{p}\pa{\psi|\varphi_i}$ in the definition of the updated probability, which becomes simply an instance of the law of total probability.

The standard Jeffrey formula assumes the use of classical logic and probability functions governed by Kolmogorov's axioms. Therefore, in the classic approach, given evidence $\varphi$, we can partition the space into two: $\{\varphi,\neg \varphi\}$, and thus Jeffrey's update takes the form:
$$\overline{p}\pa{\psi}=p\pa{\psi|\varphi}  \uplambda+p\pa{\psi|\neg \varphi}\pa{1-  \uplambda}=\frac{p\pa{\psi\land\varphi}}{p\pa{\varphi}}  \uplambda+\frac{p\pa{\psi\land\neg\varphi}}{p\pa{\neg\varphi}}\pa{1-  \uplambda}$$

In contexts of non-classical probability, Jeffrey's rule can be generalized or reinterpreted, but its essence remains the same: it allows for a soft update based on a change in the probabilities of a partition, while keeping the conditional probabilities given each element of the partition invariant.

In \cite{klein2021probabilities}, for example, using $\fde$ logic, two distinct generalizations of Jeffrey updating are proposed. The first is named as non-standard Jeffrey updating and the second simply as four-valued Jeffrey updating. Both are defined from a semantic approach using the probabilistic models that the authors define using $\fde$ logic, and are then identified with their syntactic counterparts.

As we will see below, if we move from $\fde$ logic to $\letkp$ logic, we obtain a framework in which the results of the generalizations from \cite{klein2021probabilities} can be recovered. Moreover, thanks to the expressive power of $\letkp$, we can, on one hand, simplify the presentation for updating when information is only about $\varphi$, or on the other hand, make a finer partition by identifying those cases where the reliability of the information can be distinguished.

\subsection{Single-Valued Jeffrey Updating}

In this case, we start from the assumption that the information provided for the update only gives us information about the probability of $\varphi$, establishing that the new value is $  \uplambda\in[0,1]$. This information is independent of the information one may have for $\neg\varphi$ and even of the information related to $\circ\varphi$. It is important to emphasize that, unlike the definition of non-standard Jeffrey updating in \cite{klein2021probabilities} based on $\fde$, here it is possible to define the complement of the set $|\varphi|_{\mathsf{M}}^+$ as the extension of some other formula in the language, in particular $|{\sim}\varphi|_{\mathsf{M}}^+$. Therefore, it is possible to create a partition of the space into two: $\{\varphi,{\sim}\varphi\}$. Thus, despite the limited information provided, it is possible to perform a Jeffrey update for that partition according to the following definition.

\begin{definition}\label{def:semantic-update-1}
\textbf{(Semantic single-valued Jeffrey update)}
Let $\mathsf{M}_\mu=\langle X, \mu,  v \rangle$ be a probability model for $\letkp$ based on $\mathsf{M}$ and $\mu$, let $  \uplambda\in[0,1]$, and let $\varphi\in \forsigma$ such that $\mu\pa{|\varphi|_{\mathsf{M}}^+}\in (0,1)$. The semantic single-valued Jeffrey update is a new probability model for $\letkp$, $\mathsf{M}_{\mu^{\varphi,  \uplambda}}=\langle X, \mu^{\varphi,  \uplambda},  v \rangle$, based on $\mathsf{M}$ and $\mu^{\varphi,  \uplambda}$, where:
$$\mu^{\varphi,  \uplambda}\pa{\{x\}}=\left\lbrace 
\begin{tabular}{ll}
$\displaystyle \mu\pa{\{x\}}\frac{  \uplambda}{\mu\pa{|\varphi|_{\mathsf{M}}^+}}$ & if $x\in |\varphi|_{\mathsf{M}}^+$ \\ 
&\\
$\displaystyle \mu\pa{\{x\}}\frac{1-  \uplambda}{\mu\pa{|{\sim}\varphi|_{\mathsf{M}}^+}}$ & if $x\in |{\sim}\varphi|_{\mathsf{M}}^+$\\ 
\end{tabular} \right.$$
\end{definition}

As expected, for $\varphi \in \forsigma$ such that $\mu\pa{|\varphi|_{\mathsf{M}}^+}\in (0,1)$, the update is well-defined and satisfies $\mu^{\varphi,  \uplambda}\pa{|\varphi|_{\mathsf{M}}^+}=  \uplambda$. Note that Definition~\ref{def:semantic-update-1} is a semantic approach for single-valued Jeffrey update. Using Definition~\ref{def:funciones-probabilidad-inducidas}, we can arrive to a purely syntactic approach for single-valued Jeffrey update.

\begin{definition}\label{def:sintactic_one_update}
\textbf{(Syntactic single-valued Jeffrey update)} Let $p:\forsigma \longrightarrow [0,1]$ be a probability function for $\letkp$, let $\varphi\in \forsigma$, and let $  \uplambda\in[0,1]$ such that $p\pa{\varphi}\in(0,1)$. The syntactic single-valued Jeffrey update is a new probability function for $\letkp$, $p^{\varphi,  \uplambda}:\forsigma \longrightarrow [0,1]$, given by:
$$p^{\varphi,  \uplambda}\pa{\psi}=p\pa{\psi\land\varphi}\frac{  \uplambda}{p\pa{\varphi}}+p\pa{\psi\land{\sim}\varphi}\frac{1-  \uplambda}{p\pa{{\sim}\varphi}}$$

\noindent As particular cases we have $\displaystyle p^{\varphi,1}\pa{\psi}=\frac{p\pa{\psi\land\varphi}}{p\pa{\varphi}}$ and $\displaystyle p^{\varphi,0}\pa{\psi}=\frac{p\pa{\psi\land{\sim}\varphi}}{p\pa{{\sim}\varphi}}$.
\end{definition}

\begin{obs}\label{obs:dot-prod}
The syntactic single-valued Jeffrey update can be seen as a special linear combination of the joint probabilities $p\pa{\psi\land\varphi}$ and $p\pa{\psi\land{\sim}\varphi}$. If we consider a vector with the join probabilities $\pa{p\pa{\psi\land\varphi},p\pa{\psi\land{\sim}\varphi}}$ and a vector with the contraction or expansion factors $\pa{\frac{  \uplambda}{p\pa{\varphi}},\frac{1-  \uplambda}{p\pa{{\sim}\varphi}}}$, then the syntactic single-valued Jeffrey update can be expressed as the dot product\footnote{Given two vectors of the same dimension $\vec{a}=\pa{a_1,a_2,\ldots,a_k}$  and $\vec{b}=\pa{b_1,b_2,\ldots,b_k}$ the dot product of $\vec{a}$ and $\vec{b}$ is defined as $\displaystyle \vec{a}\boldsymbol{\cdot}\vec{b}=\Sigma_{i=1}^k a_i b_i$.} of these vectors, i.e. $p^{\varphi,  \uplambda}\pa{\psi}=\pa{p\pa{\psi\land\varphi},p\pa{\psi\land{\sim}\varphi}}\boldsymbol{\cdot}\pa{\frac{  \uplambda}{p\pa{\varphi}},\frac{1-  \uplambda}{p\pa{{\sim}\varphi}}}$.
\end{obs}

Semantic and syntactic single-valued Jeffrey update's definitions coincide. Given a probability model, we obtain the same function if we first obtain the semantic single-valued Jeffrey update and then induce the probability function for the new model, or if, conversely, starting from the original probability model, we first induce the probability function and then perform the syntactic single-valued Jeffrey update, that is to say, Proposition~\ref{prop:semantic=sintactic-1} holds.

\begin{proposition}\label{prop:semantic=sintactic-1}
  Let $\mathsf{M}_\mu=\langle X, \mu,  v \rangle$ be a probability model for $\letkp$ based on $\mathsf{M}$ and $\mu$, let $  \uplambda\in[0,1]$, and let $\varphi\in \forsigma$ such that $\mu\pa{|\varphi|_{\mathsf{M}}^+}\in (0,1)$, then  $p_{\mu^{\varphi,  \uplambda}}=p_\mu^{\varphi,  \uplambda}$.
\end{proposition}
\begin{proof}
    Let $\mu^{\varphi,  \uplambda}$ the probability measure in the semantic single-valued Jeffrey update according to Definition~\ref{def:semantic-update-1}, then: 
    
    $\begin{aligned}[t]
            p_{\mu^{\varphi,  \uplambda}}\pa{\psi} & = \mu^{\varphi,  \uplambda}\pa{|\psi|^+}\\ 
                                               & =\sum_{x\in|\psi|^+}\mu^{\varphi,  \uplambda}\pa{\{x\}}\\
                                               &= \sum_{x\in|\psi|^+\cap|\varphi|^+}\mu^{\varphi,  \uplambda}\pa{\{x\}}+ \sum_{x\in|\psi|^+\cap|\sim\varphi|^+}\mu^{\varphi,  \uplambda}\pa{\{x\}}\\
            &=\sum_{x\in|\psi|^+\cap|\varphi|^+}\mu\pa{\{x\}}\frac{  \uplambda}{\mu\pa{|\varphi|^+}}+ \sum_{x\in|\psi|^+\cap|\sim\varphi|^+}\mu\pa{\{x\}}\frac{1-  \uplambda}{\mu\pa{|{\sim}\varphi|^+}}\\
            &= \mu\pa{|\psi|^+\cap|\varphi|^+}\frac{  \uplambda}{\mu\pa{|\varphi|^+}}+ \mu\pa{|\psi|^+\cap|{\sim}\varphi|^+}\frac{1-  \uplambda}{\mu\pa{|{\sim}\varphi|^+}}\\
            &= \mu\pa{|\psi\land\varphi|^+}\frac{  \uplambda}{\mu\pa{|\varphi|^+}}+ \mu\pa{|\psi\land{\sim}\varphi|^+}\frac{1-  \uplambda}{\mu\pa{|{\sim}\varphi|^+}}\\
            &= p_\mu\pa{\psi\land\varphi}\frac{  \uplambda}{p_\mu\pa{\varphi}}+ p_\mu\pa{\psi\land{\sim}\varphi}\frac{1-  \uplambda}{p_\mu\pa{{\sim}\varphi|^+}}= p_\mu^{\varphi,  \uplambda}\pa{\psi}.
        \end{aligned}$
    
\end{proof}

Thanks to this equivalence we can simply refer to the single-valued Jeffrey update without regard to whether it was obtained from the semantic or syntactic approach.

\begin{obs}
The non-standard Jeffrey update proposed in \cite{klein2021probabilities}, given by:
$$p^{\varphi,  \uplambda}\pa{\psi}=p\pa{\psi\land\varphi}\frac{  \uplambda}{p\pa{\varphi}}+\pa{p\pa{\psi}-p\pa{\psi\land\varphi}}\frac{1-  \uplambda}{1-p\pa{\varphi}},$$

coincides with the single-valued Jeffrey update proposed here, since if $p$ is a probability function for $\letkp$, then $\displaystyle \frac{p\pa{\psi\land{\sim}\varphi}}{p\pa{{\sim}\varphi}}=\frac{p\pa{\psi}-p\pa{\psi\land \varphi}}{1- p\pa{\varphi}}$.
\end{obs}

\subsection{Bayesian Updating}

In the classic approach, Bayesian updating is a particular case of Jeffrey updating when it is assumed that the evidence $\varphi$ attains a probability of 1 in positive Bayesian updating, or 0 in negative Bayesian updating. Using the single-valued Jeffrey update proposed here, we have that for every formula $\varphi$ such that $p\pa{\varphi}>0$, positive Bayesian updating coincides with the classical case:
$$p^{\varphi,pos}\pa{\psi}=p^{\varphi,1}\pa{\psi}$$

On the other hand, negative Bayesian updating corresponds to the case where the probability of $\varphi$ is 0. Therefore, using the single-valued Jeffrey update, negative Bayesian updating for every formula $\varphi$ such that $p\pa{\varphi}<1$ again coincides with the classic case using negation ${\sim}$:

$$p^{\varphi,neg}\pa{\psi}=p^{\varphi,0}\pa{\psi}$$

Evidently, $p^{{\sim}\varphi,pos}=p^{\varphi,neg}$. But given the independence between $\varphi$ and $\neg\varphi$, it holds $p^{\varphi,neg}\neq  p^{\neg\varphi,pos}$ as expected, this leads us to treat evidence independently; therefore, there will be positive and negative Bayesian updates for evidence $\neg\varphi$ or for evidence $\circ\varphi$.

Finally, we have that the order of positive and negative Bayesian updates is independent; that is:

\begin{lemma}
    Let $p:\forsigma \longrightarrow [0,1]$ be a probability function for $\letkp$, and let $\varphi,\psi\in \forsigma$ such that $p\pa{\varphi}$, $p\pa{\psi}$, $p^{\varphi,\uplambda}\pa{\psi}$, $ p^{\psi,\uplambda'}\pa{\varphi}\in (0,1)$ with $\uplambda,\uplambda' \in\{0,1\}$. Then $\pa{p^{\varphi, \uplambda}}^{\psi,\uplambda'}=\pa{p^{\psi,\uplambda'}}^{\varphi,\uplambda}$.
\end{lemma}

\begin{proof}
We have to prove $\pa{p^{\varphi, \uplambda}}^{\psi,\uplambda'}=\pa{p^{\psi,\uplambda'}}^{\varphi,\uplambda}$ when $(\uplambda, \uplambda') \in \lla{(0,0),(0,1), (1,0), (1,1)}$. The proof is straightforward from Definition~\ref{def:sintactic_one_update}, just as an example we can consider the case when $(\uplambda, \uplambda') = (1,0)$.

$$\displaystyle \pa{p^{\varphi, 1}}^{\psi,0}\pa{\gamma} 
= \dfrac{p^{\varphi,1}\pa{\gamma \land {\sim}\psi}}{p^{\varphi,1}\pa{{\sim}\psi}}
= \dfrac{\dfrac{p\pa{\gamma \land {\sim}\psi \land \varphi}}{p\pa{\varphi}}}{\dfrac{p\pa{{\sim} \psi \land \varphi}}{p\pa{\varphi}}}
= \dfrac{\dfrac{p\pa{\gamma \land \varphi \land {\sim} \psi}}{p\pa{{\sim} \psi}}}{\dfrac{p\pa{\varphi \land {\sim} \psi}}{p\pa{{\sim} \psi}}}
= \frac{p^{\psi,0}\pa{\gamma \land \varphi}}{p^{\psi,0}\pa{\varphi}}
= \pa{p^{\psi,0}}^{\varphi, 1}\pa{\gamma}$$
    
\end{proof}

\subsection{6-Valued Jeffrey Updating}

If we now start from the assumption that the information provided for the update is complete for a finer partition, for example using the partition $\mathcal{P}_\varphi=\{B^\circ_\varphi, B_\varphi, C_\varphi, D_\varphi, D^\circ_\varphi, U_\varphi\}$ for a specific formula $\varphi$ (see Proposition~\ref{prop:6-particion}), then Jeffrey updating will expand or contract each of the regions according to the provided evidence. It is important to emphasize that expansion or contraction only makes sense in those cases where the region has positive measure; therefore, those regions with measure zero must remain with measure zero, which leads us to the following definition analogous to Klein's admissibility definition in \cite{klein2021probabilities}.

\begin{definition}
\textbf{(Admissible vector)}
    Let $\varphi\in \forsigma$, and let $\widehat{p}$ be a 6-valued probability function (either semantically induced by a model $\mathsf{M}_\mu$ or syntactically defined) with $\widehat{p}\pa{\varphi}=\pa{b^\circ_\varphi,b_\varphi,c_\varphi,d_\varphi,d^\circ_\varphi,u_\varphi}$. We say that the vector $\vec{  \uplambda}=\pa{b^\circ,b,c,d,d^\circ,u}\in[0,1]^6$ is admissible for $\varphi$ given $\widehat{p}$ if it satisfies that $b^\circ=0$ if $b^\circ_\varphi=0$, $b=0$ if $b_\varphi=0$, $c=0$ if $c_\varphi=0$, $d=0$ if $d_\varphi=0$, $d^\circ=0$ if $d^\circ_\varphi=0$,  and $u=0$ if $u_\varphi=0$.
\end{definition}

Once the notion of admissibility for a vector is established, we can define the semantic 6-valued Jeffrey update. As we did for the semantic single-valued update, we define the actualized probability measure of unitary sets $\{x\}$ as the original measure multiplied by a expansion or contraction factor depending on which element of the partition $x$ belongs to.

\begin{definition}\label{def:semantic-update-6}
\textbf{(Semantic 6-valued Jeffrey update)}
Let $\mathsf{M}_\mu = \langle X, \mu, v \rangle$ be a probability model for $\letkp$ based on $\mathsf{M}$ and $\mu$, let $\varphi \in \forsigma$ with $\widehat{p}_\mu(\varphi) = (b^\circ_\varphi, b_\varphi, c_\varphi, d_\varphi, d^\circ_\varphi, u_\varphi)$, and let $\vec{  \uplambda}=(b^\circ, b, c, d, d^\circ, u) \in [0,1]^6$ be a vector admissible for $\varphi$ given $\pmuhat$. The semantic 6-valued Jeffrey update is a new probability model for $\letkp$,
$$
\mathsf{M}_{\mu^{\varphi,\vec{  \uplambda}}} = \langle X, \mu^{\varphi,\vec{  \uplambda}}, v \rangle,
$$
based on $\mathsf{M}$ and $\mu^{\varphi,\vec{  \uplambda}}$, where:
$$\displaystyle \mu^{\varphi,\vec{  \uplambda}}\pa{\{x\}}=\left\lbrace 
\begin{tabular}{cl}
$ \mu\pa{\{x\}}\frac{b^\circ}   {\mu\pa{B^\circ_\varphi}}$ & if $x\in B^\circ_\varphi$ \\ 
$ \mu\pa{\{x\}}\frac{b}         {\mu\pa{B_\varphi}}$ & if $x\in B_\varphi$ \\ 
$ \mu\pa{\{x\}}\frac{c}         {\mu\pa{C_\varphi}}$ & if $x\in C_\varphi$ \\
$ \mu\pa{\{x\}}\frac{d}         {\mu\pa{D_\varphi}}$ & if $x\in D_\varphi$\\ 
$ \mu\pa{\{x\}}\frac{d^\circ}   {\mu\pa{D^\circ_\varphi}}$ & if $x\in D^\circ_\varphi$ \\ 
$ \mu\pa{\{x\}}\frac{u}        {\mu\pa{U_\varphi}}$ & if $x\in U_\varphi$\\ 
\end{tabular} \right.$$
using the convention that $\frac{0}{0}=0$.
\end{definition}

As in the case of the single-valued Jeffrey update, a syntactic definition can also be given for the six-valued version. However, before presenting the definition, it is necessary to identify the components involved. In will be convenient to analyze first from the sematical perspective. It is clear that it is necessary to handle two formulas simultaneously: the formula whose information is provided to perform the update, say, $\varphi$, and the formula for which one wishes to compute the new probability, say, $\psi$. 

Consider a finite set $X$ and a probability model $\mathsf{M}_\mu$ for $\letkp$ over $X$. From Proposition~\ref{prop:6-particion} we know that  $\mathsf{M}$ induces a partition $\mathcal{P}_\varphi=\{B^{\circ}_\varphi, B_\varphi, C_\varphi, D_\varphi, D^{\circ}_\varphi,  U_\varphi\}$ for every formula $\varphi$. A graphical representation of the individual situation for each formula is shown in Figure~\ref{fig:new-regions}. To treat $\varphi$ and $\psi$ simultaneously, the induced partitions  $\mathcal{P}_\varphi$ and  $\mathcal{P}_\psi$ divide the set $X$ into 36 regions based on the intersection of these partitions, as illustrated in Figure~\ref{fig:36-partition}.

\begin{figure}[H]
\begin{center}
\includegraphics[scale=0.3]{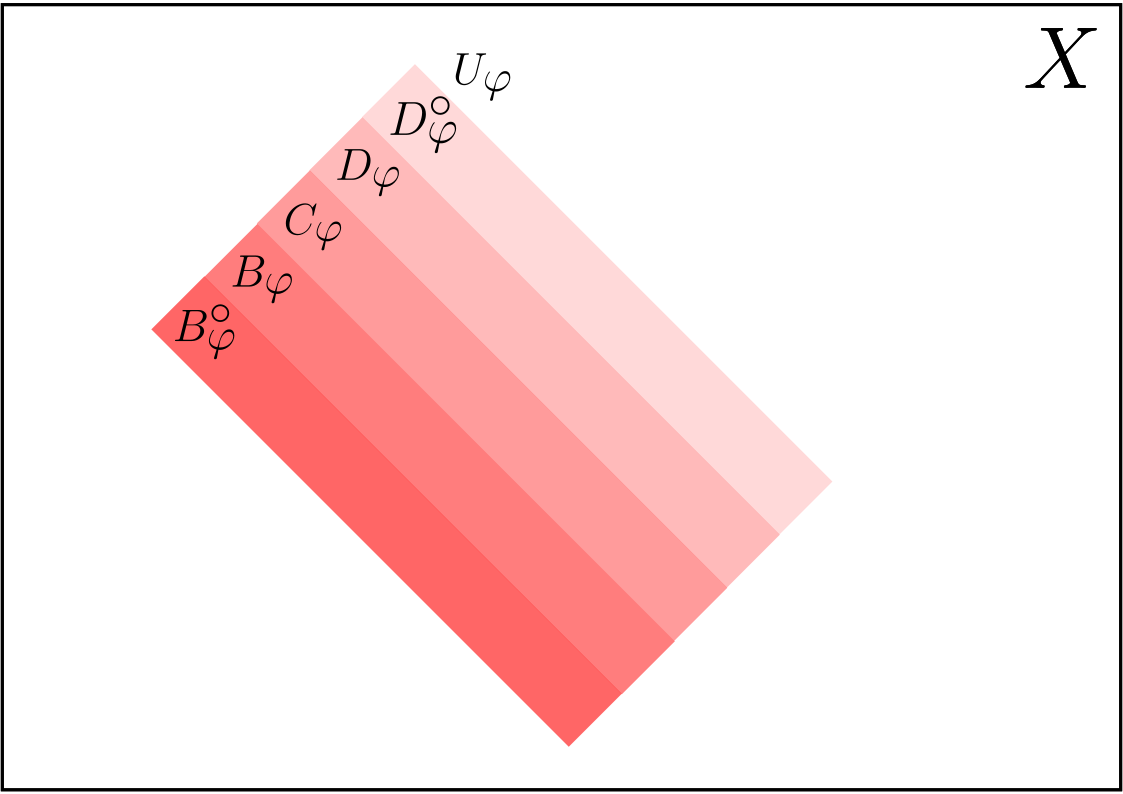}
\includegraphics[scale=0.3]{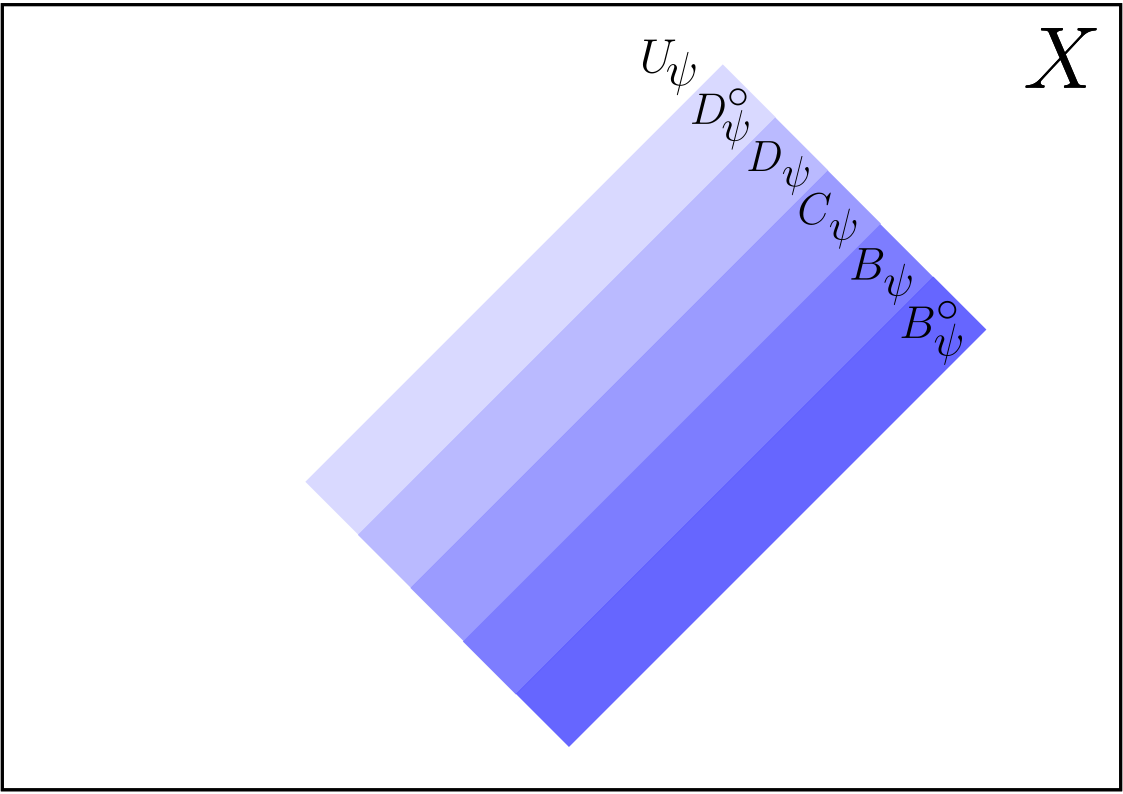}

\caption{Partition of the set $X$ induced by a twist model for $\letkp$ for formulas $\varphi$ (left) and $\psi$ (right).}
\label{fig:new-regions}
\end{center}
\end{figure}

\begin{figure}[H]
\begin{center}
\includegraphics[scale=0.3]{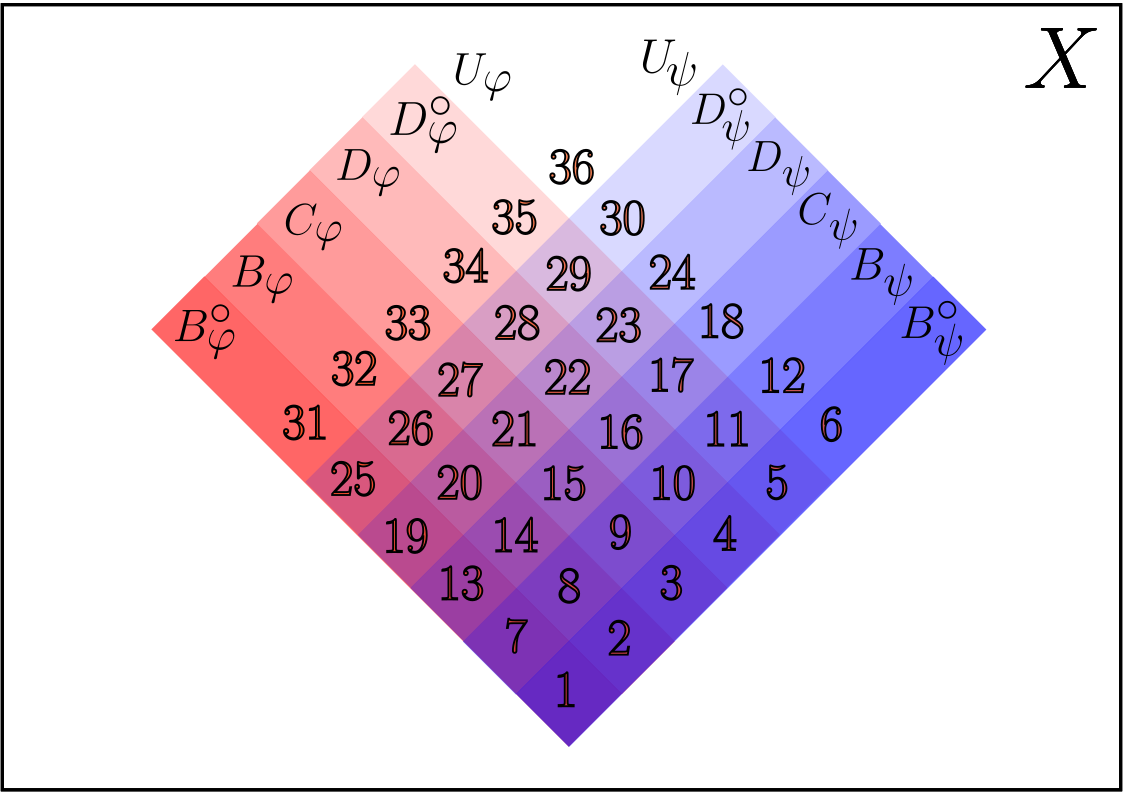}
\end{center}
\caption{Partition of the set $X$ into 36 regions generated by the intersection of the partitions induced by two formulas $\varphi$ and $\psi$ given a twist model for $\letkp$.}
\label{fig:36-partition}
\end{figure}

In particular, each element of the partitions for $\varphi$ and for $\psi$ consists of the union of six regions, as indicated in the following lists.

\[
\begin{array}{lclrlcl}
B^\circ_\varphi &=&\bigcup \{1,7,13,19,25,31\}  &\quad \quad& B^\circ_\psi &=&\bigcup \{1,2,3,4,5,6\}\\
B_\varphi       &=&\bigcup \{2,8,14,20,26,32\}  &&            B_\psi       &=&\bigcup \{7,8,9,10,11,12\}\\
C_\varphi       &=&\bigcup \{3,9,15,21,27,33\}  &&            C_\psi       &=&\bigcup \{13,14,15,16,17,18\}\\
D_\varphi       &=&\bigcup \{4,10,16,22,28,34\} &&            D_\psi       &=&\bigcup \{19,20,21,22,23,24\}\\
D^\circ_\varphi &=&\bigcup \{5,11,17,23,29,35\} &&            D^\circ_\psi &=&\bigcup \{25,26,27,28,29,30\}\\
U_\varphi       &=&\bigcup \{6,12,18,24,30,36\} &&            U_\psi       &=&\bigcup \{31,32,33,34,35,36\}\\
\end{array}
\]

In order to get a compact definition for the syntactic 6-valued Jeffrey update in its six-valued version it will be useful to treat information from a vectorial perspective. As pointed out in Observation~\ref{obs:dot-prod} the single-valued Jeffrey actualization can be think of as the dot product of two vectors, one containing the joint probabilities of a formula in conjunction with the formulas that determine the partition of the space, and a second vector that has the contraction or expansion factors for each element in the partition. In order to do this we introduce the following notions.

\begin{definition} \label{def:vector}
    \textbf{(Vector of contraction-expansion factors)}
    Let $\widehat{p}:\forsigma \longrightarrow [0,1]^6$ be a 6-valued probability function for $\letkp$, let $\varphi\in \forsigma$, with $\widehat{p}(\varphi) = (b^\circ_\varphi, b_\varphi, c_\varphi, d_\varphi, d^\circ_\varphi, u_\varphi)$, and let $\vec{  \uplambda}=(b^\circ,b,c,d,d^\circ,u)\in[0,1]^6$ be admissible for $\varphi$ given $\widehat{p}$. The vector of contraction-expansion factors for $\varphi$ given $\vec{  \uplambda}$ is the vector: $$\vec{  \uplambda} \oslash \widehat{p}\pa{\varphi}=\pa{\frac{b^\circ}{b^\circ_\varphi},\frac{b}{b_\varphi},\frac{c}{c_\varphi},\frac{d}{d_\varphi},\frac{d^\circ}{d^\circ_\varphi},\frac{u}{u_\varphi}} $$
   where the operator $\oslash$ denotes the Hadamard division\footnote{Hadamard division, also known as element wise division is defined for matrices. Given $A=(a_{ij})$ and $B=(b_{ij})$ two matrices of the same size, then  $A\oslash B=C$ where $C=(c_{ij})$ and $c_{ij}=\frac{a_{ij}}{b_{ij}}$.} and using the convention that $\frac{0}{0}=0$.
\end{definition}
Let $X$ and $\mathsf{M}_\mu$ are such that  $\widehat{p}_{\mu}=\widehat{p}$ (recall Proposition~\ref{prop:completez-6}).
To define the vectors of joint probabilities, we need to compute the measure of each of the 36 regions shown in Figure~\ref{fig:36-partition} using the syntactic 6-valued probability function $\widehat{p}$. To this end, we identify each region using a specific conjunction $\psi_J \wedge \varphi_K$ of two formulas, where $J,K \in \{B^\circ,B,C,D,D^\circ,U\}$, recalling the notation introduced in Observation~\ref{obs:particion-con-extensones-positivas}. This will allow us to transition from region measures to formula probabilities. We need to state first a general result.

\begin{proposition}\label{prop:region-intersection} 
Let $\mathsf{M}$ be twist model for $\letkp$. Let $\varphi,\psi\in \forsigma$, $R_1\in \mathcal{P}_\psi$ and $R_2\in \mathcal{P}_\varphi$. Then, their intersection can be expressed as the region $B^\circ_\alpha$, $B_\alpha$ or $C_\alpha$ of a specific formula $\alpha$ according to Table~\ref{table:36-equations}.

\begin{table}[H]
$$
\begin{array}{c|llllll}
\cap & \hspace{.5cm}B^\circ_\varphi & \hspace{.5cm} B_\varphi & \hspace{.5cm} C_\varphi & \hspace{.5cm} D_\varphi & \hspace{.5cm} D^\circ_\varphi & \hspace{.5cm} U_\varphi \\ \hline

B^\circ_\psi  & B^\circ_{\psi_{B^\circ}\land\varphi_{B^\circ}} & B_{\psi_{B^\circ}\land\varphi_B} & C_{\psi_{B^\circ}\land\varphi_C} & B_{\psi_{B^\circ}\land\varphi_D} & B^\circ_{\psi_{B^\circ}\land\varphi_{D^\circ}} & B_{\psi_{B^\circ}\land\varphi_U} \\

B_\psi        & B_{\psi_{B}\land\varphi_{B^\circ}} & B_{\psi_{B}\land\varphi_B} & C_{\psi_{B}\land\varphi_C} & B_{\psi_{B}\land\varphi_D} & B_{\psi_{B}\land\varphi_{D^\circ}} & B_{\psi_{B}\land\varphi_U} \\

C_\psi        & C_{\psi_{C}\land\varphi_{B^\circ}} & C_{\psi_{C}\land\varphi_B} & C_{\psi_{C}\land\varphi_C} & C_{\psi_{C}\land\varphi_D} & C_{\psi_{C}\land\varphi_{D^\circ}} & C_{\psi_{C}\land\varphi_U} \\

D_\psi        & B_{\psi_{D}\land\varphi_{B^\circ}} & B_{\psi_{D}\land\varphi_B} & C_{\psi_{D}\land\varphi_C} & B_{\psi_{D}\land\varphi_D} & B_{\psi_{D}\land\varphi_{D^\circ}} & B_{\psi_{D}\land\varphi_U} \\

D^\circ_\psi  & B^\circ_{\psi_{D^\circ}\land\varphi_{B^\circ}} & B_{\psi_{D^\circ}\land\varphi_B} & C_{\psi_{D^\circ}\land\varphi_C} & B_{\psi_{D^\circ}\land\varphi_D} & B^\circ_{\psi_{D^\circ}\land\varphi_{D^\circ}} & B_{\psi_{D^\circ}\land\varphi_U} \\

U_\psi        & B_{\psi_{U}\land\varphi_{B^\circ}} & B_{\psi_{U}\land\varphi_B} & C_{\psi_{U}\land\varphi_C} & B_{\psi_{U}\land\varphi_D} & B_{\psi_{U}\land\varphi_{D^\circ}} & B_{\psi_{U}\land\varphi_U} \\
\end{array}
$$
\caption{Intersection of the 36 regions in the partition of space $X$ from Figure~\ref{fig:36-partition}}\label{table:36-equations}
\end{table}
\end{proposition}
\begin{proof}
To prove each case, we can profit from the identification between regions and values pointed out in Observation~\ref{obs:particion-con-extensones-positivas}. Let us first analyze one of these regions, to clarify ideas. Suppose that $x\in B^\circ_\psi\cap B^\circ_\varphi$; then, $x\in B^\circ_\psi$ and $x\in B^\circ_\varphi$.  Using Observation~\ref{obs:identificacion-regiones-valores} these can be thought of as $v_x\pa{\psi}=\mathsf{T}$ and $v_x\pa{\varphi}=\mathsf{T}$. Then, the value for the region-identification formulas are $v_x\pa{\psi_{B^\circ}}=\mathsf{T}$ and $v_x\pa{\varphi_{B^\circ}}=\mathsf{T}$. By Observation~\ref{obs:funciones-fw-identifican-regiones} and Lemma~\ref{lema:identificacion-valores-formulas} this is the only valuation in which $\psi_{B^\circ}$ and $\varphi_{B^\circ}$ are designated, even more $v_x\pa{\psi_{B^\circ}\land \varphi_{B^\circ}}=\mathsf{T}$ which means that $x\in B^\circ_{\psi_{B^\circ}\land \varphi_{B^\circ}}$. For the converse, if $x\in B^\circ_{\psi_{B^\circ}\land \varphi_{B^\circ}}$, then $v_x\pa{\psi_{B^\circ}\land \varphi_{B^\circ}}=\mathsf{T}$ which is a designated value. Then $\psi_{B^\circ}$ and $\varphi_{B^\circ}$ must be designated, but according to Observation~\ref{obs:funciones-fw-identifican-regiones} and Lemma~\ref{lema:identificacion-valores-formulas} this only happens if $v_x\pa{\psi}=\mathsf{T}$ and $v_x\pa{\varphi}=\mathsf{T}$ respectively. These means that  $x\in B^\circ_\psi$ and $x\in B^\circ_\varphi$, so $x\in B^\circ_\psi\cap B^\circ_\varphi$. Therefore  $B^\circ_\psi\cap B^\circ_\varphi=B^\circ_{\psi_{B^\circ}\land \varphi_{B^\circ}}$. The proof for the rest of the cases is analogous only considering that $R_1\cap R_2 =B^\circ_{\psi_{R_1}\land\varphi_{R_2}}$ if $v_x\pa{\psi_{R_1}\land \varphi_{R_2}}=\mathsf{T}$, $R_1\cap R_2 =B_{\psi_{R_1}\land\varphi_{R_2}}$ if $v_x\pa{\psi_{R_1}\land \varphi_{R_2}}=\mathsf{t}$ and $R_1\cap R_2 =C_{\psi_{R_1}\land\varphi_{R_2}}$ if $v_x\pa{\psi_{R_1}\land \varphi_{R_2}}=\mathsf{b}$,
\end{proof}

Now we can define the join probability vectors just by applying the measure to the sets in each row of Table~\ref{table:36-equations} in a suitable probability model for $\letkp$.
Thus, recalling Proposition~\ref{prop:completez-6}, take  $X$ and $\mathsf{M}_\mu$ such that  $\widehat{p}_{\mu}=\widehat{p}$, and let $R$ be an entry in Table~\ref{table:36-equations}. Then, $\mu(R) \in \{\mu(B^\circ_\alpha),\mu(B_\alpha),\mu(C_\alpha)\}$ for some formula $\alpha$. According to Definition~\ref{def:funciones-probabilidad-inducidas}(3), $\mu(R)$ corresponds to the first, second or third entry of $\widehat{p}_{\mu}(\alpha) = \widehat{p}(\alpha)$, depending on whether the region is $B^\circ$, $B$, or $C$, respectively. That is, $\mu(R) \in \{\widehat{p}_i(\alpha) : i=1,2,3\}$. This allows us to obtain the joint probability vectors. For example, the first join probability vector corresponds to the measure of regions 1,2,3,4,5 and 6, i.e., fixing $B^\circ_\psi$ and consider all the possibilities for $\varphi$, according to the following definition.

\begin{definition}
\textbf{(Vectors of joint probabilities)} 
Let $\widehat{p}:\forsigma \longrightarrow [0,1]^6$ be a 6-valued probability function for $\letkp$ (either semantically induced or syntactically defined), and let $\varphi, \psi \in \forsigma$, then the vectors of joint probabilities $jp_i(\psi,\varphi)$ for $1\leq i\leq 6$ are defined as:\\
\begin{itemize}
\item  $jp_1(\psi,\varphi)=\pa{
\widehat{p}_1\pa{\psi_{B^\circ}\land \varphi_{B^\circ}},
\widehat{p}_2\pa{\psi_{B^\circ}\land \varphi_{B}},
\widehat{p}_3\pa{\psi_{B^\circ}\land \varphi_{C}},
\widehat{p}_2\pa{\psi_{B^\circ}\land \varphi_{D}},
\widehat{p}_1\pa{\psi_{B^\circ}\land \varphi_{D^\circ}},
\widehat{p}_2\pa{\psi_{B^\circ}\land \varphi_{U}}
}$
    
\item $jp_2(\psi,\varphi)=\pa{
\widehat{p}_2\pa{\psi_{B}\land \varphi_{B^\circ}},
\widehat{p}_2\pa{\psi_{B}\land \varphi_{B}},
\widehat{p}_3\pa{\psi_{B}\land \varphi_{C}},
\widehat{p}_2\pa{\psi_{B}\land \varphi_{D}},
\widehat{p}_2\pa{\psi_{B}\land \varphi_{D^\circ}},
\widehat{p}_2\pa{\psi_{B}\land \varphi_{U}}
}$

\item $jp_3(\psi,\varphi)=\pa{
\widehat{p}_3\pa{\psi_{C}\land \varphi_{B^\circ}},
\widehat{p}_3\pa{\psi_{C}\land \varphi_{B}},
\widehat{p}_3\pa{\psi_{C}\land \varphi_{C}},
\widehat{p}_3\pa{\psi_{C}\land \varphi_{D}},
\widehat{p}_3\pa{\psi_{C}\land \varphi_{D^\circ}},
\widehat{p}_3\pa{\psi_{C}\land \varphi_{U}}
}$

\item $jp_4(\psi,\varphi)=\pa{
\widehat{p}_2\pa{\psi_{D}\land \varphi_{B^\circ}},
\widehat{p}_2\pa{\psi_{D}\land \varphi_{B}},
\widehat{p}_3\pa{\psi_{D}\land \varphi_{C}},
\widehat{p}_2\pa{\psi_{D}\land \varphi_{D}},
\widehat{p}_2\pa{\psi_{D}\land \varphi_{D^\circ}},
\widehat{p}_2\pa{\psi_{D}\land \varphi_{U}}
}$

\item $jp_5(\psi,\varphi)=\pa{
\widehat{p}_1\pa{\psi_{D^\circ}\land \varphi_{B^\circ}},
\widehat{p}_2\pa{\psi_{D^\circ}\land \varphi_{B}},
\widehat{p}_3\pa{\psi_{D^\circ}\land \varphi_{C}},
\widehat{p}_2\pa{\psi_{D^\circ}\land \varphi_{D}},
\widehat{p}_1\pa{\psi_{D^\circ}\land \varphi_{D^\circ}},
\widehat{p}_2\pa{\psi_{D^\circ}\land \varphi_{U}}
}$

\item $jp_6(\psi,\varphi)=\pa{
\widehat{p}_2\pa{\psi_{U}\land \varphi_{B^\circ}},
\widehat{p}_2\pa{\psi_{U}\land \varphi_{B}},
\widehat{p}_3\pa{\psi_{U}\land \varphi_{C}},
\widehat{p}_2\pa{\psi_{U}\land \varphi_{D}},
\widehat{p}_2\pa{\psi_{U}\land \varphi_{D^\circ}},
\widehat{p}_2\pa{\psi_{U}\land \varphi_{U}}
}$.
\end{itemize}
\end{definition}

\begin{definition}
\textbf{(Syntactic 6-valued Jeffrey update)} Let $\widehat{p}:\forsigma \longrightarrow [0,1]^6$ be a 6-valued probability function for $\letkp$, let $\varphi\in \forsigma$, with $\widehat{p}(\varphi) = (b^\circ_\varphi, b_\varphi, c_\varphi, d_\varphi, d^\circ_\varphi, u_\varphi)$, and let $\vec{  \uplambda}=(b^\circ,b,c,d,d^\circ,u)$ in $[0,1]^6$ be admissible for $\varphi$ given $\widehat{p}$. Given  the vector of contraction-expansion factors $\vec{  \uplambda}\oslash\widehat{p}(\varphi)$ and the vectors of joint probabilities $jp_i(\psi,\varphi)$ for $1\leq i\leq 6$.  The syntactic 6-valued Jeffrey update is a new 6-valued probability function for $\letkp$, $\widehat{p}^{\ \varphi,\vec{  \uplambda}}:\forsigma \longrightarrow [0,1]^6$, whose components are named as: 
$$\widehat{p}^{\ \varphi,\vec{  \uplambda}}\pa{\psi}=
\pa{{\widehat{p}_1\kern 0.01em}^{\varphi,\vec{  \uplambda}}(\psi), 
{\widehat{p}_2\kern 0.01em}^{\varphi,\vec{  \uplambda}}(\psi), 
{\widehat{p}_3\kern 0.01em}^{\varphi,\vec{  \uplambda}}(\psi), 
{\widehat{p}_4\kern 0.01em}^{\varphi,\vec{  \uplambda}}(\psi), 
{\widehat{p}_5\kern 0.01em}^{\varphi,\vec{  \uplambda}}(\psi), 
{\widehat{p}_6\kern 0.01em}^{\varphi,\vec{  \uplambda}}(\psi)}$$ and are given by $ {\widehat{p}_i\kern 0.01em}^{\varphi,\vec{  \uplambda}}(\psi)=jp_i(\psi,\varphi)\boldsymbol{\cdot} \vec{  \uplambda}\oslash\widehat{p}(\varphi).$
\end{definition}

Both definitions coincide, as in the case of the single-valued update. That is, we obtain the same function if, given a probability model, we first obtain the semantic 6-valued Jeffrey update and then induce the 6-valued probability function for the new model, or if, conversely, starting from the original probability model, we first induce the 6-valued probability function and then perform the 6-valued syntactic Jeffrey update. 

\begin{proposition}\label{prop:semantic=sintactic-6}
Let $\mathsf{M}_\mu=\langle X, \mu,  v \rangle$ be a probability model for $\letkp$ based on $\mathsf{M}$ and $\mu$, let $\varphi\in \forsigma$ with $\widehat{p}_\mu(\varphi) = (b^\circ_\varphi, b_\varphi, c_\varphi, d_\varphi, d^\circ_\varphi, u_\varphi)$, and let $\vec{  \uplambda}=(b^\circ,b,c,d,d^\circ,u)\in[0,1]^6$ be admissible for $\varphi$ given $\pmuhat$, then  $\widehat{p}_{\mu^{\varphi,\vec{  \uplambda}}}={\pmuhat \kern 0.01em}^{\varphi,\vec{  \uplambda}}$
\end{proposition} 

\begin{proof}
To prove the equality it is necessary to prove that each of the six components of the functions agree. Let us check the case of the first components denoted respectively by $\widehat{p}_{\mu^{\varphi,\vec{  \uplambda}}1}$ and ${\widehat{p}_{\mu 1}\kern 0.01em}^{\varphi,\vec{  \uplambda}}$, the proof of the rest of them is analogous. The strategy is the same used in  the proof of Proposition~\ref{prop:semantic=sintactic-1} but here dividing the corresponding set into six disjoint pieces and using the results in Proposition~\ref{prop:region-intersection}.\\

$\begin{array}{ll}
\widehat{p}_{\mu^{\varphi,\vec{  \uplambda}}1}\pa{\psi} &= \mu^{\varphi,\vec{  \uplambda}}\pa{B^\circ_\psi} =\displaystyle\sum_{x\in B^\circ_\psi}\mu^{\varphi,\vec{\uplambda}}\pa{\{x\}}\\
&=
\displaystyle\sum_{x\in B^\circ_\psi\cap B^\circ_\varphi}\mu^{\varphi,\vec{\uplambda}}\pa{\{x\}}+  
\displaystyle\sum_{x\in B^\circ_\psi\cap B_\varphi}\mu^{\varphi,\vec{\uplambda}}\pa{\{x\}}+ 
\displaystyle\sum_{x\in B^\circ_\psi\cap C_\varphi}\mu^{\varphi,\vec{\uplambda}}\pa{\{x\}}+\\
&\hspace*{0.4cm}
\displaystyle\sum_{x\in B^\circ_\psi\cap D_\varphi}\mu^{\varphi,\vec{\uplambda}}\pa{\{x\}}+
\displaystyle\sum_{x\in B^\circ_\psi\cap D^\circ_\varphi}\mu^{\varphi,\vec{\uplambda}}\pa{\{x\}}+ 
\displaystyle\sum_{x\in B^\circ_\psi\cap U_\varphi}\mu^{\varphi,\vec{\uplambda}}\pa{\{x\}}\\
&=
\displaystyle\sum_{x\in B^\circ_\psi\cap B^\circ_\varphi}\mu\pa{\{x\}}\frac{b^\circ}{\mu\pa{B^\circ_\varphi}}+  
\displaystyle\sum_{x\in B^\circ_\psi\cap B_\varphi}\mu\pa{\{x\}}\frac{b}{\mu\pa{B_\varphi}}+ 
\displaystyle\sum_{x\in B^\circ_\psi\cap C_\varphi}\mu\pa{\{x\}}\frac{c}{\mu\pa{C_\varphi}}+
\\
&\hspace*{0.4cm}
\displaystyle\sum_{x\in B^\circ_\psi\cap D_\varphi}\mu\pa{\{x\}}\frac{d}{\mu\pa{D_\varphi}}+
\displaystyle\sum_{x\in B^\circ_\psi\cap D^\circ_\varphi}\mu\pa{\{x\}}\frac{d^\circ}{\mu\pa{D^\circ_\varphi}}+ 
\displaystyle\sum_{x\in B^\circ_\psi\cap U_\varphi}\mu\pa{\{x\}}\frac{u}{\mu\pa{U_\varphi}}\\
&=
\mu\pa{B^\circ_\psi\cap B^\circ_\varphi}\frac{b^\circ}{\mu\pa{B^\circ_\varphi}}+  
\mu\pa{B^\circ_\psi\cap B_\varphi}\frac{b}{\mu\pa{B_\varphi}}+ 
\mu\pa{B^\circ_\psi\cap C_\varphi}\frac{c}{\mu\pa{C_\varphi}}+
\\
&\hspace*{0.4cm}
\mu\pa{B^\circ_\psi\cap D_\varphi}\frac{d}{\mu\pa{D_\varphi}}+
\mu\pa{B^\circ_\psi\cap D^\circ_\varphi}\frac{d^\circ}{\mu\pa{D^\circ_\varphi}}+ 
\mu\pa{B^\circ_\psi\cap U_\varphi}\frac{u}{\mu\pa{U_\varphi}}\\
&=
\mu\pa{B^\circ_{\psi_{B^\circ}\land \varphi_{B^\circ}}}\frac{b^\circ}{\mu\pa{B^\circ_\varphi}}+  
\mu\pa{B_{\psi_{B^\circ}\land \varphi_{B}}}\frac{b}{\mu\pa{B_\varphi}}+
\mu\pa{ C_{\psi_{B^\circ}\land \varphi_{C}}}\frac{c}{\mu\pa{C_\varphi}}+
\\
&\hspace*{0.4cm}
\mu\pa{B_{\psi_{B^\circ}\land \varphi_{D}}}\frac{d}{\mu\pa{D_\varphi}}+
\mu\pa{ B^\circ_{\psi_{D^\circ}\land \varphi_{B^\circ}}}\frac{d^\circ}{\mu\pa{D^\circ_\varphi}}+ 
\mu\pa{ B_{\psi_{B^\circ}\land \varphi_{U}}}\frac{u}{\mu\pa{U_\varphi}}\\
&=
\widehat{p}_{\mu 1}\pa{\psi_{B^\circ}\land \varphi_{B^\circ}}\frac{b^\circ}{b^\circ_\varphi}+  
\widehat{p}_{\mu 2}\pa{\psi_{B^\circ}\land \varphi_{B}}\frac{b}{b_\varphi}+ 
\widehat{p}_{\mu 3}\pa{\psi_{B^\circ}\land \varphi_{C}}\frac{c}{c_\varphi}+\\
&\hspace*{0.4cm}
\widehat{p}_{\mu 2}\pa{\psi_{B^\circ}\land \varphi_{D}}\frac{d}{d_\varphi}+
\widehat{p}_{\mu 1}\pa{\psi_{B^\circ}\land \varphi_{D^\circ}}\frac{d^\circ}{d^\circ_\varphi}+ 
\widehat{p}_{\mu 2}\pa{\psi_{B^\circ}\land \varphi_{U}}\frac{u}{u_\varphi}\\
&=\pa{
\widehat{p}_{\mu 1}\pa{\psi_{B^\circ}\land \varphi_{B^\circ}},  
\widehat{p}_{\mu 2}\pa{\psi_{B^\circ}\land \varphi_{B}}, 
\widehat{p}_{\mu 3}\pa{\psi_{B^\circ}\land \varphi_{C}},
\widehat{p}_{\mu 2}\pa{\psi_{B^\circ}\land \varphi_{D}},
\widehat{p}_{\mu 1}\pa{\psi_{B^\circ}\land \varphi_{D^\circ}},
\widehat{p}_{\mu 2}\pa{\psi_{B^\circ}\land \varphi_{U}}}\boldsymbol{\cdot}\\
&\hspace*{0.4cm}
\displaystyle\pa{
\frac{b^\circ}{b^\circ_\varphi},
\frac{b}{b_\varphi},
\frac{c}{c_\varphi},
\frac{d}{d_\varphi},
\frac{d^\circ}{d^\circ_\varphi},
\frac{u}{u_\varphi}}
=jp_1(\psi,\varphi)\boldsymbol{\cdot} \vec{  \uplambda}\oslash\pmuhat(\varphi)
={\widehat{p}_{\mu 1}\kern 0.01em}^{\varphi,\vec{\uplambda}}(\psi).\\
\end{array}$

\end{proof}
   
Thanks to Proposition~\ref{prop:semantic=sintactic-6}, from now on we can simply refer to the 6-valued Jeffrey update without regard to whether it was obtained from the semantic or syntactic approach. In particular, we have the following:

\begin{corolario} Given a 6-valued probability $\widehat{p}$, let $\mathsf{M}_\mu=\langle X, \mu,  v \rangle$ be a probability model such that $\widehat{p}=\pmuhat$ (Proposition \ref{prop:completez-6}). Let $\varphi\in \forsigma$, and $\vec{  \uplambda}\in[0,1]^6$ be admissible for $\varphi$. Then, it holds that  ${\widehat{p} \kern 0.15em}^{\varphi,\vec{  \uplambda}}=\widehat{p}_{\mu^{\varphi,\vec{  \uplambda}}}$.
\end{corolario}

\begin{obs}
There is a probability function $\widehat{p}:\forsigma \longrightarrow [0,1]^6$ for $\letkp$, $\varphi,\psi\in \forsigma$ and $\vec{  \uplambda}_1, \vec{  \uplambda}_2\in [0,1]^6$, admissible vectors for the respective formulas and probability functions, such that $\pa{\widehat{p}^{\varphi,\vec{  \uplambda}_1}}^{\psi,\vec{  \uplambda}_2}\neq\pa{\widehat{p}^{\psi,\vec{  \uplambda}_2}}^{\varphi,\vec{  \uplambda}_1}$.  That is, 6-valued Jeffrey's update is not commutative.
\end{obs}

\begin{obs}
Just as in the classical case, we can define 6-valued Bayesian updating as a special instance of Jeffrey updating where the information acquired is extremal. The agent learns the vector:
\begin{itemize}
\item $\vec{  \uplambda_1}=(1,0,0,0,0,0)$, i.e. she acquires full reliable belief in $\varphi$. 
\item $\vec{  \uplambda_2}=(0,1,0,0,0,0)$, i.e. she acquires full unreliable belief in $\varphi$
\item $\vec{  \uplambda_3}=(0,0,1,0,0,0)$, i.e. she acquires full conflict in $\varphi$
\item $\vec{  \uplambda_4}=(0,0,0,1,0,0)$, i.e. she acquires full unreliable disbelief in $\varphi$. 
\item $\vec{  \uplambda_5}=(0,0,0,0,1,0)$, i.e. she acquires full reliable disbelief in $\varphi$
\item $\vec{  \uplambda_6}=(0,0,0,0,0,1)$, i.e. she acquires full uncertainty in $\varphi$. 
\end{itemize}
In either case the definition of 6-valued Jeffrey updating simplifies. For example the acquisition of full unreliable belief in $\varphi$, i.e. an updating with $\vec{  \uplambda_2}=(0,1,0,0,0,0)$, is:
$$\widehat{p}^{\ \varphi,\vec{  \uplambda_2}}\pa{\psi}=
\frac{1}{b_\varphi}\pa{
\widehat{p}_2\pa{\psi_{B^\circ}\land \varphi_{B}},
\widehat{p}_2\pa{\psi_{B}\land \varphi_{B}},
\widehat{p}_3\pa{\psi_{C}\land \varphi_{B}},
\widehat{p}_2\pa{\psi_{D}\land \varphi_{B}},
\widehat{p}_2\pa{\psi_{D^\circ}\land \varphi_{B}},
\widehat{p}_2\pa{\psi_{U}\land \varphi_{B}}}.$$
\end{obs}

\section{Final Remarks} \label{sect:final}

In this paper, we introduced the study of probability functions based on the versatile logic \letkp. This framework allowed us to consider gaps, gluts, and reliability (or classicality) of the events, extending the detailed proposal for \fde\ presented in~\cite{klein2021probabilities}. A distinctive feature of our proposal is the use of twist structure semantics, which gives rise to a natural interpretation of logical probabilities over \letkp\ in terms of the three or six regions associated with each formula by a valuation in such models. The \letkp-probability functions were defined axiomatically and semantically, obtaining soundness and completeness results, as one would expect. Finally, conditional probabilities based on $\letkp$ were also studied. Specifically, both a semantic and a syntactic characterization of Jeffrey's update over \letkp-based probabilities were proposed, showing their equivalence.

For future work, based on ideas introduced in~\cite{HGE} and later explored in~\cite{FlaGodMar11} and~\cite{bilkova2024, bilkova2025two}, we plan to develop a $[0, 1]$-valued modal logic over the infinite-valued \L ukasiewicz logic. In such logic, a fuzzy modality $P$ would have as its intended semantics the \letkp-probability functions introduced in this paper.

In the Introduction, we pointed out the close relationship between the motivations behind \letkp\ and the AGM$\circ$ framework for paraconsistent belief revision based on LFIs. As another line of future research, we plan to extend and adapt AGM$\circ$ to deal with paradefinite belief revision in the context of \letkp. Furthermore, the relationship between Jeffrey's update and belief revision will be investigated in the same setting.


\

\

\noindent
{\large \bf Acknowledgements:} Borja was supported by SECIHTI, Mexico,  Estancias Sabáticas Vinculadas a la Consolidación de Grupos de Investigación (2025/10976546).  Coniglio acknowledges support by an individual research grant from the National Council for Scientific and Technological Development (CNPq, Brazil), grant 309830/2023-0, and by the São Paulo Research Foundation (FAPESP, Brazil), thematic project
Rationality, logic and probability – RatioLog, grant 2020/16353-3.

\bibliographystyle{plain}


\end{document}